\documentclass[12pt]{article}
   
\usepackage{amsmath,amssymb,amsthm}

\usepackage{graphicx}
\usepackage{natbib}
\usepackage{authblk}
\usepackage{xcolor}
\usepackage{tcolorbox}
\usepackage[normalem]{ulem}
\allowdisplaybreaks

\usepackage{float, multirow, booktabs, threeparttable, diagbox}
\usepackage{bm}

\def\bE{\mathbf{E}}

\def\cX{\mathcal{X}}
\def\cY{\mathcal{Y}}

\def\Var{\mathrm{Var}}

\def\logit{\mathrm{logit}}

\def\Ex{{\bE}}

\def\exp{\mathrm{exp}}

\newcommand{\ben}{\begin{enumerate}}
\newcommand{\een}{\end{enumerate}}
\newcommand{\ee}{\end{equation}}
\newcommand{\bas}{\begin{eqnarray*}}
\newcommand{\eas}{\end{eqnarray*}}
\newcommand{\ba}{\begin{eqnarray}}
\newcommand{\ea}{\end{eqnarray}}
\newcommand{\bit}{\begin{itemize}}
\newcommand{\eit}{\end{itemize}}

\newtheorem{theorem}{Theorem}
\newtheorem{lemma}{Lemma}

\newtheorem{prop}{Proposition}

\newcommand{\convergeto}{\overset{d}{\longrightarrow}}
\newcommand{\convergepto}{\overset{p}{\longrightarrow}}

\newsavebox{\coloredquotationbox}

\title{\LARGE\bf Transfer Learning under Group-Label Shift: A Semiparametric Exponential Tilting Approach}

\author{Manli Cheng, Subha Maity, Qinglong Tian, Pengfei Li}
\affil{Department of Statistics and Actuarial Science\\
University of Waterloo}
\date{}
\begin{document}

\maketitle

\bigskip

\begin{abstract}
We propose a new framework for binary classification in transfer learning settings where both covariate and label distributions may shift between source and target domains. Unlike traditional covariate shift or label shift assumptions, we introduce a group-label shift assumption that accommodates subpopulation imbalance and mitigates spurious correlations, thereby improving robustness to real-world distributional changes. To model the joint distribution difference, we adopt a flexible exponential tilting formulation and establish mild, verifiable identification conditions via an instrumental variable strategy. We develop a computationally efficient two-step likelihood-based estimation procedure that combines logistic regression for the source outcome model with conditional likelihood estimation using both source and target covariates. We derive consistency and asymptotic normality for the resulting estimators, and extend the theory to receiver operating characteristic curves, the area under the curve, and other target functionals, addressing the nonstandard challenges posed by plug-in classifiers. Simulation studies demonstrate that our method outperforms existing alternatives under subpopulation shift scenarios. A semi-synthetic application using the waterbirds dataset further confirms the proposed method’s ability to transfer information effectively and improve target-domain classification accuracy.    
\end{abstract}

\noindent%




\section{Introduction}
\label{sec: intro}

\subsection{Traditional transfer learning assumptions}

Traditional statistical models typically assume that the training and testing data come from the same distribution. However, this assumption often fails in real-world applications, where the training (source) and testing (target) distributions may differ significantly, making conventional statistical inference unreliable. In response to this challenge, transfer learning has emerged as a powerful framework that allows for distributional differences between the source and target domains.
In transfer learning, the source domain refers to the distribution from which we obtain labeled data to learn knowledge, while the target domain is where we aim to apply that knowledge. 
A common setting, called unsupervised domain adaptation, involves having labeled data from the source domain (i.e., the source dataset) and unlabeled data from the target domain (i.e., the target dataset).
The goal of transfer learning is to perform prediction, as well as parameter estimation and statistical inference, on the target domain by leveraging the information learned from the source domain.

In the unsupervised setting, it is generally impossible to transfer knowledge from the source domain to the target domain without further assumptions, as an arbitrary shift between the source and target distributions might not be identifiable from the observed data distribution.
Thus, one needs to make certain similarity assumptions between the source and target distributions to rule out arbitrary shifts.
To discuss these assumptions, let us denote these two distributions of the source and target populations as $P^{(1)}$ and $P^{(0)}$, respectively, which are distributed over the sample space $\cX \times \cY$ consisting of a feature space $X \in \cX$ and a label space $Y \in \cY$. The two most common assumptions are: (1) covariate shift, that posits $P^{(1)}$ and $P^{(0)}$ differ only in their marginal distributions of $X$ (i.e., \ $P^{(1)}_{Y | X} = P^{(0)}_{Y | X}$ whereas $P^{(1)}_X \neq P^{(0)}_X$) and (2) label shift, that assumes only the marginal distribution of $Y$ is different  (i.e., $P^{(1)}_{X | Y} = P^{(0)}_{X | Y}$ whereas $P^{(1)}_Y \neq P^{(0)}_Y$).
The causal interpretation of covariate shift relies on the assumption that features $X$ cause the label $Y$ (i.e., $X\rightarrow Y$) while label shift assumes that the label $Y$ causes the features $X$ (i.e., $Y\rightarrow X$).
In both cases, the causal mechanism is assumed to be stable across domains such that the conditional distributions remain unchanged.

\subsection{Subpopulation shift and spurious correlation}
\label{sec:subpopulation-shift-spurious-correlation}

While these assumptions make transfer learning possible, they often fall short in real-world situations.
For example, consider dividing the features $X$ into two parts, $X = (X_1^\top, X_2^\top)^\top$, where $X_1$ represents demographic features (e.g., age, race, and sex) and $X_2$ represents symptoms of a disease (e.g., X-ray images, biomarkers, and vital signs).
It would be unreasonable to assume that the disease $Y$ causes demographic features $X_1$, or that the symptoms $X_2$ cause the disease $Y$.
Thus, neither covariate shift nor label shift is applicable, revealing the limitations of traditional shift assumptions in real-world settings.

The example above highlights a specific type of distributional shift known as subpopulation shift or group shift (\citealt{Sagawa2020Distributionally,change2023yang}) in machine learning (ML) research, which has received less attention than covariate shift and label shift.

Subpopulation shift occurs when certain subgroups (e.g., $X_1$ in the previous example) are over- or under-represented in the source domain relative to the target domain.
For instance, demographic distributions can differ significantly between a rural clinic and a general hospital, or between hospitals located in southern states versus coastal regions.

Subpopulation shift may introduce a range of challenges in ML and beyond.
One of the most prominent issues is spurious correlation---a phenomenon where models trained on source data rely on features that happen to correlate with the label in the source domain but fail to generalize to the target domain.
For instance, consider a model trained to distinguish between images of cows and camels. If cows in the training data typically appear on green grass and camels in the desert, the model may learn to rely on the background as a predictive cue. While this strategy may perform well if the test data shares the same background-label association, it fails when encountering target datasets with many images of a cow on sand or a camel on grass. In such cases, the model may misclassify the image due to its reliance on a spurious background-label correlation rather than learning the true distinguishing features of the animals.
Beyond prediction, subpopulation shift also affects parameter estimation and statistical inference, yet these issues remain underexplored. 

\subsection{Group-label shift assumption} 
\label{sec:group-label-shift}
To address the subpopulation shift problem, we impose the following similarity assumption between the source and target domains:
\begin{equation}
\label{eq:dist-assumption}
P^{(0)}(X_2 | Y, X_1) = P^{(1)}(X_2 | Y, X_1),
\end{equation}
where the feature vector \( X = (X_1^\top, X_2^\top)^\top \) is partitioned into \textit{group features} \( X_1 \) and \textit{non-group features} \( X_2 \). We refer to this condition as the \textit{group-label shift} assumption, which implies that, within each subgroup defined by \( X_1 \), the conditional distribution of \( X_2 \) given the label \( Y \) remains invariant across domains. For example, if \( X_1 \) represents age groups—such as minors, adults, and seniors—we may reasonably assume that within each group, the distribution of observed symptoms \( X_2 \) given disease status \( Y \) remains stable between the source and target populations. Importantly, the group-label shift assumption accommodates heterogeneity across groups: minors and seniors, for instance, may exhibit more severe symptoms than adults, reflecting real differences in disease manifestation by age.

Beyond addressing subgroup variability, the group-label shift assumption also provides a principled approach to addressing spurious correlations. By allowing the joint distribution \( P^{(j)}(X_1, Y) \) to differ across domains, the assumption decouples domain-specific associations between the group feature (e.g., grass or sand background) and labels (e.g., cow or camel) that may not generalize.
As a result, any spurious correlation learned from the source domain—such as an over-reliance on age group \( X_1 \) when predicting \( Y \)—does not propagate to the target domain. This mitigates the risk of models learning superficial patterns and improves robustness under distributional shift.

\subsection{Our contributions and overviews}
While most existing studies on transfer learning concentrate on covariate shift \citep{hashemi2018weighted, shrikumar2019calibration} or label shift \citep{lipton2018detecting, azizzadenesheli2019regularized}, our work explores a broader form of distribution shift—subpopulation shift—which subsumes both and is particularly relevant in settings with spurious correlations.
We study binary classification problems for the target population under the group-label shift assumption and propose a semiparametric model that integrates this structure via an exponential tilting formulation.
Our contributions are outlined as follows:
\begin{itemize}
    \item  We introduce a novel distributional shift assumption, namely the group-label shift assumption
 presented in \eqref{eq:dist-assumption}, which encompasses the traditional distributional shift assumption as a special case. Relative to conventional assumptions like label shift, our modeling framework not only provides enhanced flexibility but also directly addresses the issue of spurious correlations, which often hinder classifier performance in the target domain.
    \item We establish a new and easily verifiable identification condition, which is much milder and more practical than those typically used in the literature (e.g., \citealt{maity2022understanding}).
    \item Based on our proposed model, we develop a computationally efficient two-step estimation procedure.
    Notably, while our estimator is theoretically equivalent to that obtained via empirical likelihood, it overcomes the potential computational difficulties typically encountered when profiling the likelihood.
    \item We provide estimators for several target population quantities of interest, such as the mean outcome, and employ the Bayes classifier to predict class labels for unlabeled target data. To evaluate classification performance,   we consider both the Bayes classifier based on true source posterior probabilities and its plug-in counterpart based on estimated posteriors. We construct estimators for key performance metrics, including  {the receiver operating characteristic (ROC) curve at fixed thresholds and the area under the curve (AUC)}, along with their corresponding confidence intervals.
    \item 
    We establish the asymptotic properties of both the proposed parameter estimators and the estimators of the associated target functionals. 
    Despite the challenges posed by discontinuities in the estimated classifier, we overcome these difficulties by leveraging empirical process theory and show that both the estimated ROC curve values at fixed thresholds and AUC are asymptotically normal, with asymptotic variances of sandwich form under both fixed and estimated classifiers.
\end{itemize}
The rest of the paper is organized as follows.  Section~\ref{sec: method} introduces the model and corresponding estimation procedures.
Section~\ref{sec:est-functional} discusses the estimation and inference of general functionals of interest in the target population.
Section~\ref{sec: asymp} provides theoretical properties for the proposed estimators and target functionals introduced in Section~\ref{sec:est-functional}, including performance measures such as the ROC curve and AUC for the induced classifier. 
Section~\ref{sec: simulation} reports simulation results that illustrate the finite-sample performance of the proposed method. Section~\ref{sec: realdata} provides a real data application. We end the paper with concluding remarks given in Section~\ref{sec: disscussion}.
All technical proofs and additional simulation results are collected in the supplementary material.

\section{Methodology}
\label{sec: method}
This section presents the proposed statistical framework for binary classification {(i.e., with label space $\cY=\{0,1\}$)} under group-label shift. We begin by formalizing the problem setup (Section~\ref{sec:problem-setup}),  then introduce a semiparametric model, and investigate the identifiability of model parameters from an instrumental variables perspective (Section~\ref{sec: identification}). We then propose a likelihood-based two-step estimation procedure (Section~\ref{sec: estimation}).

\subsection{Problem setup}
\label{sec:problem-setup}

We formally introduce the transfer learning problem. We have a source 
dataset $\left\{x_i,y_i\right\}_{i=1}^{n_1}$,
consisting of independent and identically distributed (iid) samples drawn from the population $P^{(1)}$ and 
target dataset $\left\{x_j,y_j^*\right\}_{j= n_1 + 1}^{n_1+n_0}$ with   iid samples drawn from the population $P^{(0)}$,
{where $y_i\in \cY=\{0,1\}$ and $y_j^*\in\cY=\{0,1\}$.}
Accordingly, we define $s_i = 1$ if the $i$-th observation is from the source population, and $s_i = 0$ otherwise, for $i = 1, \dots, n_0 + n_1$. 
In the unsupervised transfer learning, the outcomes $y_j^*$ are unobserved for $j = n_1+1, \dots, n_1+n_0$.
We assume further that  
there exist probability density functions $p^{(1)}$ and $p^{(0)}$ corresponding to the measure $P^{(1)}$ and $P^{(0)}$.

As introduced in Equation \eqref{eq:dist-assumption}, we consider the group label shift assumption, which simplifies the ratio between the target and source densities as follows:
\[
\frac{p^{(0)}(x, y)}{p^{(1)}(x, y)} = \frac{p^{(0)}(x_1, y)}{p^{(1)}(x_1, y)}.
\] Our transfer learning methods will utilize this density ratio, so for a suitable transfer learning performance, it's crucial to estimate this density ratio well. 
We use the following exponential tilting model \citep{maity2022understanding} for the density ratio and relate the joint distributions in the source and target domains as:
\begin{equation}
	\label{main-model}
	\frac{p^{(0)}(x_1, y)}{p^{(1)}(x_1, y)} = \exp(\alpha_y + \beta_y^\top x_1), ~~ p^{(0)}(x, y ) = \exp\big( \alpha_y + \beta_{y}^{\top}x_1 \big)p^{(1)}(x, y ),
\end{equation} where 
{$\alpha_y \in \mathbb{R}$  and $\beta_y \in \mathbb{R}^{d}$ are class-specific tilting parameters, and $x_1 \in \mathbb{R}^d$}. This tilting assumption is a generalization of the label shift, as fixing $\beta_y = 0$ leads to the label shift assumption. Furthermore, under this tilting model, the difference between the Bayes decision functions in the source and target domains permits arbitrary interactions between $y$ and $x_1$, given by  
\[
\textstyle \logit\{ P^{(0)}(Y = 1| X =x )\} - \logit\{ P^{(1)}(Y = 1| X = x )\} = \alpha_1 - \alpha_0 + (\beta_1 - \beta_0)^\top x_1, 
\] where $\logit(t) = \log\{t/(1 - t)\}$ for $ 0 < t< 1$ denotes the logistic link function. Note that the difference is zero under covariate shift and constant under label shift. In that sense, the exponential tilting framework provides a broader class of shifts across distributions.

 Applying Bayes' formula to model \eqref{main-model} yields the posterior probability that an observation $(x, y)$ originates from the source population as
\ba
\label{main-eq-iv}
p(s=1|x,y) = \frac{1}{1 + \exp(\alpha^*_{y} + \beta_{y}^\top {x_1})},
\ea
where $\alpha^*_{y} = \alpha_{y} + \log\{p(s=0)/p(s=1)\}$. Modeling this sampling probability $p(s=1 | x, y)$ arises naturally in retrospective sampling scenarios involving distinct populations, as in our setting. Furthermore, this design parallels classic case-control studies, where samples are drawn separately from cases and controls rather than from the overall population \citep{prentice1979logistic,fears1986logistic}, and modeling selection bias in non-random sampling designs \citep{byrne2014retrospective}.   

\subsection{Identification}
\label{sec: identification}
Denote $\theta = (\alpha_0, \beta_0^{\top}, \alpha_1, \beta_1^{\top})^{\top}$ and $p^{(1)}(x)$ and $p^{(0)}(x)$ as the marginal  probability density functions of $X$ in the source and target populations, respectively. Under model~\eqref{main-model}, we find that $p^{(0)}(x)$ and $p^{(1)}(x)$ satisfy 
\ba
\label{dr}
p^{(0)}(x) = \left[\exp(\alpha_0 + \beta_0^{^\top}x_1)\{1 - g(x)\} + 
\exp(\alpha_1 + \beta_1^{^\top}x_1)g(x)\right] p^{(1)}(x),
\ea
where $g(x) = P^{(1)}(Y = 1|X= x)$ denotes the conditional probability of $Y = 1$ given $x$ in the source population. 
The left-hand side $p^{(0)}(x)$ is identifiable from the unlabeled target sample $\{x_i: i = n_1 + 1, \ldots, n_1 + n_0\}$, and both $p^{(1)}(x)$ and $g(x)$ on the right-hand side are identifiable from the labeled source data.
Consequently, the key to  model identifiability lies in $\exp(\alpha_0 + \beta_0^{^\top}x_1)\{1 - g(x)\} + 
\exp(\alpha_1 + \beta_1^{^\top}x_1)g(x).$

The following proposition establishes sufficient conditions for identifying $\theta$. 
The proof is provided in Section~S1 of the supplementary material.

\begin{prop}
\label{identifiability}
Under the exponential tilting model~\eqref{main-model}, the parameter $\theta$ is identifiable if the following conditions hold.  (C1) For any point $(x_1, x_2)$ in the support of $x$, $g(x) = g(x_1,x_2) \in (0,1)$.
 (C2) There exist $d+1$ distinct points $t_1, \ldots, t_{d+1}$ in the support of $x_1$ such that: (a) for each $k = 1, \ldots, d+1$, the function $g(t_k, x_2)$ is not constant in $x_2$; and (b) the following matrix
	\[
	\begin{bmatrix}
		1 & \cdots & 1 \\
		t_1  & \cdots & t_{d+1}
	\end{bmatrix}
	\]
 has full rank.	
\end{prop}

The identification conditions above are relatively mild and practically reasonable. 
Condition (C1) ensures that $g(x)$ lies strictly between 0 and 1 for all $x$, which prevents degenerate cases and guarantees that both outcome categories are observable in the source population. Condition (C2)(a) characterizes $x_2$ as an instrumental variable \citep{wang2014instrumental} by requiring that $g(x)$ varies with $x_2$.  This guarantees that $x_2$ affects the conditional distribution of $Y$ in a way that cannot be explained by $x_1$ alone, thereby enabling identification. Condition (C2)(b) imposes a standard rank condition on the design matrix of $x_1$, which ensures the parameters in $\theta$ are uniquely recoverable. These conditions are substantially milder and more readily verifiable in practice compared to the strong support overlap or anchor variable assumptions commonly adopted in the distribution shift literature (e.g., \citealp{maity2022understanding}).


\subsection{Parameters estimation}
\label{sec: estimation}

Recall that $g(x) = p^{(1)}(y=1|x)$ denotes the source posterior probability.
Based on the labeled source data $\{(x_i, y_i): i = 1, \dots, n_1\}$ and the unlabeled target data $\{x_i: i = n_1 + 1, \dots, n_1 + n_0\}$, the semiparametric log-likelihood under model~\eqref{dr} is 
\ba
\label{likhod-1}  
\ell^{\prime} &=& \sum_{i=1}^{n_1} [ \log\{ p^{(1)}(y_i|x_i)\} + \log\{ p^{(1)}(x_i)\} ] + \sum_{j=1}^{n_0} \log\{ p^{(0)}(x_{n_1+j})\} \notag \\  
&=&  \sum_{i=1}^{n_1}   [ y_i\log\{g(x_i)\} + (1-y_i)\log\{1-g(x_i)\} ]  \notag \\
&&+ \sum_{i=1}^{n_0 + n_1}\log\{ p^{(1)}(x_i)\}  + \sum_{j=1}^{n_0}  \log\{ w(x_{n_1+j};\theta) \} 
\ea
with  $w(x;\theta) = \exp(\alpha_0 + \beta_{0}^{^\top}x_1)\left\{1-g(x)\right\} + 
\exp(\alpha_1 + \beta_{1}^{^\top}x_1)g(x)$. 
The likelihood is semiparametric, as no parametric form is assumed for $p^{(1)}(x)$. To ensure that $p^{(0)}(x)$ is a valid density under model~\eqref{dr}, the normalizing condition $\int w(x; \theta)\, p^{(1)}(x) \, dx = 1$ must hold.
Although semiparametric models offer advantages such as robustness to model mis-specification, they also present significant challenges for parameter estimation. In our case,
direct maximization of the log-likelihood $\ell^{\prime}$ is not feasible due to the presence of infinite-dimensional nuisance functions $g(x)$ and $p^{(1)}(x)$.

To address this computational challenge, we propose a two-step estimation procedure. Leveraging the identities $p^{(1)}(y| x) = p(y|x, s = 1) $  and $  p^{(k)}(x) = p(x| s = k)$ for $ k = 0, 1$, we apply Bayes' rule to re-express the log-likelihood up to a constant independent of $\theta$ as
\ba
\label{loglik-2}
\ell = \sum_{i=1}^{n_1} [ \log\{ p^{(1)}(y_i|x_i)\} + \log\{ p(s=1|x_i)/p(s=1)\} ] + \sum_{j=1}^{n_0} \log\{ p(s=0|x_{n_1+j})/p(s=0)\}. 
\ea
It follows from model~\eqref{dr} that $p(s = 0| x) = \rho^* \cdot w(x; \theta) \cdot p(s = 1| x)$,
which in turn yields
\begin{align}
\label{prsco-x}
p(s = 1| x) = \frac{1}{1 + \rho^* \cdot w(x; \theta)},
\end{align}
where \( \rho^* = p(s = 0)/p(s = 1) \) denotes the sampling ratio between the target and source populations. Replacing \( \rho^* \) with its empirical estimator $n_0 / n_1$ and substituting~\eqref{prsco-x} into~\eqref{loglik-2}, we obtain the following equivalent log-likelihood:
\ba
\ell &=& \sum_{i=1}^{n_1} \left[y_i\log\{g(x_i)\} + (1-y_i)\log\{1-g(x_i)\}\right]  \notag \\
&& \hspace{1cm} + \sum_{j=1}^{n_0}\log\{w(x_{n_1+j};\theta)\}  -\sum_{i=1}^{n_1+n_0}\log\{n_1+n_0\cdot w(x_i;\theta)\}.
\ea
This log-likelihood above and the density ratio function $w(x; \theta)$ both involve the unknown conditional probability function $g(x)$. While $g(x)$ cannot be identified from the marginal covariate distribution alone, the labeled source data $\{(x_i, y_i)\}_{i=1}^{n_1}$ enable us to estimate $g(x)$ using any standard regression method. 
For simplicity and computational tractability, we specify a parametric logistic regression model:
\begin{equation}
\label{eq-g}
g(x; \xi) := \frac{1}{1 + \exp(\xi_0 + \xi_1^{^\top} x)},
\end{equation}
where $\xi = (\xi_0, \xi_1^{^\top})^{^\top}$ denotes the nuisance parameter vector. We henceforth replace $g(x)$ by $g(x; \xi)$, and correspondingly write $w(x; \theta)$ as $w(x; \theta, \xi)$. 

We then propose a two-step procedure to estimate the model parameters $\theta$ and $\xi$ using all available data from both the source and target populations. In the first step, we estimate $\xi$ based solely on the labeled source data by fitting a regression model for the conditional probability $g(x) = p^{(1)}(Y = 1|x)$. Specifically, we maximize the partial log-likelihood corresponding to the first term in $\ell$:
\begin{equation}
\label{est-xi}
\hat{\xi} = \arg\max\limits_{\xi} \sum_{i=1}^{n_1} \left[y_i\log g(x_i;\xi) + (1 - y_i)\log \{1 - g(x_i;\xi)\} \right].
\end{equation}
This step relies only on the source data, where labels  are available.
The unlabeled target data, by contrast, provide no information for $g(x)$ and are thus excluded at this stage. This corresponds to fitting a standard logistic regression model for the source population. 

In the second step, with $\xi$ fixed at $\hat{\xi}$, we estimate $\theta$ by maximizing the part of the log-likelihood corresponding to the sampling indicators:
\ba
\label{est-theta}
\hat{\theta} = \arg\max\limits_{\theta} 
\left[
\sum_{j=1}^{n_0} \log\{w(x_{n_1 + j}; \theta, \hat{\xi})\} 
- \sum_{i=1}^{n_1 + n_0} \log\{n_1 + n_0 \cdot w(x_i; \theta, \hat{\xi})\}
\right].
\ea
This step focuses on estimating the distribution shift parameters by maximizing the conditional likelihood of the sampling indicators $\left\{s_i\right\}_{i=1}^{n_1 + n_0}$ given the covariates $\left\{x_i\right\}_{i=1}^{n_0 + n_1}$.

\section{Functionals of interest in the target population}
\label{sec:est-functional}
Estimators of $(\xi, \theta)$ not only characterize distributional shift between the source and target populations, but also enable estimation and inference for target-specific functionals. We consider two main goals: (i) estimating $\eta=\Ex_0[h(X, Y)]$, the expectation under the target distribution $P^{(0)}$,  for a general function $h$, and (ii) evaluating classification performance.
For the latter, we focus on some standard classification metrics, including the ROC curve and AUC.

We introduce several notations to streamline the presentation of target population estimators. Define the joint weight function $w(x, y; \theta) = y \exp(\alpha_1 + \beta_1^{\top} x_1) + (1 - y) \exp(\alpha_0 + \beta_0^{\top} x_1)$, and the conditional components $w_1(x; \theta, \xi) = \exp(\alpha_1 + \beta_1^{\top} x_1) g(x; \xi)$ and $w_0(x; \theta, \xi) = \exp(\alpha_0 + \beta_0^{\top} x_1) \{1 - g(x; \xi)\}$, so that $w(x; \theta, \xi) = w_1(x; \theta, \xi) + w_0(x; \theta, \xi)$. 
We use $\Ex_0$ to denote expectation under the $P^{(0)}$, and $\Ex_1$ under the $P^{(1)}$.


\subsection{Estimating expectations of general functions}
\label{sec:est-expectation}

Under the exponential tilting model~\eqref{main-model}, the $\eta = \Ex_0\{h(X, Y)\}$ can be rewritten as $\eta = \Ex_1\{h(X, Y) w(X, Y; \theta)\}$, leading to the importance weighted (IW) estimator:
\begin{equation}
	\label{e-ipw}
	\hat{\eta}_{\mathrm{IW}} = \frac{1}{n_1} \sum_{i=1}^{n_1} h(x_i, y_i) w(x_i, y_i; \hat{\theta}).
\end{equation}
Alternatively, using the conditional distribution of $Y|X$ in the target population, $\eta$ admits the equivalent form
\[
\eta = \Ex_0\Big\{ h(X, 1) \cdot \frac{w_1(X; \theta, \xi)}{w(X; \theta, \xi)} + h(X, 0) \cdot \frac{w_0(X; \theta, \xi)}{w(X; \theta, \xi)}\Big\},
\]
which yields the regression-type (REG) estimator:
\begin{equation}
	\label{e-reg}
	\hat{\eta}_{\mathrm{REG}} = \frac{1}{n_0} \sum_{j=1}^{n_0}\Big\{ h(x_{n_1 +j}, 1) \cdot \frac{w_1(x_{n_1 +j}; \hat{\theta}, \hat{\xi})}{w(x_{n_1 +j}; \hat{\theta}, \hat{\xi})} + h(x_j, 0) \cdot \frac{w_0(x_{n_1 +j}; \hat{\theta}, \hat{\xi})}{w(x_{n_1 +j}; \hat{\theta}, \hat{\xi})} \Big\}.
\end{equation}
When $h(x, y) = y$, the target functional reduces to the mean $\mu = \Ex_0(Y)$, i.e., the prevalence of $Y=1$ in the target population, with corresponding estimators:
\begin{equation}
\label{ey-ipwreg}
	\hat{\mu}_{\mathrm{IW}} = \frac{1}{n_1} \sum_{i=1}^{n_1} y_i \cdot w(x_i, y_i; \hat{\theta}) \quad \text{and} \quad
	\hat{\mu}_{\mathrm{REG}} = \frac{1}{n_0} \sum_{j=1}^{n_0} \frac{w_1(x_{n_1 + j}; \hat{\theta}, \hat{\xi})}{w(x_{n_1 + j}; \hat{\theta}, \hat{\xi})}.
\end{equation}

\subsection{Estimating AUC and ROC curve}
\label{sec:auc-roc}

{For binary outcomes, an important inferential target is the classification performance of a score function \( c(x) \) under the target population, where \( c(x) \) assigns higher values to observations more likely to have \( Y = 1 \). Such a score function may represent, for example, the posterior probability \( P^{(0)}(Y = 1 | X = x) \), or a continuous-valued biomarker commonly used in biomedical studies to differentiate diseased and non-diseased individuals \citep{hu2024receiver}. In practice, \( c(x) \) may either be pre-specified or estimated from labeled source data.}

We shall consider two commonly used diagnostic tools for evaluating such performance: (a) the ROC curve {as a function of the threshold $u$}:
\begin{equation}
\label{roc}
ROC(u) = 1 - F_1\{F_0^{-1}(1 - u)\},
\end{equation}
where $F_y$ denotes the cumulative distribution functions of $c(X)$ under $P^{(0)}(\cdot | Y= y)$, $y \in \{0, 1\}$, and (b) the AUC,
\begin{equation}
\label{auc}
AUC = \int_0^1 ROC(u) \, du = \int F_0(u) \, dF_1(u).
\end{equation}


To estimate the ROC curve and AUC under the target distribution, we approximate the  \( F_1 \) and \( F_0 \) using Bayes' rule. Since \( Y \) is unobserved in the target data, we express the class-conditional densities as
\[
p_0(x | Y = 1) = \frac{p_0(Y = 1| x) \cdot p_0(x)}{\Ex_0(Y)}, \quad 
p_0(x | Y = 0) = \frac{p_0(Y = 0|x) \cdot p_0(x)}{1 - \Ex_0(Y)}.
\]
The \( p_0(x) \) is estimated empirically using the target covariates \( \{x_{n_1 +j}\}_{j=1}^{n_0} \), and the posterior probabilities \( p_0(Y = y | x) \) are estimated via Bayes’ rule under model~\eqref{main-model} as
\begin{equation}
    \label{estimate_ptx}
{\widehat{p}}^{(0)}(Y = 1|x) = \frac{w_1(x; \hat{\theta}, \hat{\xi})}{w(x; \hat{\theta}, \hat{\xi})}, \quad
{\widehat{p}}^{(0)}(Y = 0| x) = \frac{w_0(x; \hat{\theta}, \hat{\xi})}{w(x; \hat{\theta}, \hat{\xi})}.
\end{equation}

{If $c(x)$ is prespecified, substituting \eqref{estimate_ptx} into the empirical analogs of \( F_1 \) and \( F_0 \) yields
\begin{align}
\label{emhatF1}
\widehat{F}_1(u) &= \frac{1}{n_0 \hat{\mu}} \sum_{j=1}^{n_0} \widehat{p}^{(0)}(Y = 1 | x_{n_1+j}) \cdot I\{c(x_{n_1+j}) \leq u \}, \\
\label{emhatF0}
\widehat{F}_0(u) &= \frac{1}{n_0 (1 - \hat{\mu})} \sum_{j=1}^{n_0} \widehat{p}^{(0)}(Y = 0 | x_{n_1+j}) \cdot I\{c(x_{n_1+j}) \leq u \},
\end{align}
where $\hat{\mu}=n_0^{-1}\sum_{j=1}^{n_0} \widehat{p}^{(0)}(Y = 1 | x_{n_1+j}) $. 
Then, the ROC curve and AUC are estimated in a plug-in manner via
\begin{align}
\label{eroc}
\widehat{ROC}(u) &= 1 - \widehat{F}_1\left( \widehat{F}_0^{-1}(1 - u) \right), ~~
\widehat{AUC} = \int \widehat{F}_0(u) \, d\widehat{F}_1(u).
\end{align}}

{If  $c(x)$ depends on $({\theta},\xi)$,   the estimated classifier $\widehat{c}(x)$ takes the form $c(x; \hat{\theta},\hat\xi)$.   A common example is the plug-in classifier based on the estimated posterior probability $\widehat{p}^{(0)}(Y = 1| x)$ in~\eqref{estimate_ptx}.
We can estimate $F_1$ and $F_0$  by replacing $c(x_{n_1+j})$ in \eqref{emhatF1} and \eqref{emhatF0} with $\widehat{c}(x_{n_1+j})=c(x_{n_1+j}; \hat{\theta},\hat\xi)$. 
Let $\widetilde F_1(u)$ and $\widetilde F_0(u)$
denote the corresponding estimators of $F_1$ and $F_0$, respectively. The ROC curve and AUC
are then estimated by 
\begin{align}
\label{eroc2}
\widetilde{ROC}(u) &= 1 - \widetilde{F}_1\left( \widetilde{F}_0^{-1}(1 - u) \right), ~~
\widetilde{AUC} = \int \widetilde{F}_0(u) \, d\widetilde{F}_1(u).
\end{align}
}
\section{Asymptotic properties}
\label{sec: asymp}
Let $\xi_0$ and $\theta_0$ denote the true values of $\xi$ and $\theta$, respectively.  Assume the $n_0 / n_1 = \rho \in (0, \infty)$ is a constant in our asymptotic study. 
The following result establishes the joint asymptotic normality of $(\hat{\xi}, \hat{\theta})$ obtained from the proposed two-step estimation procedure.

\begin{theorem}
\label{thm1}
Suppose that models~\eqref{main-model} and~\eqref{eq-g} are correctly specified. Under Conditions 1-6 in Section S2 of the supplementary material, as $N = n_0 + n_1 \to \infty$, we have
\[
\sqrt{N}(\hat{\xi}^{\top}-\xi_0^{\top}, \hat{\theta}^{\top} - \theta_0^{\top})^\top
\convergeto \mathcal{N}\big(0,(1+\rho)\Sigma\big)
\]
where $\Sigma$ is defined in (A.12) of the supplementary material.
\end{theorem}
Based on the estimators defined in~\eqref{e-ipw} and~\eqref{e-reg}, we derive their asymptotic distributions in Theorem~\ref{thm2}.
\begin{theorem}
\label{thm2} 
Under the conditions of Theorem~\ref{thm1}, as $N = n_0 + n_1 \to \infty$,  we have
\[
\sqrt{N}(\hat{\eta}_{\mathrm{IW}}  - \eta_0) \convergeto \mathcal{N}\big(0, (1+\rho) \sigma_{\mathrm{IW}}^2 \big), ~~ \sqrt{N}(\hat{\eta}_{\mathrm{REG}}  - \eta_0) \convergeto \mathcal{N}\big(0, (1+\rho) \sigma_{\mathrm{REG}}^2\big)\,. 
\]
 where $\sigma_{\mathrm{IW}}^2 $ and 
 $\sigma_{\mathrm{REG}}^2$ are defined in (A.13) and (A.14) of the supplementary material.
\end{theorem}


Proofs of Theorems~\ref{thm1} and~\ref{thm2} are deferred to Section~S3 of the supplementary material. Motivated by Section~\ref{sec:auc-roc}, we next examine two representative types of classifiers: one is fully known and fixed, often based on a pre-specified function of covariates; the other is data-driven, typically constructed from estimated quantities such as posterior probabilities. Below, we establish the asymptotic properties of the ROC and AUC estimators under both scenarios.

\paragraph{Fixed classifier:} In the following theorem we study the asymptotic properties of the proposed $\widehat{\mathrm{AUC}}$ and $\widehat{\mathrm{ROC}}(u)$ in \eqref{eroc} for a classifier $c(x)$ that is fully specified and does not depend on estimated parameters.  The proof is provided in Section~S4 of the supplementary material.
\begin{theorem}
\label{thm3}
Under Conditions 1--7 in Section S2 of the supplementary material, as $N = n_0 + n_1 \to \infty$, we have
\[
\sqrt{N}(\widehat{AUC} - {AUC} ) \convergeto \mathcal{N}(0,\sigma_{\mathrm{AUC}}^2), ~~ \sqrt{N}( \widehat{ROC}(u) - ROC(u)) \convergeto \mathcal{N}\big(0,\sigma_{\mathrm{ROC}}^2(u)\big)\,,
\] 
where $\sigma_{\mathrm{AUC}}^2$ and $\sigma_{\mathrm{ROC}}^2(u)$ are defined in (A.36) and (A.44) of the supplementary material, respectively.
\end{theorem}

\paragraph{Estimated classifier:}  
{Analogous to the case with a known classifier,  an asymptotic distribution can be derived for 
$\widetilde{\mathrm{AUC}}$ and $\widetilde{\mathrm{ROC}}(u)$
in \eqref{eroc2}
when the classifier $c(x)$ depends on $(\theta,\xi)$  and is estimated by  $\widehat{c}(x) = c(x; \hat\theta,\hat\xi)$, as stated in the following theorem.
}
\begin{theorem}
    \label{thm4}
    Suppose that Conditions 1–6 and 8 in the Section S2 of the supplementary material hold. As $N= n_0+n_1\convergeto \infty$, we have 
    \[
\sqrt{N}(\widetilde{AUC} - {AUC} ) \convergeto N(0,\widetilde\sigma_{\mathrm{AUC}}^2), ~~ \sqrt{N}( \widetilde{ROC}(u) - ROC(u)) \convergeto N(0,\widetilde \sigma_{\mathrm{ROC}}^2(u))\,,
\] where the $\widetilde\sigma_{\mathrm{AUC}}^2$ and $\widetilde\sigma_{\mathrm{ROC}}^2(u)$ are described in (A.84) and (A.92) of the supplementary material.
\end{theorem}

The proof of Theorem~\ref{thm4} is provided in the supplementary material.
A key technical challenge arises when using an estimated $c(x; \hat\theta,\hat\xi)$, as it
induces non-smooth dependence on estimated parameters through indicator functions, complicating asymptotic analysis. 
{We address this issue by decomposing the empirical process induced by the estimated classifier and establishing uniform stochastic equicontinuity results despite the discontinuity introduced by indicator functions. This allows us to derive the asymptotic normality of the proposed estimators of the AUC and ROC values at fixed thresholds, even under plug-in classification rules \( c(x; \hat\theta,\hat\xi) \).
}

For practical inference, although the asymptotic variance formulas in Theorems~\ref{thm2}--\ref{thm4} can be used to construct plug-in estimators, their practical implementation is often challenging due to the involvement of numerous nuisance terms. For the parameter estimates $(\hat{\xi}, \hat{\theta})$, plug-in estimators are relatively straightforward and perform well. However, for complex functionals such as AUC  and ROC$(u)$, the accumulation of estimation error across intermediate steps often leads to instability in finite samples. 
To circumvent these difficulties, we adopt a nonparametric bootstrap procedure to construct confidence intervals for these quantities in both simulations and real data analysis. This avoids estimating complex variance expressions and offers improved empirical coverage.
\section{Simulation study}\label{sec: simulation}
In this section, we assess the finite-sample performance of the proposed methodology via simulation.
The objectives are threefold: (i) to evaluate the {relative bias (RB)} and mean squared error (MSE) of $\hat\theta$, and coverage of ${\theta}$;
{(ii) to evaluate the classification performance of the estimated posterior probabilities $\widehat{p}^{(0)}(Y = 1 | X)$; 
and (iii)} to examine the performance of estimates for $\Ex_0\{h(X, Y)\}$, AUC, and the ROC curve, as introduced in Section~\ref{sec:est-functional}.

Since both $X$ and $Y$ are observed in the source data, standard techniques (e.g., model selection) can be used to specify $p^{(1)}(Y = 1|X)$. We therefore assume the working model $g(X; \xi)$ is correctly specified.

\subsection{Simulation setup}
The binary variables $(Y, X_1)$ are drawn from multinomial distributions with different probabilities across the source and target domains.

\paragraph{Source data:} In the source distribution, the pair $(Y, X_1)$ follows a multinomial distribution with probabilities $\pi^{(S)} = (0.1, 0.4, 0.4, 0.1)$ corresponding to the values $(Y, X_1) \in \{(0,0), (0,1), (1,0), (1,1)\}$. The $X_2$ consists of four coordinates. The distribution of $X_2 | (Y=y, X_1 = x_1)$ is described below: 
\begin{align*}
        X_{21} &| (y, x_1) \sim \mathcal{N}( \gamma_{1}(x_1 - 1), \sigma_1^2), \quad
    X_{22} | (y, x_1) \sim \mathcal{N}( \gamma_{2}(y - 1), \sigma_2^2),\\
    X_{23} &| (y, x_1) \sim \text{Bernoulli}(1,0.5), \quad \;\,
    X_{24} | (y, x_1) \sim \text{Exponential}(\lambda),
    \end{align*}
    where $(\gamma_1, \gamma_2) = (7, -3)$, $(\sigma_1, \sigma_2) = (2, 2)$, and $\lambda = 1$.

\paragraph{Target data:} The target data distribution is identical to the source one in every aspect, except for the multinomial distribution probabilities $\pi^{(T)} = (0.5, 0.1, 0.1, 0.3)$ for $(Y, X_1)$. 



We consider four combinations of sample sizes with $n_1, n_0 \in \{500, 2000\}$, each replicated 500 times. For brevity, we present results for the balanced cases ($n_1 = n_0$) here, and defer those for unbalanced settings ($n_1 \neq n_0$) to Section~S6 of the supplementary material. In each simulation replication, we construct confidence intervals using 500 nonparametric bootstrap resamples and report the average coverage probability (CP) and average interval length (AL) across the 500 simulations. {For clearer comparison across methods, RB is expressed as a percentage ($\%$), and the mean squared error (MSE) is scaled by 1,000.}
To maintain focus on the primary functional estimands, we present the full evaluation of $\hat{\theta}$ in the supplementary material. These results support the accuracy and efficiency of our two-step estimation procedure under varying sample sizes.

\subsection{Prediction performance}
\label{simu.pred}
To evaluate the classification performance of our method in the target population, we construct binary classifiers based on the estimated posterior probabilities $\widehat{p}^{(0)}(Y = 1 | X)$ in \eqref{estimate_ptx}, using a threshold of 0.5. 
These classifiers are applied to the covariates in the target population, and their predictions are compared with the true responses. 
We consider four different approaches for estimating the posterior probability:
\begin{description}  
    \item[-] \emph{Proposed method:}
    Our method that estimates $p^{(0)}(Y = 1 | X)$ using the two-step estimators $\hat{\xi}$ and $\hat{\theta}$ described in Section~\ref{sec: estimation}.
    \item[-] \emph{Importance weighting method:} An approach widely used for correcting distribution shift in transfer learning \citep{shimodaira2000improving, byrd2019effect}. It reweights source samples to approximate the target distribution better. We compute weights as $w(x, y;\hat{\theta}) = \exp(\hat{\alpha}_y + \hat{\beta}_y^\top x_1)$ using $(\hat{\alpha}_y, \hat{\beta}_y)$ from the proposed method, and estimate $p^{(0)}(Y = 1|X)$ by minimizing a weighted loss over source data:
    \[
    \min_f \frac{1}{n_1} \sum_{i=1}^{n_1} w(x_i, y_i;\hat{\theta})\, \mathcal{L}(f(x_i), y_i) + \Omega(f),
    \]
    where $f$ approximates the target posterior probability and  $\Omega(f)$ denotes a possible regularization penalty \citep{huang2006correcting,zhuang2020comprehensive}. We refer to it as \emph{Reweight method}. 
    \item[-]\emph{Naive method:} A baseline method that neglects the distribution shift and directly uses the estimated conditional probability from the source population, i.e., using $\widehat{p}^{(1)}(Y = 1|X)$ to estimate ${p}^{(0)}(Y = 1|X)$. 
    \item[-]\emph{Oracle method:} An idealized benchmark that estimates $p^{(0)}(Y = 1 | X)$ using the observed responses in the target population. Although not implementable in practice, it serves as a performance upper bound for evaluating other methods.
\end{description}

We evaluate each method in terms of recall, accuracy, and precision, averaged over 500 simulation replications. For each metric, we report the empirical mean, {RB, and MSE } in Table~\ref{tab:summary-simu}.

\begin{table}[!htt]
  \centering
  \setlength{\tabcolsep}{2pt} 
  \caption{Simulation results on classification performance of four methods. RB: relative bias ($\%$); MSE: mean squared error ($\times 1000$).}
    \begin{tabular}{lccccccccccccccc}
   		\toprule 
   		&  {mean} &  {RB} &  {MSE} &  {mean} &  {RB} &  {MSE} &  {mean} &  {RB} &  {MSE}  \\
   		\midrule 
   		 & \multicolumn{9}{c}{$n_1 = 500, n_0 = 500$} \\
         \cmidrule(lr){2-10}  
          & \multicolumn{3}{c}{Recall} & \multicolumn{3}{c}{Accuracy} & \multicolumn{3}{c}{Precision} \\
          \cmidrule(lr){2-4} \cmidrule(lr){5-7} \cmidrule(lr){8-10}
    Proposed & 0.785 & -0.791 & 2.856 & 0.842 & -0.919 & 0.394 & 0.817 & -1.189 & 1.711 \\
    Reweight & 0.785 & -0.779 & 2.960  & 0.841 & -1.045 & 0.426 & 0.815 & -1.411 & 1.949 \\
    Naive & 0.561 & -29.117 & 56.215 & 0.535 & -37.032 & 100.79 & 0.438 & -47.057 & 153.252 \\
    Oracle & 0.795 & 0.390  & 0.799 & 0.853 & 0.290  & 0.282 & 0.829 & 0.331 & 0.442 \\
    \midrule
          & \multicolumn{9}{c}{$n_1 = 2000, n_0 = 2000$} \\
          \cmidrule(lr){2-10}
          & \multicolumn{3}{c}{Recall} & \multicolumn{3}{c}{Accuracy} & \multicolumn{3}{c}{Precision} \\
          \cmidrule(lr){2-4} \cmidrule(lr){5-7} \cmidrule(lr){8-10}
    Proposed & 0.789 & -0.329 & 0.806 & 0.848 & -0.221 & 0.082 & 0.825 & -0.255 & 0.395 \\
    Reweight & 0.790  & -0.242 & 0.806 & 0.848 & -0.247 & 0.081 & 0.824 & -0.366 & 0.405 \\
    Naive & 0.562 & -29.038 & 53.592 & 0.535 & -37.103 & 99.93 & 0.436 & -47.236 & 152.969 \\
    Oracle & 0.792 & 0.071 & 0.213 & 0.851 & 0.104 & 0.072 & 0.828 & 0.107 & 0.122 \\
    \bottomrule
   	\end{tabular}%
  \label{tab:summary-simu}%
\end{table}%

{Across both sample sizes, the \emph{Proposed} method generally achieves the smallest RB among all practically feasible methods, except in Recall, where the \emph{Reweight} method performs slightly better. However, in terms of MSE, the \emph{Proposed} method is almost never outperformed by \emph{Reweight}, and its advantage is more pronounced in smaller samples. In contrast, the \emph{Naive}
method shows substantially larger RB and MSE across all metrics, highlighting how ignoring distributional shifts can severely degrade classification performance. Overall, these results underscore the practical value of our method for constructing accurate classifiers under distribution shift, even without labeled target data.
}


\subsection{Estimation and inference for the target mean}\label{sec: simu-estinf-mu}
Table~\ref{tab:ey-simu} reports simulation results for estimating the target mean $\mu=\Ex_0(Y)$ using both the IW and REG estimates described in Section~\ref{sec:est-expectation}. The IW estimate, which requires an estimate of $\theta$, is applicable only under the Proposed method.
The REG estimate, which averages the estimated posterior probabilities, is applied to all four methods in Section \ref{simu.pred} using their respective estimates. In this table, ``IW'' denotes the importance weighting estimate under the Proposed method, and the remaining rows correspond to REG estimates applied to each of the four methods.

{Table \ref{tab:ey-simu} shows that the \emph{Oracle} method provides the most accurate results, with negligible RB, coverage close to nominal, and the shortest intervals. Among implementable methods, the \emph{Proposed} method achieves the lowest MSE and shortest interval lengths across both sample sizes, while maintaining near-nominal coverage. Its RB is the smallest in the smaller sample and remains comparable in the larger sample. In contrast, the \emph{Naive} method performs poorly on all metrics, with substantial RB, high MSE, and severe undercoverage from ignoring distributional shift. As sample size increases, coverage for all feasible methods approaches the nominal level, while interval lengths decrease by about 50\% and MSE is further reduced. }

\begin{table}[!htt]
   \centering
  \setlength{\tabcolsep}{4pt} 
  \caption{Simulation results for estimating $\mu = \Ex_0(Y)$ (true value 0.4). RB: relative bias ($\%$); MSE: mean square error ($\times 1000$).}
   	\begin{tabular}{lcccccccc}
   	\toprule
   	&  {RB} &  {MSE} &  {CP} &  {AL} &  {RB} &  {MSE} &  {CP} &  {AL} \\
   	\midrule 
   	& \multicolumn{4}{c}{$n_1=500,n_0=500$} & \multicolumn{4}{c}{$n_1=2000,n_0=2000$} \\
   	\cmidrule(lr){2-5} \cmidrule(lr){6-9}
   	IW    & 0.584  & 1.884  & 94.60\% & 0.169  & -0.061  & 0.435  & 96.00\% & 0.083  \\
    Proposed & 0.472  & 1.840  & 93.80\% & 0.164  & -0.074  & 0.424  & 95.20\% & 0.081  \\
    Reweight & 0.649  & 1.998  & 93.80\% & 0.170  & 0.122  & 0.440  & 95.00\% & 0.084  \\
    Naive & 27.173  & 12.413  & 1.20\% & 0.103  & 27.439  & 12.204  & 0.00\% & 0.051  \\
    Oracle & 0.022  & 0.476  & 94.40\% & 0.085  & -0.144  & 0.120  & 94.60\% & 0.042  \\
   	\bottomrule
   \end{tabular}%
  \label{tab:ey-simu}%
\end{table}%

\subsection{Estimation and inference for AUC and ROC}
\label{sec: simu-estinf-aucroc}
We evaluate the finite-sample performance of the proposed AUC and ROC estimates under two classifier settings: a fixed classifier $c(x) = p^{(1)}(Y = 1 |x)$, which is known in simulation, and an estimated classifier $\widehat{c}(x) = \widehat{p}^{(0)}(Y = 1|x)$. 
Tables~\ref{tab:estedAUC-simu} and~\ref{tab:estedROC-simu} summarize the results under both sample size configurations.

\begin{table}[!htt]
   \centering
  \setlength{\tabcolsep}{2pt} 
  \caption{Simulation results for AUC estimation of two classifiers.  RB: relative bias ($\%$); MSE: mean square error ($\times 1000$). } %
   	\begin{tabular}{clcccccccc}
   	\toprule
   	&       &  {RB} &  {MSE} &  { CP} &  {AL} &  {RB} &  {MSE} &  { CP} &  {AL} \\
   	\midrule
   	&       & \multicolumn{4}{c}{$n_1=500,n_0=500$} & \multicolumn{4}{c}{$n_1=2000,n_0=2000$} \\
   	\cmidrule(lr){3-6} \cmidrule(lr){7-10}
   	\multirow{4}[0]{*}{$c(x)$}
   &Proposed & 1.996  & 3.298  & 92.00\% & 0.227  & 0.875  & 0.841  & 93.60\% & 0.115  \\
    &Reweight & 2.643  & 3.453  & 91.20\% & 0.224  & 1.072  & 0.889  & 93.60\% & 0.117  \\
    &Naive & 46.775  & 73.407  & 0.00\% & 0.064  & 46.804  & 73.287  & 0.00\% & 0.031  \\
    &Oracle & -0.079  & 0.613  & 96.00\% & 0.101  & 0.182  & 0.157  & 95.20\% & 0.050  \\    
   	\midrule
   	&       & \multicolumn{4}{c}{$n_1=500,n_0=500$} & \multicolumn{4}{c}{$n_1=2000,n_0=2000$} \\
   	\cmidrule(lr){3-6} \cmidrule(lr){7-10}
   	\multirow{4}[0]{*}{$\widehat{c}(x)$} 
    &Proposed & -0.011  & 0.213  & 97.40\% & 0.064  & -0.090  & 0.071  & 95.60\% & 0.032  \\
    &Reweight & 0.204  & 0.250  & 97.20\% & 0.068  & -0.043  & 0.078  & 95.00\% & 0.034  \\
    &Naive & -7.372  & 4.883  & 1.80\% & 0.062  & -7.756  & 5.168  & 0.00\% & 0.031  \\
    &Oracle & 0.172  & 0.162  & 91.00\% & 0.046  & 0.078  & 0.037  & 94.00\% & 0.023  \\
   	\bottomrule
   \end{tabular}%
  \label{tab:estedAUC-simu}%
\end{table}

From Table~\ref{tab:estedAUC-simu}, the \emph{Proposed} method consistently achieves accurate AUC estimation under both classifier settings, exhibiting small RB and MSE. 
Under the estimated classifier $\widehat{c}(X)$, it attains near-zero relative bias and the shortest AL among all practically implementable methods. Coverage behavior varies across classifier types. 
Under the smaller sample size, with the estimated classifier $\widehat{c}(x)$, the  \emph{Proposed} and  \emph{Reweight} methods exhibit slight overcoverage (CP = 97.4\% and 97.2\%, respectively), while the  \emph{Oracle} method undercovers (CP = 91.0\%). In contrast, under the fixed classifier $c(x)$, the  \emph{Proposed} and  \emph{Reweight} methods exhibit mild undercoverage (CP = 92.0\% and 91.2\%, respectively).
The  \emph{Reweight} method performs comparably to the  \emph{Proposed} approach, with marginally larger RB and slightly wider intervals. In contrast, the  \emph{Naive} method yields substantial RB and severely under-covered intervals, highlighting the importance of correcting for distributional shift. As expected, increasing the sample size improves estimation accuracy across all metrics, with notable reductions in RB, MSE, and AL.

\begin{table}[!htt]
  \centering
  \setlength{\tabcolsep}{1.5pt} 
  \caption{Simulation results for ROC estimation at thresholds $u = 0.1$ and $0.2$ of two classifiers. RB: relative bias ($\%$); MSE: mean square error ($\times 1000$).}
   \begin{tabular}{cclcccccccc}
   	\toprule
   	& Threshold	&       &  {RB} &  {MSE} &  {CP} &  {AL} &  {RB} &  {MSE} &  {CP} &  {AL} \\
   	\midrule
   	&  & & \multicolumn{4}{c}{$n_1=500,n_0=500$} & \multicolumn{4}{c}{$n_1=2000,n_0=2000$} \\
     \cmidrule(lr){4-7}  \cmidrule(lr){8-11}
    \multirow{8}[0]{*}{$c(x)$} & 
    \multirow{4}[0]{*}{0.1} 
    &Proposed & 6.468  & 3.570  & 94.20\% & 0.242  & 2.494  & 0.878  & 95.00\% & 0.119  \\
    &&Reweight & 8.709  & 3.906  & 92.40\% & 0.246  & 3.272  & 0.950  & 94.40\% & 0.122  \\
    &&Naive & 177.296  & 133.026  & 0.00\% & 0.154  & 177.530  & 132.102  & 0.00\% & 0.075  \\
    &&Oracle & 0.205  & 0.763  & 95.00\% & 0.115  & 0.740  & 0.189  & 96.60\% & 0.057  \\
   	\cmidrule{2-11}    
   	&\multirow{4}[0]{*}{0.2} 
    &Proposed & 5.364  & 5.273  & 93.60\% & 0.292  & 2.222  & 1.327  & 94.60\% & 0.145  \\
    &&Reweight & 6.816  & 5.547  & 92.40\% & 0.292  & 2.730  & 1.405  & 94.00\% & 0.148  \\
    &&Naive & 137.707  & 177.618  & 0.00\% & 0.130  & 138.087  & 177.701  & 0.00\% & 0.063  \\
    &&Oracle & 0.259  & 1.175  & 96.20\% & 0.138  & 0.528  & 0.287  & 96.00\% & 0.069  \\
   	\midrule 
   	& &  & \multicolumn{4}{c}{$n_1=500,n_0=500$} & \multicolumn{4}{c}{$n_1=2000,n_0=2000$} \\
    \cmidrule(lr){4-7}  \cmidrule(lr){8-11}
    \multirow{8}[0]{*}{$\widehat{c}(x)$} & 
   	\multirow{4}[0]{*}{0.1} 
    &Proposed & -0.234  & 2.342  & 97.40\% & 0.208  & -0.397  & 0.757  & 95.40\% & 0.106  \\
    &&Reweight & 0.427  & 2.523  & 97.60\% & 0.214  & -0.255  & 0.790  & 95.60\% & 0.109  \\
    &&Naive & -25.348  & 40.185  & 0.20\% & 0.154  & -26.391  & 42.134  & 0.00\% & 0.075  \\
    &&Oracle & 0.422  & 1.360  & 91.80\% & 0.133  & 0.155  & 0.302  & 94.80\% & 0.067  \\
   	\cmidrule{2-11}
   	&\multirow{4}[0]{*}{0.2} 
   &Proposed & -0.192  & 0.894  & 97.00\% & 0.130  & -0.251  & 0.302  & 95.80\% & 0.066  \\
    &&Reweight & 0.183  & 0.986  & 97.80\% & 0.135  & -0.173  & 0.317  & 95.60\% & 0.069  \\
    &&Naive & -16.495  & 22.216  & 0.60\% & 0.127  & -17.242  & 23.297  & 0.00\% & 0.063  \\
    &&Oracle & 0.213  & 0.567  & 93.00\% & 0.087  & 0.097  & 0.132  & 95.00\% & 0.044  \\
   	\bottomrule
   \end{tabular}%
  \label{tab:estedROC-simu}%
\end{table}%

Estimation results for the ROC curve at thresholds $u = 0.1$ and $u = 0.2$ are summarized in Table~\ref{tab:estedROC-simu}. {Consistent with the AUC findings, for a known classifier ${c}(x)$,  the \emph{Proposed} method achieves the smallest RB, MSE, and interval length, with coverage close to the nominal level, second only to the \emph{Oracle} estimator. 
For the estimated classifier $\widehat{c}(x)$, the \emph{Reweight} method shows slightly smaller RB at certain points, but the \emph{Proposed} estimator retains clear advantages in MSE, interval length, and coverage.}
{The \emph{Naive} method continues to show large RB and severe undercoverage across all settings.}
As expected, increasing the sample size improves estimation accuracy, narrows confidence intervals, and brings CP closer to the nominal level, while preserving the relative ranking of method performance.

Overall, the \emph{Proposed} method offers reliable estimation and inference under distributional shift. By effectively leveraging information from the source population, it consistently delivers accurate results in the target domain across various objectives and sample sizes.
\section{Real data application: waterbirds dataset}
\label{sec: realdata}
To demonstrate the practical utility of the proposed method, we apply it,
along with several competing approaches from the simulation study, to the waterbirds dataset.
This semi-synthetic dataset is constructed by superimposing bird images from the Caltech-UCSD Birds-200-2011 (CUB) dataset \citep{wah2011caltech} onto backgrounds from the Places dataset \citep{zhou2017places}, resulting in two bird categories—\textit{waterbirds} and \textit{landbirds}—crossed with two types of backgrounds: \textit{water} and \textit{land}. Thus, each bird category appears on both types of background.
As a result, true labels are available for all observations, including those in the target domain. This dataset has been widely used to assess model robustness under distributional shifts; see \citet{Sagawa2020Distributionally} for details.

The original training set exhibits a pronounced spurious association between bird type and background: 95\% of waterbirds appear on water backgrounds, and 95\% of landbirds appear on land backgrounds. In contrast, the test set is balanced across all bird-background combinations. This design induces a substantial shift in feature distribution due to altered background association and the proportion of landbirds and waterbirds between the source and target data.
Thus, this dataset likely exhibits both covariate shift and label shift. 
We use a preprocessed version of this dataset provided by \citet{maity2022understanding}, where each image is represented by a 512-dimensional feature vector extracted from the penultimate layer of a ResNet18 model pretrained on ImageNet. These feature vectors serve as covariates in our analysis.
We define the binary response variable $Y$ to indicate bird type, with values corresponding to \textit{waterbirds} and \textit{landbirds}. The background type, denoted as \textit{place}, takes values \textit{land} or \textit{water}, and is treated as a non-instrumental covariate ($X_1$). The remaining 512-dimensional image embeddings are used as instrumental variables ($X_2$). The source data consist of 4,795 observations from the original training split, while the target data include 6,993 observations by combining the validation and test splits.

\subsection{Prediction performances}
Since the waterbirds dataset is semi-synthetic, true labels are available in the target data, enabling direct evaluation of the classification performance for all methods considered  Section \ref{simu.pred}. 
To mitigate overfitting due to high-dimensional covariates, we use penalized logistic regression (implemented via the \texttt{cv.glmnet} function in R) to estimate the conditional probability of the binary outcome $Y$ given covariates $x$. Specifically, the \emph{Proposed} and \emph{Oracle} methods estimate $g(x) = p^{(1)}(Y = 1|x)$ and $p^{(0)}(Y = 1|x)$ using cross-validated logistic regression. In contrast, the reweighting method incorporates observation-specific weights obtained from the reweighting procedure.

Table~\ref{tab:summary-water} presents the classification results on the target data across different methods. As shown in Table~\ref{tab:summary-water}, the \emph{Proposed} method achieves strong performance, with recall (0.711), and accuracy (0.907) values that are comparable to those of the 
\emph{Oracle} method. Although its precision (0.846) is slightly lower than that of the \emph{Reweight} method (0.883), the proposed estimate consistently outperforms reweighting across other key metrics, highlighting its superior ability to mitigate spurious correlations and improve generalization to the target domain.
In contrast, the \emph{Naive} method, which ignores distributional shifts, performs substantially worse across all metrics, particularly in accuracy (0.738) and precision (0.431), illustrating the detrimental impact of neglecting distribution mismatch between the source and target domains.

\begin{table}[!htt]
  \centering
\caption{Classification performance on the waterbirds dataset.}
    	\begin{tabular}{lccc}
    		\toprule
    	& Recall & Accuracy & Precision  \\ 
    	\cmidrule{2-4}
    	Proposed & 0.711  & 0.907  & 0.846  \\ 
    	Reweight & 0.599  & 0.894  & 0.883  \\ 
    	Naive & 0.575  & 0.738  & 0.431  \\ 
    	Oracle & 0.792  & 0.936  & 0.906 \\
    	\bottomrule
    \end{tabular}%
  \label{tab:summary-water}%
\end{table}%

These patterns are further reflected in the ROC curves, which provide a more visual comparison of classification performance on the target data. 
{Using the estimated classifier $\widehat{p}^{(0)}(Y=1|\bm{x})$ and the true labels in the target data, we compute the true positive rate and false positive rate across a sequence of thresholds $\{0.005, 0.01,  \ldots,  0.995\}$ and plot the corresponding empirical ROC curves.} As shown in Figure~\ref{fig:roc-tprfpr}, the \emph{Proposed} method again exhibits the best performance among all non-oracle methods. 

\begin{figure}[!http]
	\centering
\includegraphics[width=10cm, height=10cm ]{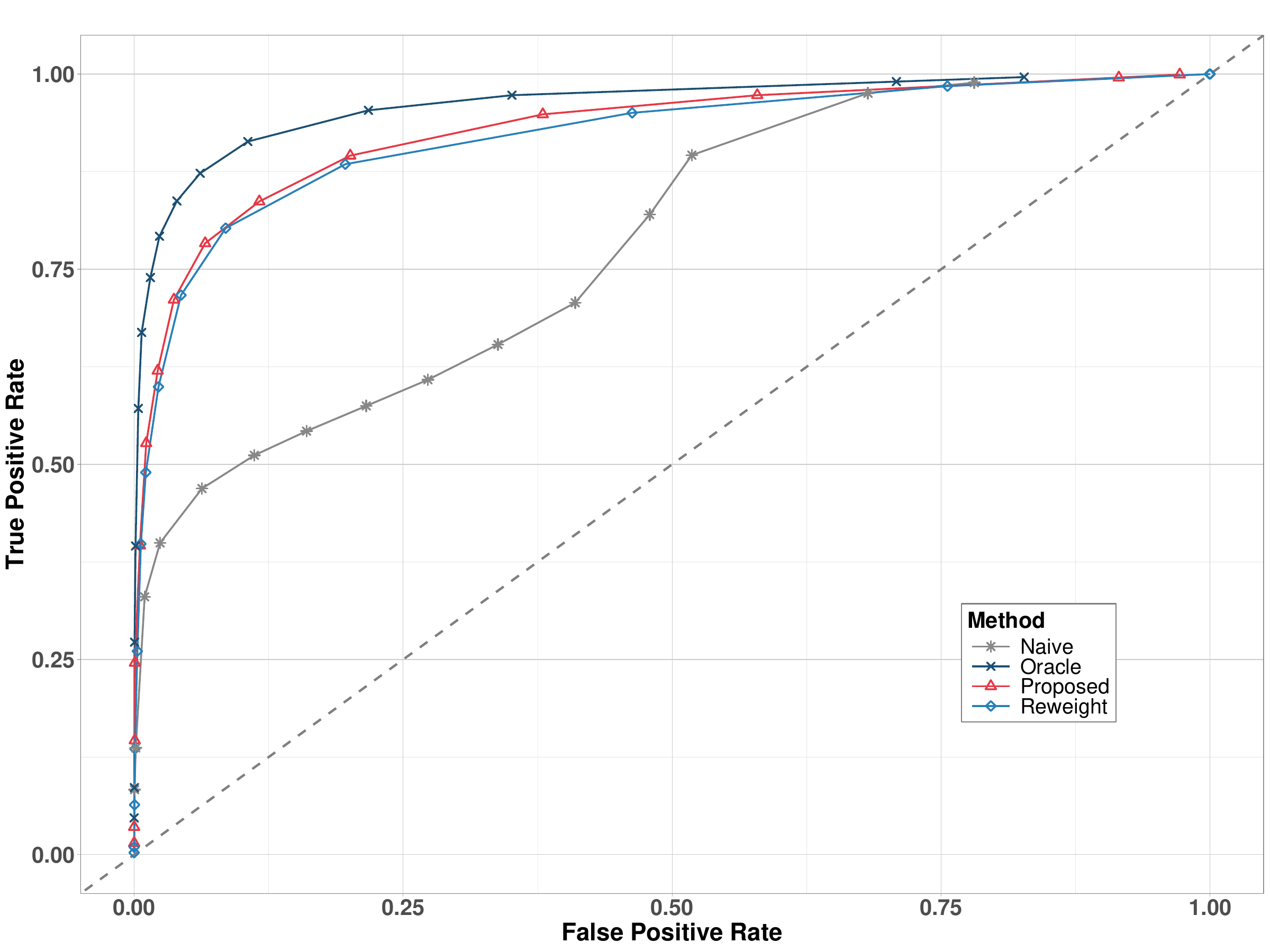} \\
	\caption{Empirical ROC curves of four classifiers for waterbirds dataset.}
	\label{fig:roc-tprfpr}
\end{figure}

\subsection{Estimation and inference for the target mean}
Table~\ref{tab:estey-water} reports point estimates and confidence intervals for $\mu = \Ex_0(Y)$, with 95\% confidence intervals constructed using 500 nonparametric bootstrap replications. The \emph{Oracle} method yields an estimate of 0.222, which exactly matches the sample mean of the target outcomes and has the shortest confidence interval (length: 0.019), serving as a performance benchmark.
Among the non-oracle methods, the IW estimate provides the estimate closest to the oracle (0.212), though it is associated with the widest confidence interval (length: 0.061). The regression-based \emph{Reweight} method also performs well, producing an estimate of 0.230 with the shortest interval (length: 0.034). The proposed regression-based estimate yields a slightly higher estimate (0.241) with a marginally wider interval (length: 0.037), showing reasonable performance but somewhat less accurate than the \emph{Reweight} and \emph{IW} methods.
All three methods—\emph{IW}, \emph{Reweight}, and \emph{Proposed}—produce intervals that successfully cover the sample mean. In contrast, the \emph{Naive} method exhibits substantial upward bias (estimate: 0.316), and its interval fails to include the sample mean, highlighting the adverse effects of ignoring distributional shifts.

\begin{table}[!htt]
	\centering
	\caption{Estimates and confidence intervals (CIs) for $\mu$ using the waterbirds dataset (sample mean 0.222).}
		\begin{tabular}{lccc}
		\toprule 
		& Est    & CI & CI length \\
		\cmidrule{2-4}
	IW   & 0.212  & [0.198,0.260] & 0.061  \\
	Proposed & 0.241  & [0.217,0.254] & 0.037  \\
	Reweight & 0.230  & [0.211,0.245] & 0.034  \\
	Naive & 0.316  & [0.285,0.330] & 0.045  \\
	Oracle & 0.222  & [0.213,0.231] & 0.019  \\
		\bottomrule
	\end{tabular}%
	\label{tab:estey-water}%
\end{table}%

\subsection{Estimation and inference for AUC and ROC}
Following the same procedure described in Section~\ref{sec: simu-estinf-aucroc}, we evaluate the estimated AUC and ROC curve based on the classifier $\widehat{p}^{(0)}(Y=1|x)$ using the waterbirds dataset. Table~\ref{tab:estauc-water} reports the estimated AUC values and the associated 95\% confidence intervals.
Given that all target labels are observed, we treat the empirical AUC as a finite-sample approximation of the true value, which is reported in the “True” column of Table~\ref{tab:estauc-water}.

 \begin{table}[!htt]
	\centering
	\caption{Estimates and confidence intervals (CIs) for AUC using the waterbirds dataset.}
	\begin{tabular}{lcccc}
		\toprule
		& True & Est    & CI & CI length \\
		\cmidrule{2-5}
	Proposed & 0.930  & 0.897  & [0.885,0.967] & 0.082  \\
	Reweight & 0.925  & 0.847  & [0.839,0.959] & 0.120  \\
	Naive & 0.786  & 0.932  & [0.929,0.972] & 0.043  \\
	Oracle & 0.963  & 0.945  & [0.956,0.974] & 0.017  \\
		\bottomrule
	\end{tabular}%
	\label{tab:estauc-water}%
\end{table}%

From Table~\ref{tab:estauc-water}, the \emph{Proposed} method yields an AUC estimate of 0.897, the closest to the empirical AUC among all non-oracle methods, with a moderately short confidence interval of $[0.885, 0.967]$ (length 0.082). In contrast, the \emph{Reweight} method produces a lower estimate of 0.847 and a wider interval of $[0.839, 0.959]$ (length 0.120), indicating greater uncertainty and reduced accuracy. This performance gap may reflect the proposed method’s ability to better mitigate spurious correlations in the waterbirds dataset and to capture features more relevant for target classification. The \emph{Naive} method performs the worst, overestimating the AUC (0.932) with a narrow interval of $[0.929, 0.972]$ (length 0.043) that entirely misses the empirical truth. These results highlight the superiority of the \emph{Proposed} method in estimating AUC under distributional shift.


Figure~\ref{fig:roc-pce-ci} presents the estimated ROC curves with their 95\% bootstrap confidence bands. The empirical ROC curve, treated as the finite-sample ground truth and shown in Figure~\ref{fig:roc-tprfpr}, is covered by all bands except that of the \emph{Naive} method. The \emph{Proposed} method achieves the narrowest confidence band among all non-oracle methods, indicating superior estimation stability. The \emph{Oracle} method’s band does not fully cover its curve, likely due to the tightness of percentile bootstrap intervals.

\begin{figure}[!htt]
	\centering
	\includegraphics[width=16cm, height=12cm ]{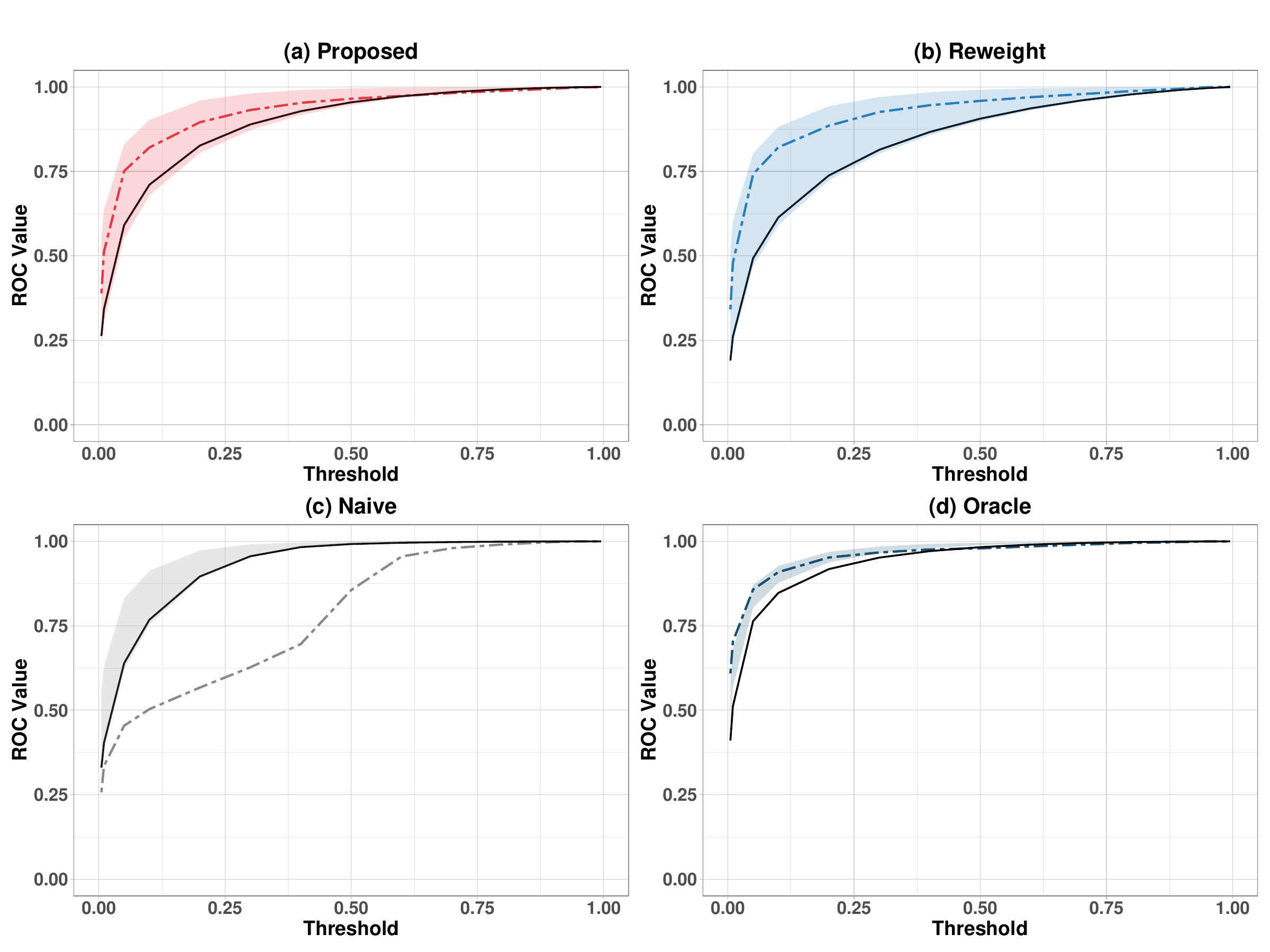} \\
	\caption{Estimated and empirical ROC curves with 95\% confidence bands.
		Solid lines show estimated ROC curves, dashed lines denote empirical ROC curves, and shaded areas represent the confidence bands.}
	\label{fig:roc-pce-ci}
\end{figure}

Overall, results from both the simulation and the real data application suggest that the \emph{Proposed} method offers an effective and reliable strategy for transferring information from the source population to improve classification performance in the target data.


\section{Discussion}
\label{sec: disscussion}

This paper investigates estimation and inference under a general subpopulation shift, leveraging labeled source data to improve estimation in an unlabeled target population. We propose a semiparametric exponential tilting model that incorporates a novel and practically plausible group-label shift assumption, which separates predictive group-level features from non-group features and mitigates spurious cross-domain correlations.
Within this framework, we establish the asymptotic normality of estimators for key target-population functionals, including the target mean, ROC curve, and AUC. Crucially, the proposed modeling setup also addresses a fundamental identification challenge: the group-label structure yields interpretable and verifiable identification conditions—conceptually akin to instrumental variables in missing data and causal inference, yet more directly applicable in practice. These conditions are notably milder and more transparent than those required in recent work such as \cite{maity2022understanding}.
Extensive simulations and a semi-synthetic dataset—designed to induce spurious correlations between covariates and labels—demonstrate that our method outperforms leading importance-weighting approaches in both classification accuracy and estimation stability. Additional results and implementation details are provided in the supplementary material to illustrate the method’s behavior across diverse settings.

While the proposed framework is broadly applicable, it depends on the correct specification of the exponential tilting model and adequate overlap between the source and target populations. Future work may focus on relaxing these conditions or extending the framework to more complex settings.



\section*{Data Availability Statement}

The authors confirm that the data supporting the findings of this study are available within the article and its supplementary materials.

\bibliographystyle{apalike}
\bibliography{reference}

\newpage

\appendix

\section{Proof of Proposition 1 in the main paper}
\label{proof-pro1}
\begin{proof}
Recall that 
$$
w(x;\theta) = \exp(\alpha_0 + \beta_{0}^{^\top}x_1)\left\{1-g(x)\right\} + 
\exp(\alpha_1 + \beta_{1}^{^\top}x_1)g(x),
$$
where $g(x) = P^{(1)}(Y = 1|X= x)$ denotes the conditional probability of $Y = 1$ given $x$ in the source population. 
Since the labeled source data contain both the covariates $x$ and the label $y$, $g(x)$ is identifiable. 
As discussed in the main paper, identifying the parameters  $(\alpha_0, \beta_0)$ and $(\alpha_1, \beta_1)$
in model (2) of the main paper
is equivalent to identifying 
them in $w(x;\theta)$. 

Consider two arbitrary sets of parameters $(\alpha_0,\beta_0, \alpha_1, \beta_1)$  and $(\alpha_0^*,\beta_0^*, \alpha_1^*, \beta_1^*)$ that satisfy $w(x;\theta) = w(x;\theta^*)$, i.e., 
\begin{align}
\label{eqq}
    &\exp(\alpha_0 + \beta_{0}^{\top}x_1)\{1-g(x)\} + 
\exp(\alpha_1 + \beta_{1}^{\top}x_1)g(x) \notag  \\[6pt]
= &\exp(\alpha_0^* + \beta_{0}^{*\top}x_1)\{1-g(x)\} 
+ \exp(\alpha_1^* + \beta_{1}^{*\top}x_1)g(x). 
\end{align}
We now argue that under Conditions (C1) and (C2), Equation \eqref{eqq} implies  
$$
(\alpha_0,\beta_0, \alpha_1, \beta_1) = (\alpha_0^*,\beta_0^*, \alpha_1^*, \beta_1^*).
$$

Let $\gamma_0 = \alpha_1 - \alpha_0$ and $\gamma_1 = \beta_1 - \beta_0$. 
 Similarly, define $\gamma_0^*$ and $\gamma_1^*$  in terms of $\alpha_k^*$ and $\beta_k^*$ for $k = 0, 1$.
Under Condition (C1), i.e., $g(x)\in (0,1)$,  Equation \eqref{eqq}
implies 
\bas
(\alpha_0^* - \alpha_0^*) + 
(\beta_{0} - \beta_{0}^*)^\top  x_1
= 
\log\left\{
\frac{1 + \exp(\gamma_0^* + \gamma_1^{*\top}x_1)\cdot R(x)}{1 + \exp(\gamma_0 + \gamma_1^{\top}x_1)\cdot R(x)}\right\},
\eas
where $R(x) = g(x)/\{1-g(x)\}$. 

For any  given $t_{k}$ in Condition (C2) of the main paper,  by Part (a) of Condition (C2),
there exists $x_2 \neq x_2^*$ in the support of $x_2$ such  that
$R(t_{k},x_2) \neq R(t_ {k},x_2^*)$. Then
\begin{align}
\label{eq-X2}
    (\alpha_0 - \alpha_0^*) + 
(\beta_{0} - \beta_{0}^*)^\top  
t_{k}
&=
\log\left\{
\frac{1 + \exp(\gamma_0^* + \gamma_1^{*\top}t_{k})\cdot R(t_{k},x_2)}{1 + \exp(\gamma_0 + \gamma_1^{\top}t_{k})\cdot R(t_{k},x_2)}\right\}  \notag \\[6pt]
&=
\log\left\{
\frac{1 + \exp(\gamma_0^* + \gamma_1^{*\top}t_{k})\cdot R(t_{k},x_2^*)}{1 + \exp(\gamma_0 + \gamma_1^{\top}t_{k})\cdot R(t_{k},x_2^*)}\right\},
\end{align}
which implies 
$$
\bigl\{\exp(\gamma_0^* + \gamma_1^{*\top}t_{k})
-\exp(\gamma_0 + \gamma_1^{\top}t_{k})\bigr\}\cdot\bigl\{R(t_{k},x_2) - R(t_{k},x_2^*)\bigr\} = 0.
$$
Since $R(t_{k},x_2) \neq R(t_{k},x_2^*)$, it follows that  
$$
\gamma_0^* + \gamma_1^{*\top}t_{k} = \gamma_0 + \gamma_1^{\top}t_{k}.
$$
Therefore, for any $t_{k}, k=1,\ldots,d+1,$
$$
(\gamma_0^* - \gamma_0) + (\gamma_1^* - \gamma_1)^{\top}t_{k} = 0.$$
By Part (b) of Condition (C2),  we conclude that    $\gamma_0^* = \gamma_0$  and $\gamma_1^* = \gamma_1$. 
Together with \eqref{eq-X2}, this implies
$$
 (\alpha_0 - \alpha_0^*) + 
(\beta_{0} - \beta_{0}^*)^\top  
t_{k} = 0. 
$$
Again, by Part (b) of Condition (C2), we obtain $\alpha_0 = \alpha_0^*$ and $\beta_{0} =\beta_{0}^*$.  Combining these with   $\gamma_0^* = \gamma_0$  and $\gamma_1^* = \gamma_1$, it follows that  $\alpha_1 = \alpha_1^*$ and $\beta_{1} =\beta_{1}^*$. This completes the proof. 
\end{proof}

\section{Technical preliminaries for Theorems 1–4}
\label{prelim}
Corresponding to the two estimators \(\hat\eta_{\mathrm{IW}}\) and \(\hat\eta_{\mathrm{REG}}\) introduced in the main paper, we augment the parameter vector and define
$$\omega = \bigl(\xi^{\top},\theta^{\top},\eta ,\eta \bigr)^{\top},\quad \nu = (\xi^{\top}, \theta^{\top})^{\top}.$$

Let
\bas
\tilde{S}(x,y,s;\omega) = \bigl( \tilde{S}_{1}^{\top}(x,y,s;\xi), \tilde{S}_{2}^{\top}(x,s;\theta,\xi),
\tilde{S}_{3}(x,y,s;\theta,\eta ),\tilde{S}_{4}(x,s;\theta,\xi,\eta) \bigr)^{\top},
\eas 
where 
\begin{align*}
&\tilde{S}_{1}(x,y,s;\xi) =\frac{s}{1+\rho} \Bigl\{\frac{y}{g(x;\xi)} - \frac{1-y}{1-g(x;\xi)}\Bigr\} \cdot \nabla_{\xi} g(x;\xi),\\[6pt]
&\tilde{S}_{2}(x,s;\theta,\xi) =
\frac{1}{1+\rho}
\Bigl\{ \frac{(1-s)\rho}{ w(x;\theta,\xi)} - \frac{(1-s)\rho^2}{1+\rho w(x;\theta,\xi)} -  \frac{s\rho}{1+\rho w(x;\theta,\xi)}\Bigr\} \cdot \nabla_{\theta} w(x;\theta,\xi), \\[6pt]
&\tilde{S}_{3}(x,y,s;\theta,\eta) = \frac{s}{1+\rho}\Bigl\{h(x,y)w(x,y;\theta) - \eta\Bigr\},  \\[6pt]
&\tilde{S}_{4}(x,s;\theta,\xi,\eta) = \frac{(1-s)\rho}{1+\rho}\Bigl\{h(x,1)\frac{w_1(x;\theta,\xi)}{w(x;\theta,\xi)}+h(x,0)\frac{w_0(x;\theta,\xi)}{w(x;\theta,\xi)} - \eta \Bigr\}.
\end{align*}
Here,  $\nabla_{\xi}$ and $\nabla_{\theta}$ denote partial derivatives with respect to  $\xi$ and $\theta$, respectively.
{Additionally, let $\nabla_{\xi\xi^{\top}}$ denote the Hessian matrix, i.e., 
the matrix of second-order partial derivatives with respect to $\xi$.
 }

For the transfer‐learning setup considered in the main paper, the source and target populations differ; that is, $P^{(0)}(x,y) \neq P^{(1)}(x,y)$. 
For convenience, We use $\Ex_0$ and $\Ex_1$ to denote expectations under $P^{(0)}$ and $P^{(1)}$, respectively.
For any functions $K(x,y,s;\theta,\xi)$ and $K(x,s;\theta,\xi)$, define
\begin{align*}
\Ex\left\{K(X,Y,S;\theta,\xi)\right\} &= \Ex_{1}\left\{K(X,Y,S=1;\theta,\xi)\right\} + \Ex_{0}\left\{K(X,Y,S=0;\theta,\xi)\right\}\\[6pt]
&= \Ex_{1}\left\{K(X,Y,S=1;\theta,\xi) + K(X,Y,S=0;\theta,\xi)w(X,Y;\theta)\right\}, \\[6pt]
\Ex\left\{K(X,S;\theta,\xi)\right\} &= \Ex_{1}\left\{K(X,S=1;\theta,\xi)\right\} + \Ex_{0}\left\{K(X,S=0;\theta,\xi)\right\},\\[6pt]
&= \Ex_{1}\left\{K(X,S=1;\theta,\xi) + K(X,S=0;\theta,\xi)w(X;\theta,\xi)\right\}. 
\end{align*}
{Let $\nu_0 = (\xi_0^{\top}, \theta_0^{\top})^{\top}$ and $\omega_0 = (\nu_0^{\top}, \eta_0, \eta_0)^{\top}$ be the true values for $\nu$ and $\omega$.}
It is straightforward to verify that $\Ex\bigl\{\tilde{S}(X,Y,S;\omega_0)\bigr\} = 0$.  

For any  vector or matrix $B$, we write $B^{\otimes 2} = BB^{\top}$. 
We use $\|B\|$ to denote the sum of the absolute values of all its elements. Specifically,
{
for a vector $B \in \mathbb{R}^p$, 
\[
\|B\| = \sum_{i=1}^p |B_i|,
\]
and for a matrix $B \in \mathbb{R}^{p \times q}$, 
\[
\|B\| = \sum_{i=1}^p \sum_{j=1}^q |B_{ij}|.
\]
}

The asymptotic results in the main paper rely on the following regularity conditions.

\begin{description}
\item[Condition 1.] The true parameter vector $\omega_0$ lies in the interior of a compact parameter space. Moreover, $\omega_0$ is the unique solution to  
\(
\Ex\bigl\{\tilde{S}(X,Y,S;\omega)\bigr\} = 0.
\)

\item[Condition 2.]
The functions $g(x;\xi)$, $w(x;\theta,\xi)$, and $w(x,y;\theta)$ are twice continuously differentiable with respect to $\xi$ and $\theta$, respectively. Moreover, for any $\xi$, the parametric model $g(x;\xi)$ is continuous in $x$ and strictly bounded between zero and one, i.e.,  
\(
g(x;\xi) \in (0,1) 
\)
for all $x$.

\item[Condition 3.] 
There exists a measurable function $U(x,y,s)$ with $\Ex\!\left\{U(X,Y,S)\right\} < \infty$ such that, for all $\omega$ in a small neighborhood of $\omega_0$,  
\[
|\tilde{S}(x,y,s;\omega)| \leq U(x,y,s).
\]

\item[Condition 4.] 
Both matrices 
\[
A_1 = \Ex_1\!\left[\frac{\nabla_{\xi}^{\otimes 2} g(X;\xi_0)}{g(X;\xi_0)\{1-g(X;\xi_0)\}} \right]
\]
and 
\[
A_3 = \Ex_1\!\left[\Bigl\{ \frac{\rho}{1+ \rho w(X;\theta_0,\xi_0)} - \frac{1}{w(X;\theta_0,\xi_0)} \Bigr\}
\nabla_{\theta}^{\otimes 2} w(X;\theta_0,\xi_0)\right]
\]
are invertible.

    \item[Condition 5.] 
    The sample sizes of the source and target datasets, denoted $n_1$ and $n_0$, satisfy  
\(
{n_0}/{n_1} = \rho,
\)
where $\rho$ is a fixed constant.

\item[Condition 6.]
Let $H(x;\theta,\xi) = w_1(x;\theta,\xi)/w(x;\theta,\xi)$.  
{Assume that each element of $\nabla_{\nu} H(x;\nu_0)$, 
$\nabla_{\xi} g(x;\xi_0)$, and $\nabla_{\theta} w(x;\theta_0,\xi_0)$ 
has a finite $(2+\delta)$-th moment for some $\delta > 0$ under both the source and  target populations.}

\item[Condition 7.] 
The $c(X)$ is a continuous random variable.

{
\item[Condition 8.] 
Let $\mathbb{R}_{\nu_0}$ be a small  neighborhood of $\nu_0$, and let $\mathbb{R}_{c}$ denote the range of $c(x;\nu)$ for $\nu \in \mathbb{R}_{\nu_0}$ and $x \in \mathcal{X}$. 
Assume that $c(X;\nu)$ is a continuous random variable when $\nu\in \mathbb{R}_{\nu_0}$. For any $\nu$ and each $y\in \{0,1\}$, let $F_{y}(\cdot;\nu)$ denote the cumulative distribution function of $c(X;\nu)$ conditional on $Y=y$ in the target population, with corresponding density function $f_{y}(\cdot;\nu)$.  
\begin{itemize}
    \item[(a)] 
Assume that $\mathbb{R}_{c}$ is bounded, and that
\[
    \sup_{\nu \in \mathbb{R}_{\nu_0}} \sup_{t \in \mathbb{R}_{c}} f_{y}(t;\nu) \leq M < \infty.
\]

    \item[(b)] For any $\nu_1,\nu_2 \in \mathbb{R}_{\nu_0}$, there exists a measurable function $L(x)$ and some $\delta > 0$ such that 
\[
    |c(x;\nu_1) - c(x;\nu_2)| \leq L(x)\,\|\nu_1 - \nu_2\|,
\]
and 
\[
    \Ex_{0}\!\left\{L^{2+\delta}(X)\right\} < \infty.
\]
Let $z_{cL}(\cdot,\cdot)$ denotes the joint density of  $(c(X;\nu_0),L(X))$ under the target population. Assume 
$$
\int z^*(t) t dt < \infty,
$$
where 
$
z^*(t) = \sup_{k}z_{cL}(k,t).
$

 \item[(c)] Suppose $F_{y}(\cdot;\nu)$ is twice differentiable with respect to $\nu$, and  
$$
\sup_{t \in \mathbb{R}_{c}}\|\nabla_{\nu}F_{y}(t;\nu_0)\| < \infty. 
$$ 
Furthermore,   
$$
\sup_{\nu \in \mathbb{R}_{\nu_0}} \sup_{t\in \mathbb{R}_{c}} \|\nabla_{\nu\nu^{\top}}F_{y}(t;\nu)\|  < \infty.
$$

\end{itemize}

}
\end{description}

\section{Proof of  Theorems 1 and 2 in the main paper}
\label{proof-thm12}
Recall that the target functional is defined as $\eta = \Ex_0\{h(X, Y)\}$ in the main paper.
Let
\[
 S_{n_0+n_1}(\omega) = (S_{n_0+n_1,1}^{\top}(\xi), S_{n_0+n_1,2}^{\top}(\theta,\xi), S_{n_0+n_1,3}(\theta,\eta),S_{n_0+n_1,4}(\theta,\xi,\eta))^{\top},
\]
where 
\begin{align*}
&S_{n_0+n_1,1}(\xi)  = \sum_{i=1}^{n_0+n_1} s_{i}\Bigl\{\frac{y_i}{g(x_i;\xi)} 
 - \frac{1-y_i}{1-g(x_i;\xi)} 
 \Bigr\}\nabla_{\xi}g(x_i;\xi ), \\[6pt]
&  S_{n_0+n_1,2}(\theta,\xi) = \sum_{i=1}^{n_0+n_1} \Bigl\{ \frac{(1-s_i)}{ w(x_i;\theta,\xi)} - \frac{(1-s_i)\rho}{1+\rho w(x_i;\theta,\xi)} -\frac{  s_i \rho}{1+\rho w(x_i;\theta,\xi)}      \Bigr\} \nabla_{\theta} w(x_i;\theta,\xi) ,\\[6pt]
&  S_{n_0+n_1,3}(\theta,\eta ) = \sum_{i=1}^{n_0+n_1} s_{i} \Bigl\{h(x_i,y_i)w(x_i,y_i;\theta) - \eta \Bigr\}, \\[6pt]
&  S_{n_0+n_1,4}(\theta,\xi,\eta)= 
  \sum_{i=1}^{n_0+n_1}(1-s_i) \Bigl\{h(x_i,1)\frac{w_1(x_i;\theta,\xi)}{w(x_i;\theta,\xi)}+h(x_i,0)\frac{w_0(x_i;\theta,\xi)}{w(x_i;\theta,\xi)} - \eta \Bigr\}.  
\end{align*}
By the weak law of large numbers, it follows that for any fixed $\omega$, 
\begin{equation}
\label{S3.eq1}
S_{n_0+n_1}(\omega)/(n_0+n_1) \convergepto \Ex\bigl\{\tilde{S}(X, Y, S;\omega)\bigr\}.
\end{equation}

Recall that $\hat\xi$ and $\hat\theta$ are defined in Equations (10) and (11), respectively. 
It can be verified that $\hat{\omega} = (\hat\xi^{\top},\hat\theta^{\top},\hat\eta_{\mathrm{IW}},\hat\eta_{\mathrm{REG}})^{\top} $  satisfies 
$$
S_{n_0+n_1}(\hat\omega) = 0.
$$
Using Equation \eqref{S3.eq1}, and  under Condition 1--Condition 5, it follows directly from Theorem 5.9 of \cite{van2000asymptotic} that 
$$
\hat{\omega} \convergepto \omega_0 .$$

Next, we discuss the asymptotic normality of $\hat{\omega}$. Recall that for any matrix or vector $B$, $B^{\otimes 2} = BB^{\top}$, and define
\begin{align*}
&A_1 = \Ex_1\bigl[\nabla_{\xi}^{\otimes 2}g(X;\xi_0)/\left\{g(X;\xi_0)(1-g(X;\xi_0))\right\} \bigr], \\[6pt]
&A_2 = \Ex_1\left[\left\{
\rho/(1+\rho w(X;\theta_0,\xi_0) )-1/ w(X;\theta_0,\xi_0)\right\}
\nabla_{\xi}w(X;\theta_0,\xi_0)\nabla^{\top}_{\theta} w(X;\theta_0,\xi_0) \right]
, \\[6pt]
&A_3 = \Ex_1\left[\left\{
\rho/(1+\rho w(X;\theta_0,\xi_0) )-1/ w(X;\theta_0,\xi_0)\right\}
\nabla_{\theta}^{\otimes 2} w(X;\theta_0,\xi_0)\right],
\\[6pt]
&A_4 = \Ex_1\left[
h(X,Y)
\nabla_{\theta}w(X,Y;\theta_0)
\right], \\[6pt]
&A_5 = \Ex_1\left[\left\{
h(X,1)-h(X,0)\right\}\cdot \left\{ 1-H(X;\theta_0,\xi_0) \right\}\cdot  \nabla_{\xi}w_1(X;\theta_0,\xi_0)
\right] \\[6pt]
&\hspace{1cm}  - \Ex_1\left[\left\{
h(X,1)-h(X,0)\right\}\cdot H(X;\theta_0,\xi_0)\cdot  \nabla_{\xi}w_0(X;\theta_0,\xi_0)\right]
,\\[6pt]
&A_6 = \Ex_1\left[\left\{
h(X,1)-h(X,0)\right\}\cdot \left\{ 1-H(X;\theta_0,\xi_0) \right\} \cdot  \nabla_{\theta}w_1(X;\theta_0,\xi_0)\right]
 \\[6pt]
&\hspace{1cm}  - \Ex_1
\left[ \left\{h(X,1)-h(X,0)\right\}\cdot H(X;\theta_0,\xi_0)\cdot  \nabla_{\theta}w_0(X;\theta_0,\xi_0)\right].
\end{align*}
In addition, we define 
\begin{align*}
&E_1 = \Ex_1\left[\rho^2\nabla_{\theta}^{\otimes 2}w(X;\theta_0,\xi_0)/\left\{1+ 
\rho w(X;\theta_0,\xi_0)\right\}^2 \right] \\[6pt]
& \hspace{1cm} -
\left\{\Ex_1\left[\rho \nabla_{\theta} 
w(X;\theta_0,\xi_0)/\left\{1+\rho w(X;\theta_0,\xi_0)\right\}\right]\right\}^{\otimes 2},\\[6pt]
&E_2=\Ex_1\left[\nabla_{\theta}^{\otimes 2}w(X;\theta_0,\xi_0)/\left\{w(X;\theta_0,\xi_0)(1+ 
\rho w(X;\theta_0,\xi_0))^2\right\} \right]  \\[6pt]
& \hspace{1cm} -    
\left\{\Ex_1\left[\nabla_{\theta}w(X;\theta_0,\xi_0)/
 \left\{1+\rho w(X;\theta_0,\xi_0)\right\}\right]\right\}^{\otimes 2}, \\[6pt]
&E_3 = \Ex_1\left[h^2(X,Y)W^{2}(X,Y;\theta_0) \right], \\[6pt]
&E_4 =  \Ex_1 \left[ \left\{h(X,1)w_1(X;\theta_0,\xi_0)+ h(X,0)w_0(X;\theta_0,\xi_0)\right\}^2/w(X;\theta_0,\xi_0)\right], \\[6pt]
&E_5 = \Ex_1 \left[ \left\{h(X,1)w(X,1;\theta_0)- h(X,0)w(X,0;\theta_0)\right\}\nabla_{\xi}g(X;\xi_0) \right],\\[6pt]
&E_6 
= \Ex_1 \Bigl[ 
   \nabla_{\theta} w(X;\theta_0,\xi_0)\,
   \frac{ h(X,1)H(X;\theta_0,\xi_0) + h(X,0)\left\{1-H(X;\theta_0,\xi_0)\right\} }
        {1 + \rho w(X;\theta_0,\xi_0)}
   \Bigr] \\[6pt]
&\hspace{1cm} - \Ex_1 \Bigl\{ h(X,1)\,w_1(X;\theta_0,\xi_0) + h(X,0)\,w_0(X;\theta_0,\xi_0) \Bigr\}\cdot
          \Ex_1 \Bigl\{ \frac{\nabla_{\theta} w(X;\theta_0,\xi_0)}{1 + \rho\,w(X;\theta_0,\xi_0)} \Bigr\}.
\end{align*}

It is straightforward to verify that 
$$
\Ex\left\{  S_{n_0+n_1} (\omega_0)  \right\} = 0,$$ 
and it is also easy to check that  
\begin{align*}
    \Var\left\{ S_{n_0+n_1} (\omega_0) \right\} &= 
\frac{n_0+n_1}{1+\rho}E, 
\end{align*}
where 
$$
E=
\begin{pmatrix}
A_1 & 0 & E_5 & 0 \\
0 & E_1 + \rho E_2 & 0 & \rho E_6\\
E_5^{\top} &0 & E_3 - \eta_0^2 & 0 \\
0 &\rho E_6^{\top} &0 & \rho( E_4 - \eta_0^2) 
\end{pmatrix}. 
$$
By the central limit theorem,
\begin{equation}
\label{S3.clt}
\frac{1}{\sqrt{n_0+n_1}} S_{n_0+n_1}(\omega_0) 
\convergeto {\cal N}\Bigl(0, \frac{1}{1+\rho} E\Bigr).
\end{equation}

Applying a first-order Taylor expansion and the weak law of large numbers, under Condition 6, we  have
\begin{equation}
\label{S3.expansion1}
0 = S_{n_0+n_1}(\hat{\omega}) = S_{n_0+n_1}({\omega_0}) + \nabla^{\top}_{\omega} S_{n_0+n_1}({\omega_0}) (\hat{\omega} - {\omega_0}) + o_P(\sqrt{n_0+n_1}). 
\end{equation}
By the weak law of large numbers,   
\begin{equation}
\label{S3.WLL1}
 \frac{1}{n_0+n_1} \nabla_{\omega^{\top}} S_{n_0+n_1}({\omega_0}) = \frac{1}{1+\rho} A+o_p(1), 
\end{equation}
where 
$$A=
 \begin{pmatrix}
-A_1 &  0 & 0 &0 \\
\rho A_2^{\top} & \rho A_3  &0& 0\\
0 & A_4^{\top} & -1  & 0\\
\rho A_5^{\top} & \rho A_6^{\top} & 0 & -\rho
\end{pmatrix}.
$$
Under Condition 4, the matrix $A$ is invertible, and 
\begin{equation}
    \label{S3.Ainv}
A^{-1} = \begin{pmatrix}
-A_1^{-1} &0 &0 &0 \\
A_3^{-1}A_2^{\top}A_{1}^{-1} & \frac{1}{\rho}A_3^{-1} & 0&0 \\
A_4^{\top}A_3^{-1}A_2^{\top}A_1^{-1} & \frac{1}{\rho}A_4^{\top}A_3^{-1} & -1 &0 \\
-A_5^{\top}A_1^{-1} + A_6^{\top}A_3^{-1}A_2^{\top}A_1^{-1} & \frac{1}{
\rho
}A_6^{\top}A_3^{-1} & 0 & -\frac{1}{\rho}
\end{pmatrix}. 
\end{equation}

Plugging \eqref{S3.WLL1} to \eqref{S3.expansion1}, and after some algebra, we obtain
\begin{equation}
\label{expan-omega}
    \sqrt{n_0+n_1}(\hat{\omega} - {\omega_0}) = - (1+\rho)A^{-1} \frac{1}{\sqrt{n_0+n_1}}S_{n_0+n_1}(\omega_0) + o_p(1),
\end{equation}
which, together with \eqref{S3.clt} and Slutsky's theorem, implies  
\begin{equation}
\label{S3.normality}
\sqrt{n_0+n_1} (\hat{\omega} - \omega_0) \ \convergeto \ \mathcal{N}\Bigl(0, (1+\rho) \tilde \Sigma \Bigr),
\end{equation}
where $\tilde{\Sigma} = A^{-1} E (A^{-1})^\top$. 
After some algebra, we have 
\ba
\label{S3.tilde.sigma}
\tilde{\Sigma}  
=A^{-1} \begin{pmatrix}
A_1 & 0 & E_5 & 0 \\
0 & E_1 + \rho E_2 & 0 & \rho E_6\\
E_5^{\top} &0 & E_3 - \eta_0^2 & 0 \\
0 &\rho E_6^{\top} &0 & \rho (E_4 - \eta_0^2) 
\end{pmatrix} (A^{-1})^{\top}
\ea
with 
\begin{align*}
&\tilde{\Sigma}_{11} = A_1^{-1},  \\[6pt]
&\tilde{\Sigma}_{12} = \tilde{\Sigma}_{21}^{\top} = -A_1^{-1}A_2 A_3^{-1},  \\[6pt]
&\tilde{\Sigma}_{13} = \tilde{\Sigma}_{31}^{\top} = A_1^{-1}(E_5 - A_2 A_3^{-1}A_4), \\[6pt]
&\tilde{\Sigma}_{14} = \tilde{\Sigma}_{41}^{\top} = A_1^{-1}(A_5- A_2 A_3^{-1}A_6),\\[6pt] 
&\tilde{\Sigma}_{22} = A_3^{-1}\left\{ A_2^{\top}A_1^{-1}A_2 + (E_1 + \rho E_2)/\rho^2\right\}A_3^{-1},\\[6pt]
&\tilde{\Sigma}_{23} = \tilde{\Sigma}_{32}^{\top} =
A_3^{-1}A_2^{\top}A_1^{-1}(A_2A_3^{-1}A_4 -E_5) + A_3^{-1}(E_1 +\rho E_2)A_3^{-1}A_4/\rho^2, \\[6pt]
&\tilde{\Sigma}_{24} = \tilde{\Sigma}_{42}^{\top} =
A_3^{-1}A_2^{\top}A_1^{-1}(A_2A_3^{-1}A_6 -A_5) + A_3^{-1}(E_1 +\rho E_2)A_3^{-1}A_6/\rho^2+A_3^{-1}E_6/\rho, \\[6pt]
&\tilde{\Sigma}_{33}= A_4^{\top}A_3^{-1 }\left\{ A_2^{\top}A_1^{-1}A_2 + (E_1 + \rho E_2)/\rho^2\right\}A_3^{-1 }A_4 - 2E_5^{\top}A_1^{-1}A_2A_3^{-1}A_4 + E_3 -\eta_0^2,\\[6pt]
&\tilde{\Sigma}_{34} = \tilde{\Sigma}_{43}^{\top} =
A_4^{\top}A_3^{-1}\left\{ A_2^{\top}A_1^{-1}A_2 + (E_1 + \rho E_2)/\rho^2\right\}A_3^{-1}A_6  
   \\[6pt]
&\hspace{2cm}+ A_4^{\top}A_3^{-1}\left\{ E_6/\rho - A_2^{\top}A_{1}^{-1}A_5\right\} +
E_5^{\top}A_{1}^{-1}\left\{A_5 - A_2A_3^{-1}A_6\right\}, \\[6pt]
&\tilde{\Sigma}_{44} =  A_6^{\top}A_3^{-1}\left\{A_2^{\top}A_1^{-1}A_2 + (E_1 + \rho E_2)/\rho^2\right\}A_3^{-1}A_6  
-2A_5^{\top}A_1^{-1}A_2 A_3^{-1}A_6      \\[6pt]
&\hspace{1cm} + A_5^{\top}A_1^{-1}A_5 -  (E_4 -\eta_0^2)/\rho. 
\end{align*}

Let $d_{\xi}$ and $d_{\theta}$ denote the dimensions of $\xi$ and $\theta$, respectively, and let $I_d$ denote the  $d\times d$ identity matrix. Define $D \in \mathbb{R}^{(d_{\xi}+d_{\theta}) \times (d_{\xi}+d_{\theta}+2)}$ as
\begin{equation*}
    D = \begin{pmatrix}
    I_{d_{\xi}} & 0 & {0}  & 0\\
    0&  I_{d_{\theta}} & {0} & 0
\end{pmatrix}.
\end{equation*}

By \eqref{S3.normality}, we directly obtain
\begin{equation}
\label{asy-nu}
    \sqrt{n_0+n_1} 
\begin{pmatrix}
\hat{\xi} - \xi_0 \\
\hat{\theta} - \theta_0
\end{pmatrix} 
= \sqrt{n_0+n_1} \, D (\hat{\omega} - \omega_0) 
\ \convergeto \ \mathcal{N}\Bigl(0, (1+\rho) \Sigma\Bigr),
\end{equation}
where 
$$
\Sigma=D  \tilde\Sigma D^{\top}.
$$
Together with \eqref{S3.tilde.sigma}, this implies that 
\ba
\Sigma= \begin{pmatrix}
A_{1}^{-1} & -A_{1}^{-1}A_{2}A_{3}^{-1} \\
-A_{3}^{-1}A_{2}^{\top}A_{1}^{-1} & A_{3}^{-1}\left\{
A_{2}^{\top}A_{1}^{-1}A_{2} + (E_{1} + \rho E_{2})/\rho^2 \right\} A_{3}^{-1}
\end{pmatrix}.
\ea 

Let ${e}_{k}$ be the $(d_{\xi}+d_{\theta} +2)$-dimensional vector whose $k$-th element is 1 and the rest are 0. 
Then, for estimators $\hat{\eta}_{\mathrm{IW}}$ and $\hat{\eta}_{\mathrm{REG}}$, we have
\bas
\sqrt{n_0+n_1}(\hat{\eta}_{\mathrm{IW}} - \eta_0) = \sqrt{n_0+n_1}  {e}^{\top}_{d_{\xi}+d_{\theta} +1} (\hat{\omega}-\omega_0)  
& \convergeto &
\mathcal{N}\bigl(0, (1+\rho) \sigma_{\mathrm{IW}}^2\bigr)
\eas 
and 
\bas
\sqrt{n_0+n_1}(\hat{\eta}_{\mathrm{REG}} - \eta_0) =\sqrt{n_0+n_1} {e}^{\top}_{d_{\xi}+d_{\theta} +2} (\hat{\omega}-\omega_0)  
& \convergeto &
 \mathcal{N}\bigl(0, (1+\rho) \sigma_{\mathrm{REG}}^2\bigr),
\eas
where
\begin{equation}
\sigma_{\mathrm{IW}}^2
=A_4^{\top}A_3^{-1 }\left\{ A_2^{\top}A_1^{-1}A_2 + 
(E_1 + \rho E_2)/\rho^2\right\}A_3^{-1 }A_4 - 2E_5^{\top}A_1^{-1}A_2A_3^{-1}A_4 +E_3 -\eta_0^2
\end{equation}
and 
\begin{align}
    \sigma_{\mathrm{REG}}^2
&=A_6^{\top}A_3^{-1}\left\{ A_2^{\top}A_1^{-1}A_2 + (E_1 + \rho E_2)/\rho^2\right\}A_3^{-1}A_6 
-2A_5^{\top}A_1^{-1}A_2 A_3^{-1}A_6  \notag\\
&\hspace{0.2cm} + A_5^{\top}A_1^{-1}A_5-  (E_4 -\eta_0^2)/\rho.
\end{align}
This completes the proof of Theorem 2 in the main paper.

\section{Proof of Theorem 3 in the main paper}
\label{proof-thm3}
\subsection{Some preparation\label{S4.section1}}
In this section, we introduce notation and present preliminary results for the proof of Theorem 3 in the main paper. Recall that Theorem 3 concerns a fixed, known classifier 
$c(x)$, which is assumed to be independent of the estimated parameters.

Define $$
H(x;\nu_0) = w_{1}(x;\nu_0)/w(x;\nu_0).
$$
For any $u \in (0,1)$, let
$$
\tau_{1-u} = F_{0}^{-1}(1-u). 
$$
Further, define 
\[Z_{n_0+n_1} = (Z_{n_0+n_1,1}^{\top},Z_{n_0+n_1,2},Z_{n_0+n_1,3},Z_{n_0+n_1,4},Z_{n_0+n_1,5})^{\top},\]
where 
\begin{align*}
        Z_{n_0+n_1,1} &= {(n_0+n_1)}^{-1}  \begin{pmatrix}
    S_{n_0+n_1,1}(\xi_0) \\[6pt]
    S_{n_0+n_1,2}(\theta_0,\xi_0)
\end{pmatrix}, \\[6pt]
Z_{n_0+n_1,2} &= {(n_0+n_1)}^{-1} \sum_{i=1}^{n_0+n_1} (1-s_i)H(x_i;\nu_0)\left\{F_0(c(x_i)) - AUC\right\}, \\[6pt]
Z_{n_0+n_1,3} &=  {(n_0+n_1)}^{-1} \sum_{i=1}^{n_0+n_1} (1-s_i)\left\{1-H(x_i;\nu_0)\right\} \left\{1 - F_1(c(x_i)) - AUC\right\}, \\[6pt]
Z_{n_0+n_1,4} &=  {(n_0+n_1)}^{-1} \sum_{i=1}^{n_0+n_1} (1-s_i)\left\{1-H(x_i;\nu_0)\right\}  \left\{I(c(x_i)\leq \tau_{1-u}) - F_0(\tau_{1-u})\right\}, \\[6pt]
Z_{n_0+n_1,5} &= {(n_0+n_1)}^{-1} \sum_{i=1}^{n_0+n_1} (1-s_i) H(x_i;\nu_0) \left\{I(c(x_i)\leq \tau_{1-u}) - F_1(\tau_{1-u})\right\}.
\end{align*}

It can be verified that $$\Ex(Z_{n_0+n_1}) = 0.$$
Note that $Z_{n_0+n_1}$ is the average of independent random vectors. 
By the central limit theorem,  we have  
\begin{equation}
\label{asy-Z}
    \sqrt{n_0+n_1} Z_{n_0+n_1} \convergeto {\cal N}(0, \Sigma_{Z}),
\end{equation}
where 
$\Sigma_Z$ is the variance and covariance matrix of $\sqrt{n_0+n_1}Z_{n_0+n_1}$.

Define
\begin{align}
\bm{E}_{1} &= \Ex_{0} \left[ \nabla_{\nu}H(X;\nu_0)\left\{ F_0(c(X)) -AUC\right\}\right],\label{matrix-e1}\\
\bm{E}_{2} &= \Ex_{0}\left[\nabla_{\nu}H(X;\nu_0)  \left\{ 1- F_1(c(X)) - AUC\right\}\right],\label{matrix-e2}\\
\bm{E}_{3} &= \Ex_{0}\left[ \nabla_{\nu}H(X;\nu_0) \left\{ I(c(X)\leq\tau_{1-u}) -F_0(\tau_{1-u})\right\}\right],\label{matrix-e3}\\
\bm{E}_{4} &= \Ex_{0}  \left[ \nabla_{\nu}H(X;\nu_0)  \left\{ I(c(X)\leq\tau_{1-u}) -F_1(\tau_{1-u})\right\} \right]. \label{matrix-e4}
\end{align}
Let $$
\bar{H} = \sum_{j=1}^{n_0} H(x_{n_1 + j};\hat{\nu}) / n_0.
$$
Given a known classifier $c(x)$, the estimators for the distribution functions $F_0$ and $F_1$ are defined as
\begin{equation}
    \label{e-F0}
    \widehat{F}_{0}(u) = \frac{1}{n_0(1-\bar{H})}\sum_{j=1}^{n_0}(1-H(x_{n_1 + j};\hat{\nu}))\cdot I(c(x_{n_1 + j})\leq u)
\end{equation}
and 
\begin{equation}
    \label{e-F1}
    \widehat{F}_{1}(u) = \frac{1}{n_0\bar{H}}\sum_{j=1}^{n_0}H(x_{n_1 + j};\hat{\nu})\cdot I(c(x_{n_1 + j}) \leq u).
\end{equation}

Let 
\begin{equation*}
    V =  \begin{pmatrix}
    -A_1^{-1} & 0 \\
    A_3^{-1}A_2^{\top}A_1{-1} & A_3^{-1}/\rho 
\end{pmatrix},
\end{equation*}
which corresponds to the upper-left
$(d_{\xi} + d_{\theta})\times (d_{\xi} + d_{\theta})$ submatrix of $A^{-1}$ in \eqref{S3.Ainv}. 
 \begin{lemma}
   \label{lem-toprovethm3}
   Suppose Conditions 1--7  in Section~\ref{prelim} are satisfied. We have
   \begin{description}
       \item[(a)] $\hat{\nu} - {\nu}_0 = - (1+\rho)VZ_{n_0+n_1,1} + o_p((n_0+n_1)^{-1/2}) = O_p((n_0+n_1)^{-1/2})$;   
       \item[(b)] the processes $\sqrt{n_0+n_1}\bigl\{\widehat{F}_0(u) - F_0(u)\bigr\}$ and $\sqrt{n_0+n_1}\bigl\{\widehat{F}_1(t) - F_1(t)\bigr\}$ converge jointly to a bivariate, tight, mean-zero Gaussian process.
   \end{description}
 \end{lemma}
\begin{proof} 
(a)
By \eqref{expan-omega} and \eqref{asy-nu}, 
we have
$$
\hat{\nu} - \nu_0 = - D A^{-1}S_{n_0+n_1}/(n_0+n_1) + o_p((n_0+n_1)^{-1/2})  = - (1+\rho)V Z_{n_0+n_1,1} +  o_p((n_0+n_1)^{-1/2}).$$ 
From the asymptotic normality of $\hat{\nu} - \nu_0$ established in Theorem~1 of the main paper, it follows that $\hat{\nu} - \nu_0 = O_p((n_0 + n_1)^{-1/2})$.

(b) We can verify that 
\begin{align}
\label{eF0-dif0}
        \widehat{F}_1(t) - F_1(t) &=\frac{D_{n_0 1}(t)}{\bar{H}}, 
  \end{align}
  where 
  \begin{align*}
        D_{n_0 1}(t)=
        \frac{1}{n_0}
\sum_{j=1}^{n_0}H(x_{n_1 + j};\hat{\nu}) \left\{I(c(x_{n_1 + j}) \leq t) - F_1(t)\right\}.
    \end{align*}

Next, we derive an approximation for $D_{n_01}(t)$. Applying a first-order Taylor expansion, we obtain
\begin{align}
\label{expan-D0}
        D_{n_0 1}(t) =& \frac{1}{n_0}\sum_{j=1}^{n_0} H(x_{n_1 + j};\nu_0) \left\{I(c(x_{n_1 + j}) \leq t) - F_1(t)\right\} \notag  \\
        &+ \frac{1}{n_0}\sum_{j=1}^{n_0}
        \nabla_{\nu}^{\top} H(x_{n_1 + j};\nu_0)\left\{I(c(x_{n_1 + j}) \leq t) - F_1(t)\right\}(\hat{\nu} - \nu_0) + e_{n_0 1}(t) \notag \\
            =&\frac{1}{n_0}\sum_{j=1}^{n_0}H(x_{n_1 + j};\nu_0)\left\{I(c(x_{n_1 + j}) \leq t) - F_1(t)\right\}+\bm{K}_{1}^{\top}(t) (\hat{\nu} - \nu_0) + e_{n_0}(t),
    \end{align}
where $e_{n_0 1}(t)$ denotes the remainder term from the first-order expansion, and
\bas
e_{n_0}(t) = e_{n_0 1}(t) - 
\bm{e}_{n_0 2}^{\top}(t)(\hat{\nu} - \nu_0)
\eas 
with 
\bas
\bm{e}_{n_0 2}(t) = \frac{1}{n_0}\sum_{j=1}^{n_0}
        \nabla_{\nu} H(x_{n_1 + j};\nu_0)\left\{I(c(x_{n_1 + j}) \leq t) - F_1(t)\right\}-\bm{K}_{1}(t)
\eas
and 
\begin{equation*}
\bm{K}_{1}(t) = \Ex_0 \left[ \nabla_{\nu} H(X;\nu_0) \left\{ I(c(X) \leq t) - F_1(t)\right\}\right].
\end{equation*}

Let $\mathcal{G}_{1}$ be a function class indexed by $t$, defined as 
\bas
\mathcal{G}_{1} = \left\{x \mathrel{\mapsto}
\nabla_{\nu}H(x;\nu_0) \left\{I(c(x) \leq t) - F_1(t)\right\}: t \in \mathbb{R}
\right\}. 
\eas 
Since $\nabla_{\nu} H(x; \nu_0)$ is square-integrable under Condition 6, it follows from Theorem 2.10.6 and Example 2.10.10 in \cite{van1996weak} that the function class $\mathcal{G}_1$ is Donsker. Moreover, the condition $n_0 / n_1 = \rho$ implies $n_0 = O(n_1)$, and hence
\ba
\label{errtem2}
\sup_{t}\|\bm{e}_{n_02}(t)\| = O_p(n_0^{-1/2})=  O_p((n_0+n_1)^{-1/2}).
\ea

Since $I(c(x) \leq t) - F_1(t)$ is uniformly bounded in $u$ and $\nabla_{\nu}H(x;\nu_0)$ has finite second  moment, it follows from Part (a) of Lemma~\ref{lem-toprovethm3} that
    \begin{align}
    \label{errtem1}
\sup_{t}| e_{n_0 1}(t)|  &= o_p(1)\sup_{t}| \frac{1}{n_0}\sum_{j=1}^{n_0}
        \nabla_{\nu}^{\top} H(x_{n_1 + j};\nu_0))\left\{I(c(x_{n_1 + j}) \leq t) - F_1(t)\right\}(\hat{\nu} - \nu_0) | \notag  \\
        &= o_p((n_0+n_1)^{-1/2}).
    \end{align}
Combining \eqref{errtem1} and \eqref{errtem2}, 
and using the consistency of $\hat{\nu}$, we conclude that
\ba
\label{errtem}
\sup_{t}\|e_{n_0}(t)\| = o_p((n_0+n_1)^{-1/2}).
\ea 
Substituting \eqref{errtem} into 
\eqref{expan-D0} gives
\begin{align}
        \label{re-dm0}
D_{n_0 1}(t) &= \frac{1}{n_0}\sum_{j=1}^{n_0}H(x_{n_1 + j};\nu_0) \left\{I(c(x_{n_1 + j}) \leq t) - F_1(t)\right\} \notag \\
&\hspace{0.5cm}+(\hat{\nu} - \nu_0)^{\top}\bm{K}_{1}(t)  +  o_p((n_0+n_1)^{-1/2}). 
\end{align}

Note that, uniformly,
\bas
0 \leq H(x;\nu_0) \leq 1.
\eas 
By Example 2.10.10 of \cite{van1996weak}, the function class 
\bas
\mathcal{G}_{2} = \left\{x \mathrel{\mapsto}
H(x;\nu_0)  \left\{I(c(x) \leq t) - F_1(t)\right\}: t \in \mathbb{R}
\right\} 
\eas 
is also Donsker. Hence, we have
\begin{equation*}
\sup_{t}|
\frac{1}{n_0}\sum_{j=1}^{n_0}H(x_{n_1 + j};\nu_0) \left\{I(c(x_{n_1 + j}) \leq t) - F_1(t)\right\}
| = O_p((n_0+n_1)^{-1/2}). 
\end{equation*}

Under Condition 7, $c(X)$ is a continuous random variable, and by Condition 6, $\nabla_{\nu} H(x; \nu_0)$ has finite second moment. It then follows that
\bas
\sup_{t}\|\bm{K}_1(t)
\| < \infty,
\eas 
and therefore,
\ba
\label{ev}
\sup_{t}|(\hat{\nu} - \nu_0)^{\top}\bm{K}_{1}(t)
| = O_p((n_0+n_1)^{-1/2}).
\ea 
Combing  \eqref{re-dm0}-\eqref{ev}, we obtain
\bas
\sup_{t}|D_{n_0 1}(t)| = O_p((n_0+n_1)^{-1/2}).
\eas

Let $\mu_0$ denote the true value of the target population mean. By Theorem~2 with $h(x, y) = y$,  $\bar{H} = \hat{\mu}_{\mathrm{REG}}$ is consistent for $\mu_0$. Combining this with \eqref{eF0-dif0} and \eqref{re-dm0},  we obtain the expansion:
    \begin{align*}
        &\widehat{F}_1(t) - F_1(t) \notag \\
 &=\frac{1}{\mu_0} \left[
\frac{1}{n_0}\sum_{j=1}^{n_0}H(x_{n_1 + j};\nu_0) \left\{I(c(x_{n_1 + j}) \leq t) - F_1(t)\right\} +(\hat{\nu} - \nu_0)^{\top}\bm{K}_{1}(t) 
\right]  \notag  \\
&\hspace{0.5cm}+ o_p((n_0+n_1)^{-1/2}).
\end{align*}

Part (a) of Lemma 1 implies 
\bas
\hat{\nu} -\nu_0 = \frac{1}{n_0+n_1}\sum_{i=1}^{n_0+n_1} -(1+\rho)V T(x_i,y_i,s_i;\nu_0) + o_p((n_0+n_1)^{-1/2}),
\eas 
where $-(1+\rho) VT(x,y,s;\nu_0)$ can be regarded as the influence function of $\hat{\nu}$.
Then, we have
 \begin{align}
    \label{final-f1}
\widehat{F}_1(t) - F_1(t) =
\frac{1}{n_0+n_1}
\sum_{i=1}^{n_0+n_1} L_{1}(x_i,y_i,s_i;t) + o_p((n_0+n_1)^{-1/2}), 
 \end{align}
where 
\begin{align*}
L_{1}(x,y,s;t) =& 
\left[
\frac{(1-s) (1+\rho)}{\mu_0 \rho}
H(x;\nu_0)\left\{I(c(x) \leq t) - F_1(t)\right\} \right. \notag  \\
&\left.
- \frac{1+\rho}{\mu_0}\bm{K}^{\top}_{1}(t)V T(x,y,s;\nu_0) \right].
    \end{align*}

Following a similar procedure, we have
    \begin{align}
    \label{final-f0}
\widehat{F}_0(u) - F_0(u) =  \frac{1}{n_0+n_1}
\sum_{i=1}^{n_0+n_1} L_{0}(x_i,y_i,s_i;u) + o_p((n_0+n_1)^{-1/2}),
\end{align}
where 
 \begin{align*}
 L_{0}(x,y,s;u)
= & \left[
\frac{(1- s)(1+\rho )}{(1-\mu_0) \rho}
\{ 1-H(x;\nu_0)\} \left\{I(c(x)) \leq u) - F_0(u)\right\}
\right. \notag \\
& \left.
+ \frac{1+\rho}{1-\mu_0} \bm{K}^{\top}_{0}(u) V T(x,y,s;\nu_0) \right]   
\end{align*}
and
\begin{align*}
  \bm{K}_{0}(u) = \Ex_0\left[ \nabla_{\nu} H_{0}(X;\nu_0) \left\{I(c(X) \leq u) - F_0(u)\right\}\right]
\end{align*}  
is continuous in $u$ and uniformly bounded.

The approximations in \eqref{final-f1} and \eqref{final-f0} imply that both $\widehat{F}_{1}(t)$
and $\widehat{F}_{0}(u)$ are asymptotically
linear, with influence functions $L_1(x, y,s;t)$ and $L_0(x,y,s;u)$, respectively.
Based on the above discussion, the class of $ L_0(x,y,s;u) $ (indexed by $u$) and the class of $L_1(x,y,s;t)$ (indexed by $t$) are both Donsker classes. Hence, the processes $\sqrt{n_0+n_1} \{ \widehat{F}_0(u) - F_0(u) \}$ and 
and  $\sqrt{n_0+n_1} \{ \widehat{F}_1(t) - F_1(t) \}$
converge jointly to a bivariate, mean-zero, Gaussian process.
\end{proof}

\subsection{Proof of the first part of Theorem 3}
Recall that 
\[
\widehat{AUC} = \int \widehat{F}_0(u) d\widehat{F}_{1}(u). 
\]
By Part (b) of Lemma~\ref{lem-toprovethm3} and the functional delta method (see Theorem 20.8, Lemma 20.10 in \cite{van2000asymptotic}), we have 
\begin{align}
\label{expan-auc0}
     \widehat{AUC} - {AUC} &= - \int \bigl\{\widehat{F}_{1}(u) - F_1(u)\bigr\}dF_0(u) \notag  \\
     &\hspace{0.5cm} + \int \bigl\{\widehat{F}_{0}(u) - F_0(u)\bigr\}dF_1(u) + o_p((n_0+n_1)^{-1/2}).
\end{align}
Using \eqref{e-F0}-\eqref{e-F1} and after straightforward calculations, we further obtain 
\begin{align}
    \label{expan-auc}
     \widehat{AUC} - {AUC} = & \frac{1}{n_0 \bar{H}}\sum_{j=1}^{n_0}H(x_{n_1 + j};\hat{\nu})\left\{F_0(c(x_{n_1 + j})) - AUC\right\}   \notag \\
     & +\frac{1}{n_0 (1-\bar{H})}\sum_{j=1}^{n_0}\bigl\{1-H(x_{n_1 + j};\hat{\nu})\bigr\}\bigl\{1-F_1(c(x_{n_1 + j})) - AUC\bigr\} \notag \\
      & + o_p((n_0+n_1)^{-1/2}).
\end{align}

Using the first-order Taylor expansion and { Condition 6}, we get
\ba
\label{taylor1}
 &&\frac{1}{n_0}\sum_{j=1}^{n_0}H(x_{n_1 + j};\hat{\nu})\left\{F_0(c(x_{n_1 + j})) - AUC\right\}  \notag\\
    &&=\frac{1}{n_0}\sum_{j=1}^{n_0}H(x_{n_1 + j};\nu_0)\left\{F_0(c(x_{n_1 + j})) - AUC\right\} \notag\\
    &&\hspace{0.5cm} + (\hat{\nu}-\nu_0)^{\top}\cdot   \frac{1}{n_0}\sum_{j=1}^{n_0} \nabla_{\nu}
    {H}(x_{n_1 + j};\nu_0)\left\{F_0(c(x_{n_1 + j}))  - AUC\right\}+ o_p((n_0+n_1)^{-1/2}) \notag \\
    &&=\frac{(1+\rho)}{\rho} \cdot  Z_{n_0+n_1,2} + (\hat{\nu}-\nu_0)^{\top} \bm{E}_1 + o_p((n_0+n_1)^{-1/2}) \notag \\
    &&=\frac{(1+\rho)}{\rho} \cdot  Z_{n_0+n_1,2} - (1+\rho) \bm{E}_1^{\top}V Z_{n_0+n_1,1} + o_p((n_0+n_1)^{-1/2}),
\ea
where the second equality follows from the law of large numbers, and $\bm{E}_1$ is defined in \eqref{matrix-e1}, with the expansion of $\hat{\nu} - \nu_0$ provided in Part (a) of Lemma~\ref{lem-toprovethm3}. 

Similarly, with $\bm{E}_2$ defined in \eqref{matrix-e2}, we have
    \begin{align}
    \label{taylor2}
        &\frac{1}{n_0}\sum_{j=1}^{n_0}(1-H(x_{n_1 + j};\hat{\nu}))\left[1-F_1(c(x_{n_1 + j})) - AUC\right] \notag  \\
        &= 
    \frac{(1+\rho)}{\rho} \cdot  Z_{n_0+n_1,3} - (\hat{\nu}-\nu_0)^{\top} \bm{E}_2 + o_p((n_0+n_1)^{-1/2}) \notag \\
    &=\frac{(1+\rho)}{\rho} \cdot  Z_{n_0+n_1,3} + (1+\rho) \bm{E}_2^{\top}VZ_{n_0+n_1,1} + o_p((n_0+n_1)^{-1/2}). 
    \end{align}
    
Plugging \eqref{taylor1}-\eqref{taylor2} into \eqref{expan-auc} yields
\begin{align}
\label{final-auc-expan}
    \widehat{AUC} - {AUC} =& \frac{1}{\bar{H}}\left\{\frac{1+\rho}{\rho}Z_{n_0+n_1,2} - (1+\rho) \bm{E}_1^{\top}V Z_{n_0+n_1,1}\right\} \notag \\
&  +\frac{1}{1-\bar{H}}\left\{\frac{1+\rho}{\rho} Z_{n_0+n_1,3} + (1+\rho)\bm{E}_2^{\top}V Z_{n_0+n_1,1}\right\}  + o_p((n_0+n_1)^{-1/2}) \notag \\ 
=&\frac{1}{\mu_0}\left\{\frac{1+\rho}{\rho}Z_{n_0+n_1,2} - (1+\rho) \bm{E}_1^{\top}V Z_{n_0+n_1,1}\right\} \notag \\
&  + \frac{1}{1-\mu_0}\left\{\frac{1+\rho}{\rho} Z_{n_0+n_1,3} + (1+\rho)\bm{E}_2^{\top}V Z_{n_0+n_1,1}\right\}  + o_p((n_0+n_1)^{-1/2}) \notag \\
=& \bm{M}_1^{\top}Z_{n_0+n_1} + o_p((n_0+n_1)^{-1/2}),
\end{align}
where
\bas
\bm{M}_1 = \begin{pmatrix}
     V^{\top}\left[\bm{E}_2/(1-\mu_0) - \bm{E}_1/\mu_0\right] \cdot (1+\rho)\\
    1/\mu_0\cdot (1+\rho)/\rho\\
    1/(1-\mu_0)\cdot (1+\rho)/\rho\\
    0\\
    0
\end{pmatrix}.
\eas
Here, we have applied the asymptotic normality of $\bar{H} = \hat{\mu}_{\mathrm{REG}}$ from Theorem~2 with $h(x, y) = y$, together with Slutsky’s theorem, to replace $\bar{H}$ by $\mu_0$. 

Applying Slutsky's theorem and \eqref{asy-Z}, we obtain
\begin{equation*}
    \sqrt{n_0+n_1}(\widehat{AUC} - {AUC} ) \convergeto \mathcal{N}(0,\sigma_{AUC}^2),
\end{equation*}
with
\begin{equation}
    \sigma_{AUC}^2 = \bm{M}_1^{\top}\Sigma_Z \bm{M}_1.
\end{equation}

\subsection{Proof of the second part of Theorem 3}
\label{proof-bpartthm3}
Throughout this document,  let $f_0$ and $f_1$ denote the probability density functions of $F_0$ and $F_1$, respectively. 

For any fixed $u\in (0,1)$, let $\hat{\tau}_{u} = \widehat{F}_{0}^{-1}(u)$. Then,  
    \begin{align}
    \label{expan-roc}
        \widehat{ROC}(u) - ROC(u) &= \bigl\{1 - \widehat{F}_1(\hat{\tau}_{1-u})\bigr\} 
        - \bigl\{1 - {F}_1({\tau}_{1-u})\bigr\} \notag  \\
        &= \bigl\{F_1(\tau_{1-u}) - {F}_1(\hat{\tau}_{1-u})\bigr\}
        + \bigl\{F_1(\tau_{1-u}) - \widehat{F}_1({\tau}_{1-u})\bigr\} + \varepsilon_{n_0+n_1},
    \end{align}
where 
$$
\varepsilon_{n_0+n_1} = \bigl\{F_1(\hat\tau_{1-u}) - \widehat{F}_1(\hat{\tau}_{1-u})\bigr\}  -\bigl\{F_1(\tau_{1-u}) - \widehat{F}_1({\tau}_{1-u})\bigr\}.
$$

First, we show that $\varepsilon_{n_0 + n_1} = o_p((n_0 + n_1)^{-1/2})$. By Part (b) of Lemma~1,  
\bas
\sup_{t}| \widehat{F}_{0}(t)  -{F}_{0}(t) | = O_p((n_0+n_1)^{-1/2}),
\eas 
which implies  
\begin{equation*}
    \hat\tau_{1-u} - \tau_{1-u} = O_p((n_0+n_1)^{-1/2}).
\end{equation*}
Following arguments similar to  Lemma A.2 of \cite{chen2013quantile} and Theorem~4 in the supplementary material of \cite{chen2021composite}, we obtain
 \begin{equation}
         \label{errroc}
| \varepsilon_{n_0+n_1} | = | \bigl\{F_1(\hat\tau_{1-u}) - \widehat{F}_1(\hat{\tau}_{1-u})\bigr\} - \bigl\{F_1(\tau_{1-u}) - \widehat{F}_1({\tau}_{1-u})\bigr\} |  =  o_p((n_0+n_1)^{-1/2}).
 \end{equation}
Plugging \eqref{errroc} into \eqref{expan-roc} gives
\begin{align}
     \label{reexpan-roc}
        \widehat{ROC}(u) - ROC(u) &= \bigl\{F_1(\tau_{1-u}) - {F}_1(\hat{\tau}_{1-u})\bigr\}  \notag \\
        &\hspace{0.5cm}
        + \bigl\{F_1(\tau_{1-u}) - \widehat{F}_1({\tau}_{1-u})\bigr\}  + o_p((n_0+n_1)^{-1/2}).
\end{align}

Following the similar arguments to Theorem 3.1 in \cite{chen2013quantile}, we have the following Bahadur representation:
\bas
\hat\tau_{1-u} - \tau_{1-u} = 
\frac{1-u - \widehat{F}_0(\tau_{1-u})}{f_0(\tau_{1-u})} + o_p((n_0+n_1)^{-1/2}).
\eas
Applying the Delta method gives
\begin{align*}
    F_1(\tau_{1-u}) -F_1(\hat\tau_{1-u}) &= f_1(\tau_{1-u})(\tau_{1-u} - \hat\tau_{1-u} ) o_p((n_0+n_1)^{-1/2}) \\
&=f_1(\tau_{1-u}) \frac{\widehat{F}_0(\tau_{1-u}) - (1-u)}{f_0(\tau_{1-u})} +  o_p((n_0+n_1)^{-1/2}) \\
&=\frac{f_1(\tau_{1-u})}{f_0(\tau_{1-u})} \bigl\{\widehat{F}_0(\tau_{1-u}) - (1-u)\bigr\} +  o_p((n_0+n_1)^{-1/2}). 
\end{align*}
Combining this with \eqref{reexpan-roc}, we have
    \begin{align}
    \label{reexpan1-roc}
        \widehat{ROC}(u) - ROC(u) =&\frac{f_1(\tau_{1-u})}{f_0(\tau_{1-u})} \bigl\{\widehat{F}_0(\tau_{1-u}) - (1-u)\bigr\} \notag \\
        & +  \bigl\{F_1(\tau_{1-u}) - \widehat{F}_1({\tau}_{1-u})\bigr\} + o_p((n_0+n_1)^{-1/2}).
    \end{align}
    
Plugging  \eqref{e-F0} and   \eqref{e-F1}  into \eqref{reexpan1-roc} gives
\begin{align}
\label{reexpan2-roc}
    &\widehat{ROC}(u) - ROC(u) \notag\\
    &= \frac{1}{1-\bar{H}} \cdot \frac{f_1(\tau_{1-u})}{f_0(\tau_{1-u})} \cdot \frac{1}{n_0}\sum_{j=1}^{n_0}\bigl\{1-H(x_{n_1 + j};\hat{\nu})\bigr\}\bigl\{I(c(x_{n_1 + j}) \leq \tau_{1-u}) - F_0(\tau_{1-u})\bigr\} \notag  \\
    &\hspace{0.2cm}- \frac{1}{\bar{H}} \cdot \frac{1}{n_0}\sum_{j=1}^{n_0}H(x_{n_1 + j};\hat{\nu})\bigl\{I(c(x_{n_1 + j}) \leq \tau_{1-u}) - F_1(\tau_{1-u})\bigr\} + o_p((n_0+n_1)^{-1/2}). 
\end{align}
Recall $\bm{E}_{3}$ and $\bm{E}_4$ defined in \eqref{matrix-e3} and \eqref{matrix-e4}, respectively.
Analogous to the expansion in \eqref{taylor1}, and using the linear approximation of $\hat{\nu} - \nu_0$ in Part (a) of Lemma~\ref{lem-toprovethm3}, 
we have 
\begin{align}
\label{taylor3}
    &\frac{1}{n_0}\sum_{j=1}^{n_0}\bigl\{1-H(x_{n_1 + j};\hat{\nu})\bigr\}\bigl\{I(c_{n_1 + j} \leq \tau_{1-u}) - F_0(\tau_{1-u})\bigr\} \notag \\
    &\hspace{0.2cm}= \frac{1+\rho}{\rho} Z_{n_0+n_1,4} + (1+\rho)\bm{E}_{3}^{\top}V Z_{n_0+n_1,1} + o_p((n_0+n_1)^{-1/2})
\end{align}
and 
\begin{align}
\label{taylor4}
    &\frac{1}{n_0}\sum_{j=1}^{n_0}H(x_{n_1 + j};\hat{\nu})\bigl\{I(c_{n_1 + j} \leq \tau_{1-u}) - F_1(\tau_{1-u})\bigr\} \notag \\
    &\hspace{0.2cm}= \frac{1+\rho}{\rho} Z_{n_0+n_1,5} - (1+\rho)\bm{E}_{4}^{\top} V Z_{n_0+n_1,1}  + o_p((n_0+n_1)^{-1/2}).
\end{align}

Substituting \eqref{taylor3} and \eqref{taylor4} into \eqref{reexpan2-roc}, and following the same procedure as in the derivation of \eqref{final-auc-expan},
we obtain
    \begin{align*}
        \widehat{ROC}(u) - ROC(u) =& \frac{1}{1-\bar{H}} \cdot \frac{f_1(\tau_{1-u})}{f_0(\tau_{1-u})} \left\{\frac{1+\rho}{\rho} Z_{n_0+n_1,4} + (1+\rho)\bm{E}_{3}^{\top}V Z_{n_0+n_1,1}\right\} \notag \\
&- \frac{1}{\bar{H}}  \left\{\frac{1+\rho}{\rho} Z_{n_0+n_1,5} - (1+\rho)\bm{E}_{4}^{\top} V Z_{n_0+n_1,1}\right\} + o_p((n_0+n_1)^{-1/2}) \notag\\
=& \frac{1}{1-\mu_0} \cdot \frac{f_1(\tau_{1-u})}{f_0(\tau_{1-u})} \left\{\frac{1+\rho}{\rho} Z_{n_0+n_1,4} + (1+\rho)\bm{E}_{3}^{\top}V Z_{n_0+n_1,1}\right\} \notag\\
&- \frac{1}{\mu_0} \left\{\frac{1+\rho}{\rho} Z_{n_0+n_1,5} - (1+\rho)\bm{E}_{4}^{\top} V Z_{n_0+n_1,1}\right\} + o_p((n_0+n_1)^{-1/2}) \notag \\
=& \bm{M}_2^{\top} Z_{n_0+n_1} + o_p((n_0+n_1)^{-1/2}),
    \end{align*}
where
\bas
\bm{M}_2 = 
\begin{pmatrix}
     V^{\top}\left\{\frac{f_1(\tau_{1-u})}{f_0(\tau_{1-u})} \bm{E}_3 /(1-\mu_0) + \bm{E}_4/\mu_0\right\} \cdot (1+\rho)\\
     0\\
     0\\
     1/(1-\mu_0)\cdot \frac{f_1(\tau_{1-u})}{f_0(\tau_{1-u})} \cdot (1+\rho)/\rho\\
   - 1/\mu_0\cdot (1+\rho)/\rho
\end{pmatrix}.
\eas 
Applying Slutsky's theorem and \eqref{asy-Z}, we have 
\bas
\sqrt{n_0+n_1}\left\{ \widehat{ROC}(u) - ROC(u)\right\} =
\bm{M}_2^{\top}\cdot \sqrt{n_0+n_1} Z_{n_0+n_1} + o_p(1) \convergeto \mathcal{N}(0,\sigma_{ROC}^2(u)),
\eas
where 
\begin{equation}
    \sigma_{ROC}^2(u) = \bm{M}_2^{\top}\Sigma_{Z}\bm{M}_2.
\end{equation}
This completes the proof of Theorem 3 in the main paper.

\section{Proof of Theorem 4 in the main paper}
\label{proof-thm4}

\subsection{Some preparation}
In this section, we consider a classifier parameterized by an unknown parameter $\nu$, denoted by $c(x;\nu)$. Analogous to the definition of $Z_{n_0+n_1}$ in Section \ref{S4.section1}, define 
\[
\widetilde{Z}_{n_0+n_1} = (\widetilde{Z}_{n_0+n_1,1}^{\top},\allowbreak\ \widetilde{Z}_{n_0+n_1,2},\allowbreak\ \widetilde{Z}_{n_0+n_1,3},\allowbreak\ \widetilde{Z}_{n_0+n_1,4},\allowbreak\ \widetilde{Z}_{n_0+n_1,5})^{\top}, 
\]
where all occurrences of $c(x)$ in the original definition of $Z_{n_0+n_1}$ are replaced by $c(x;\nu_0)$. The components are given by
\begin{align*}
\widetilde{Z}_{n_0+n_1,1} &= \frac{1}{n_0+n_1}  \begin{pmatrix}
    S_{n_0+n_1,1}(\xi_0) \\
    S_{n_0+n_1,2}(\theta_0,\xi_0)
\end{pmatrix}, \\[6pt]
\widetilde{Z}_{n_0+n_1,2} &=  \frac{1}{n_0+n_1}\sum_{i=1}^{n_0+n_1} (1-s_i)H(x_i;\nu_0)\left\{F_0(c(x_i;\nu_0)) - AUC\right\}, \\[6pt]
\widetilde{Z}_{n_0+n_1,3} &=  \frac{1}{n_0+n_1} \sum_{i=1}^{n_0+n_1} (1-s_i)\left\{1-H(x_i;\nu_0)\right\} \left\{1 - F_1(c(x_i;\nu_0)) - AUC\right\}, \\[6pt]
\widetilde{Z}_{n_0+n_1,4} &=  \frac{1}{n_0+n_1} \sum_{i=1}^{n_0+n_1} (1-s_i)\left\{1-H(x_i;\nu_0)\right\} \left\{I(c(x_i;\nu_0)\leq \tau_{1-u}) - F_0(\tau_{1-u})\right\} \\[6pt]
\widetilde{Z}_{n_0+n_1,5} &=  \frac{1}{n_0+n_1} \sum_{i=1}^{n_0+n_1} (1-s_i) H(x_i;\nu_0) \left\{I(c(x_i;\nu_0)\leq \tau_{1-u}) - F_1(\tau_{1-u})\right\}.
\end{align*}
It can also be verified that $$\Ex(\widetilde{Z}_{n_0+n_1}) = 0$$ and 
\begin{equation}
\label{asy-tilZ}
    \sqrt{n_0+n_1} \widetilde{Z}_{n_0+n_1} \convergeto \mathcal{N}(0, \widetilde{\Sigma}_{Z}),
\end{equation}
where $\widetilde{\Sigma}_Z$ denotes the variance and covariance matrix of $\sqrt{n_0+n_1}\widetilde{Z}_{n_0+n_1}$. 
In particular, $\widetilde{\Sigma}_Z$ coincides with $\Sigma_Z$ if all instances of $c(x)$ in 
$\Sigma_Z$ are replaced by $c(x;\nu_0)$.

{
Recall that $f_0$ and $f_1$ denote the probability density functions of $F_0$ and $F_1$ in the main paper, respectively.  
Since the classifier $c(x;\nu)$ depends on the unknown parameter $\nu$, both the distribution and density functions of $c(X;\nu)$ generally depend on $\nu$. Under Condition~8, we therefore write them as $F_y(\cdot;\nu)$ and $f_y(\cdot;\nu)$ for clarity. When $\nu=\nu_0$, however, we use the shorthand notation $F_y$ and $f_y$ without explicitly indicating the dependence on $\nu_0$.
} 
Furthermore, define 
\begin{align}
    &\widetilde{\bm{E}}_{1} = \Ex_{0}\left[\nabla_{\nu}H(X;\nu_0) \bigl\{F_0(c(X;\nu_0)) -AUC\bigr\} \right]   \notag \\
    &\hspace{1cm}  + \Ex_{0}\left[ f_0(c(X;\nu_0))H(X;\nu_0) \nabla_{\nu}c(X;\nu_0) \right], \label{matrix-tile1}\\[6pt]
&\widetilde{\bm{E}}_{2} = \Ex_{0}\left[\nabla_{\nu}H(X;\nu_0) \left\{ 1- F_1(c(X;\nu_0)) - AUC\right\} \right]  \notag \\
    &\hspace{1cm} +  \Ex_{0}\left[
f_1(c(X;\nu_0))\left\{1-H(X;\nu_0)\right\} \nabla_{\nu}c(X;\nu_0)\right],\label{matrix-tile2}\\[6pt]
&\widetilde{\bm{E}}_{3} = \Ex_{0}\left[\nabla_{\nu}H(X;\nu_0) \left\{ I(c(X;\nu_0)\leq\tau_{1-u}) -F_0(\tau_{1-u})\right\} + f_0(\tau_{1-u})\nabla_{\nu}c(X;\nu_0)\right],
\label{matrix-tile3}\\[6pt]
&\widetilde{\bm{E}}_{4} = \Ex_{0}\left[\nabla_{\nu}H(X;\nu_0) \left\{ I(c(X;\nu_0)\leq\tau_{1-u}) -F_1(\tau_{1-u})\right\}- f_1(\tau_{1-u})\nabla_{\nu}c(X;\nu_0) \right]. \label{matrix-tile4}
\end{align}

Given the estimator $\hat{\nu}$ and the resulting estimated classifier $c(x;\hat{\nu})$, 
the target distribution functions $F_0(u)$ and $F_1(u)$ can be estimated by
\begin{equation}
    \label{til-F0}
    \widetilde{F}_{0}(u) = \frac{1}{n_0(1-\bar{H})}\sum_{j=1}^{n_0}\bigl\{1-H(x_{n_1 + j};\hat{\nu})\bigr\}\cdot I(c(x_{n_1 + j};\hat{\nu})\leq u)
\end{equation}
and 
\begin{equation}
    \label{til-F1}
     \widetilde{F}_{1}(u) = \frac{1}{n_0 \bar{H}}\sum_{j=1}^{n_0}H(x_{n_1 + j};\hat{\nu})\cdot I(c(x_{n_1 + j};\hat{\nu}) \leq u),
\end{equation}
respectively.
Compared with $\widehat{F}_0$ and $\widehat{F}_1$ in Theorem 3, $\hat\nu$ appears inside the indicator function, which is nonsmooth. This poses substantial challenges in deriving the asymptotic properties of $\widetilde{F}_0$ and $\widetilde{F}_1$. We address these challenges using advanced empirical process theory. 


{ 
\begin{lemma}
\label{lem-toprovethm4}

Suppose that Conditions~1--6 and 8 in Section~\ref{prelim} hold. 
Then the processes $\sqrt{n_0+n_1}\bigl\{\widetilde{F}_0(u) - F_0(u)\bigr\}$
and
$\sqrt{n_0+n_1}\bigl\{\widetilde{F}_1(t) - F_1(t)\bigr\}$
converge jointly to a tight, mean-zero, bivariate Gaussian process.
\end{lemma}
}

\begin{proof}
{
Let 
\begin{equation*}
    \widetilde{D}_{n_0 1}(t) = \frac{1}{n_0}\sum_{j=1}^{n_0}H(x_{n_1 + j};\hat{\nu})\left\{ I(c(x_{n_1 + j};\hat{\nu}) \leq t) - F_1(t)\right\}.
\end{equation*}
Then it can be verified that
$$
\widetilde{F}_{1}(t) - {F}_{1}(t) = \widetilde{D}_{n_0 1}(t)/\bar{H}.
$$

First, we find an approximation to $\widetilde{D}_{n_0 1}(t)$. 
Note that 
\begin{align*}
\widetilde{D}_{n_0 1}(t) =& \frac{1}{n_0}\sum_{j=1}^{n_0} \left\{ H(x_{n_1 + j};\nu_0) + \nabla_{\nu}^{\top} H(x_{n_1 + j};\nu_0) (\hat{\nu} - \nu_0) \right\}  \left\{I(c(x_{n_1 + j};\hat{\nu}) \leq t) - F_1(t)\right\} \\
&+\tilde{e}_{n_0 1}(t) \\
=&\frac{1}{n_0}\sum_{j=1}^{n_0} H(x_{n_1 + j};\nu_0) \left\{I(c(x_{n_1 + j};\hat{\nu}) \leq t) - F_1(t)\right\}\\
&+\frac{1}{n_0}\sum_{j=1}^{n_0}   \left\{I(c(x_{n_1 + j};\hat{\nu}) \leq t) - I(c(x_{n_1 + j};{\nu_0}) \leq t)\right\} \nabla_{\nu}^{\top} H(x_{n_1 + j};\nu_0) (\hat{\nu} - \nu_0)  \\
&
+ \frac{1}{n_0}\sum_{j=1}^{n_0}   \left\{I(c(x_{n_1 + j};{\nu_0}) \leq t) - F_1(t)\right\} \nabla_{\nu}^{\top} H(x_{n_1 + j};\nu_0) (\hat{\nu} - \nu_0) +\tilde{e}_{n_0 1}(t) \\
=&\frac{1}{n_0}\sum_{j=1}^{n_0} H(x_{n_1 + j};\nu_0) \left\{I(c(x_{n_1 + j};\hat{\nu}) \leq t) - F_1(t)\right\} \\
&+\frac{1}{n_0}\sum_{j=1}^{n_0}   \left\{I(c(x_{n_1 + j};\hat{\nu}) \leq t) - I(c(x_{n_1 + j};{\nu_0}) \leq t)\right\} \nabla_{\nu}^{\top} H(x_{n_1 + j};\nu_0) (\hat{\nu} - \nu_0)  \\
&
+{\bm{K}}^{\top}_{1}(t)(\hat{\nu} - \nu_0) +\tilde{e}_{n_0}(t),
\end{align*} 
where 
\begin{align*}
\tilde{e}_{n_0 1}(u)
&= \frac{1}{n_0}\sum_{j=1}^{n_0} \bigl\{ H(x_{n_1 + j};\hat{\nu}) - H(x_{n_1 + j};\nu_0) - \nabla_{\nu}^{\top} H(x_{n_1 + j};\nu_0) (\hat{\nu} - \nu_0) \bigr\}  \notag  \\ 
&\hspace{1.7cm} \times  \left\{I(c(x_{n_1 + j};\hat{\nu}) \leq t)  -F_1(t)\right\}
\end{align*}
and 
$$
\tilde{e}_{n_0}(t) = \tilde{e}_{n_0 1}(t) - 
\tilde{\bm{e}}_{n_0 2}^{\top}(t)(\hat{\nu} - \nu_0)
$$
with 
$$
\tilde{\bm{e}}_{n_0 2}(t) = \frac{1}{n_0}\sum_{j=1}^{n_0}
        \nabla_{\nu} H(x_{n_1 + j};\nu_0)\left\{I(c(x_{n_1 + j}) \leq t) - F_1(t)\right\} -\bm{K}_{1}(t). 
$$

Similar to the analysis in the proof procedure of Part (b) of Lemma~\ref{lem-toprovethm3}, we obtain
$$
\sup_{u}\|\tilde{\bm{e}}_{n_0 2}(u)\| = O_p((n_0+n_1)^{-1/2})~~\mbox{ and }
~~\sup_{u}|\tilde{{e}}_{n_01}(u)| = o_p((n_0+n_1)^{-1/2}).$$
Therefore, 
\begin{equation}
    \sup_{u}|\tilde{{e}}_{n_0}(u)| = \sup_{u}|\tilde{e}_{n_0 1}(u) - 
\tilde{\bm{e}}_{n_0 2}^{\top}(u)(\hat{\nu} - \nu_0)| =  o_p((n_0+n_1)^{-1/2}).
\label{tilerr}
\end{equation}

To facilitate the proof, we introduce some notation.
For any function $a(x;\nu)$, 
we define 
$$
\mathbb{P}_{n_0} a(X;\nu)
=\frac{1}{n_0} \sum_{j=1}^{n_0}
a(x_{n_1+j};\nu)~~
\mbox{ and }~~
\mathbb{P}_{n_0} a(X;\hat\nu)
=\frac{1}{n_0} \sum_{j=1}^{n_0}
a(x_{n_1+j};\hat\nu).
$$
Further, 
$$
\mathbb{P}_{0} a(X;\nu)
=\int
a(x;\nu) d P^{(0)}(x)
~~
\mbox{ and }
~~
\mathbb{P}_{0} a(X;\hat\nu)
=\int
a(x;\hat\nu) d P^{(0)}(x).
$$
Define the empirical process $\mathbb{G}_{n_0} = \sqrt{n_0}(\mathbb{P}_{n_0} - \mathbb{P}_0)$.  
Then, under these notations and combining \eqref{tilerr}, we can rewrite
    \begin{align}
    \label{keyexp-til}
        \widetilde{D}_{n_0 1}(t) =& 
\frac{1}{n_0}\sum_{j=1}^{n_0} H(x_{n_1 + j};\nu_0) \left\{I(c(x_{n_1 + j};{\nu_0}) \leq t) - F_1(t)\right\} + {\bm{K}}^{\top}_{1}(t)(\hat{\nu} -\nu_0) \notag  \\
& + I_{n_0 1}(t) + I_{n_0 2}(t)+ \bm{I}^{\top}_{n_0 3}(t)(\hat{\nu} -\nu_0) + o_P((n_0+n_1)^{-1/2}),
    \end{align}
where 
\begin{align*}
    I_{n_0 1}(t)&= (\mathbb{P}_{n_0} - \mathbb{P}_0) \left[ \left\{I(c(X;\hat{\nu}) \leq t) - I(c(X;{\nu_0}) \leq t)\right\} H(X;\nu_0)  \right] ,\\
I_{n_0 2}(t)&=  \mathbb{P}_0 \left[\left\{I(c(X;\hat{\nu}) \leq t) - I(c(X;{\nu_0}) \leq t)\right\}  H(X;\nu_0)  \right]  ,\\
\bm{I}_{n_0 3}(t) &=  \frac{1}{n_0}\sum_{j=1}^{n_0}   \left\{I(c(x_{n_1 + j};\hat{\nu}) \leq t) - I(c(x_{n_1 + j};{\nu_0}) \leq t)\right\} \nabla_{\nu} H(x_{n_1 + j};\nu_0)
\end{align*}
and the remainder term is uniform with respect to $t$. Similarly, we can also divide $\bm{I}_{n_0 3}(t)$ into 
$$
\bm{I}_{n_0 3}(t) = \bm{I}^1_{n_0 3}(t) + \bm{I}^2_{n_0 3}(t), 
$$
where 
\begin{align*}
    \bm{I}^{1}_{n_0 3}(t)&= (\mathbb{P}_{n_0} - \mathbb{P}_0) \left[  \left\{I(c(X;\hat{\nu}) \leq t) - I(c(X;{\nu_0}) \leq t)\right\} \nabla_{\nu} H(X;\nu_0)  \right] ,\\
\bm{I}^2_{n_0 3}(t)&=  \mathbb{P}_0\left[ \left\{I(c(X;\hat{\nu}) \leq t) - I(c(X;{\nu_0}) \leq t)\right\} \nabla_{\nu} H(X;\nu_0)  \right].
\end{align*}

We begin by analyzing each term separately, starting with $I_{n_01}(t)$.

In Part~(a) of Condition~8, $\mathbb{R}_{\nu_0} \subset \mathbb{R}^{d}$ denotes a bounded neighborhood of $\nu_{0}$, where $d = d_{\xi}+d_{\theta}$ is the dimension of $\nu$. Moreover, $\mathbb{R}_{c}$ is assumed to be bounded, i.e., contained in a bounded interval of $\mathbb{R}$. Since Section~\ref{proof-thm12} has established the consistency of $\hat{\omega}$, it follows that $\hat{\nu}\to_p\nu_0$, and hence $\hat{\nu}\in\mathbb{R}_{\nu_0}$ with probability approaching one. Throughout the remainder of the proofs, we therefore restrict $\nu$ and $t$ to $\mathbb{R}_{\nu_0}$ and $\mathbb{R}_{c}$, respectively, so that any expression like $\sup_{t}$ should be interpreted as $\sup_{t\in \mathbb{R}_c}$.


Next,  we establish that $\sup_{t}I_{n_0 1}(t) = o_p((n_0+n_1)^{-1/2})$ in three steps. 

\textbf{Step 1.} We first show that the function class  
\begin{equation}
    \label{class-F}
    \mathcal{F}= \left\{f_{\nu t}(x): \nu  \in \mathbb{R}_{\nu_0}, t \in \mathbb{R}_{c}   \right\},
\end{equation}
where 
$$
f_{\nu t}(x) = I(c(x;\nu) \leq t),
$$ is a Donsker class.

Since both $\mathbb{R}_{\nu_0}$ and $\mathbb{R}_{c}$ are bounded subsets of finite-dimensional Euclidean spaces, they are totally bounded. 
Hence, by Lemma~2.7 of \citet{sen2018gentle}, for any $\delta_1, \delta_2\in (0,1)$, the covering numbers satisfy 
\begin{equation}
\label{cnb_nu0}
    N_{1} = N(\delta_1,\mathbb{R}_{\nu_0},\|\cdot\|) \lesssim \bigl(\frac{1}{\delta_1}\bigr)^d
\end{equation}
and 
\begin{equation}
\label{cnb_t}
    N_{2} = N(\delta_2,\mathbb{R}_{c},\|\cdot\|) \lesssim  \frac{1}{\delta_2},
\end{equation}
where ``$\lesssim$" denote smaller than, up to a universal constant.

In particular, there exists a finite set 
$$\left\{ \nu^*_{1},\ldots, \nu^*_{N_{1}} \right\},$$
not necessarily contained in $\mathbb{R}_{\nu_0}$, such that for any $\nu \in \mathbb{R}_{\nu}$, one can find some $j\in \left\{1,\ldots, N_{1}\right\}$ satisfying  
\begin{equation*}
    \|\nu - \nu^*_j\| \leq \delta_1.
\end{equation*}

Similarly, for $\mathbb{R}_{c}$, there exists a finite set 
\[
\{t^*_{1}, \ldots, t^*_{N_{2}}\},
\]  
not necessarily contained in $\mathbb{R}_{c}$, such that for any $t \in \mathbb{R}_{c}$, there exists $j \in \{1,\ldots,N_{2}\}$ with  
\begin{equation*}
    \|t - t^*_{j}\| \leq \delta_2.
\end{equation*}

Define the bracketing functions
\begin{align*}
    \ell_{kj}(x) &= I(c(x;\nu_j^*) \leq t_k^* -\delta_1 L(x)- \delta_2 ),\\
    u_{kj}(x) &= I(c(x;\nu_j^*) \leq t_k^* + \delta_1 L(x)+ \delta_2 ) 
\end{align*}
for $j \in \{1,\ldots,N_{1}\}$ and $k \in \{1,\ldots,N_{2}\}$. 

We verify that $[\ell_{kj}(x),u_{kj}(x)]$ indeed brackets $f_{\nu t}(x)$.
For any fixed  $x\in \mathcal{X}$,  if $\ell_{kj}(x) = 1$, then under the Lipschitz condition in Part (b) of Condition~8, 
we have 
\begin{align*}
    c(x;\nu) &\leq c(x;\nu_j^*) + L(x) \|\nu - \nu_j^*\| \\
    &\leq c(x;\nu_j^*) + \delta_1 L(x) \\
    &\leq t_k^* -\delta_1 L(x) - \delta_2 + \delta_1 L(x) \leq t,
\end{align*}
which implies that
$$f_{\nu t}(x) = I(c(x;\nu) \leq t) = 1.$$
Hence, in this case,  $[\ell_{kj}(x),u_{kj}(x)]$ covers $f_{\nu t}(x)$.

If $u_{kj}(x) = 0$,  then similarly
\begin{align*}
    c(x;\nu) &\geq c(x;\nu_j^*) - L(x) \|\nu - \nu_j^*\| \\
    & > t_k^*+ \delta_2 +  \delta_1 L(x)  - L(x) \|\nu - \nu_j^*\| \\
    & \geq t_k^*  - \delta_2 \geq t,
\end{align*}
so that 
$$f_{\nu t}(x) = I(c(x;\nu) \leq t) = 0.$$
Thus, $[\ell_{kj}(x),u_{kj}(x)]$ again covers $f_{\nu t}(x)$.

Finally, if $\ell_{kj}(x)=0$ and $u_{kj}(x)=1$, the bracket trivially covers $f_{\nu t}(x)$.   

Thus, $\{[\ell_{kj}(x),u_{kj}(x)]\}$ forms a bracketing cover of $\mathcal{F}$.  
Combining \eqref{cnb_nu0} and \eqref{cnb_t}, the bracketing number satisfies
\begin{equation}
\label{cnb_bracket-1}
    N_{[]} = N_{1}\cdot N_{2} \lesssim \bigl(\frac{1}{\delta_1}\bigr)^d\cdot \frac{1}{\delta_2}.
\end{equation}

To quantify the size of each bracket, we evaluate the $L_2(P^{(0)})$ distance between its endpoints.  Specifically, we have
\begin{align}
\label{len-bracket-1}
\|\ell_{kj}(X)-u_{kj}(X)\|_{L_2(P^{(0)})}^2 &= \mathbb{P}_{0}(|c(X;\nu_j^*)-t_{k}^*| \leq \delta_1 L(X) + \delta_2) \notag \\
&= \mathbb{P}_{0}(|c(X;\nu_j^*)-t_{k}^*| \leq \delta_1 L(X) + \delta_2, L(X) \leq L) \notag \\
&\hspace{0.5cm}+  \mathbb{P}_{0}(|c(X;\nu_j^*)-t_{k}^*| \leq \delta_1 L(X) + \delta_2, L(X) > L) \notag \\
&\leq \mathbb{P}_{0}(|c(X;\nu_j^*)-t_{k}^*| \leq \delta_1 L + \delta_2 )  +   \mathbb{P}_{0}( L(X) > L) \notag \\
&\leq \mathbb{P}_{0} (|c(X;\nu_j^*)-t_{k}^*| \leq \delta_1 L + \delta_2 ) + \Ex_{0}(L^{2+\delta}(X))/L^\delta,
\end{align}
where the last inequality follows from Markov’s inequality, with $\delta > 0$ and $L$ is some positive constant.

Finally, under the assumption
\[
\sup_{\nu \in \mathbb{R}_{\nu_0}} \sup_{t \in \mathbb{R}_{c}} f_{y}(t;\nu) \leq M < \infty,
\]
in Part (a) of Condition~8,
the first term on the right-hand side of \eqref{len-bracket-1} can be bounded as 
\begin{align}
    \label{first-term} 
    \mathbb{P}_{0}(|c(X;\nu_j^*)-t_{k}^*| \leq \delta_1 L  + \delta_2 ) 
    &=\mathbb{P}_{0}(t_{k}^* - \delta_1 L - \delta_2 \leq c(X;\nu_j^*) \leq t_{k}^* + \delta_1 L + \delta_2 ) \notag  \\
    &= \int_{t_{k}^* - \delta_1 L - \delta_2}^{t_{k}^* + \delta_1 L + \delta_2} f(z;\nu_j^*)\,dz \notag \\
    &\leq \int_{t_{k}^* - \delta_1 L - \delta_2}^{t_{k}^* + \delta_1 L + \delta_2} \sup_{z\in \mathbb{R}_{c}} f(z;\nu_j^*)\,dz \notag \\
    &\leq 2M(\delta_1 L + \delta_2),
\end{align}
where $f(z;\nu_j^*) = \mu_0 f_{1}(z;\nu_j^*) + (1-\mu_0)f_{0}(z;\nu_j^*). $

Substituting \eqref{first-term} into \eqref{len-bracket-1}, we obtain
\begin{equation}
\label{len-bracket-2}
    \|\ell_{kj}(X)-u_{kj}(X)\|_{L_2(P^{(0)})}^2 \leq 2M(\delta_1 L + \delta_2) + \Ex_{0}(L^{2+\delta}(X))/L^{2+\delta}.
\end{equation}

For any $\varepsilon >0$, choose 
\begin{equation}
\label{choice-par}
    L = \bigl[\frac{2\Ex_{0}(L^{2+\delta}(X)) }{\varepsilon^2}\bigr]^{1/(2 + \delta)},\quad\delta_1 = \frac{\varepsilon^2}{8ML},\quad \delta_2 = \frac{\varepsilon^2}{8M}. 
\end{equation}
Then \eqref{len-bracket-2} implies
\begin{equation*}
     \|\ell_{kj}(X)-u_{kj}(X)\|_{L_2(P^{(0)})} \leq \varepsilon. 
\end{equation*}

Therefore, combining \eqref{cnb_bracket-1} with \eqref{choice-par}, the $\varepsilon$-bracketing cover number of $\mathcal{F}$ satisfies
\begin{align}
\label{cnb_bracket-2}
    N_{[]}= N_{[]}(\varepsilon,\mathcal{F},L_2(P^{(0)}))  
     &\lesssim \varepsilon^{-\bigl(2 + 2d + \tfrac{2d}{2+\delta}\bigr)}. 
\end{align}

Finally, using \eqref{cnb_bracket-2}, a direct calculation shows
\begin{align}
\label{entropy-bnd}
    \int_{0}^{1} \sqrt{\log(N_{[]}(\varepsilon,\mathcal{F},L_2(P^{(0)})))} d\varepsilon &= \sqrt{2 + 2d + \frac{2d}{2+\delta}}\cdot \int_{0}^1 \sqrt{\log(\frac{1}{\varepsilon})}d\varepsilon \notag \\
    &= \sqrt{2 + 2d + \frac{2d}{2+\delta}}\cdot \int_{0}^1 z^{1/2}e^{-1}dz \notag \\
    &= \sqrt{2 + 2d + \frac{2d}{2+\delta}}\cdot \Gamma(\frac{3}{2}) < \infty.
\end{align}
Equation~\eqref{entropy-bnd}, together with the fact that the constant function $1$ is a square-integrable envelope function of $\mathcal{F}$, implies (by Theorem~11.3 in \cite{sen2018gentle}) that $\mathcal{F}$ is a Donsker class.

\textbf{Step 2.} 
Define 
\begin{equation*}
    \mathcal{F}_1 = \left\{x \mathrel{\mapsto}
H(x;\nu_0) I(c(x;\nu) \leq t) : \nu \in \mathbb{R}_{\nu_0},t \in \mathbb{R}_{c}\right\}. 
\end{equation*}
By Example~2.10.10 of \citet{van1996weak}, the class $\mathcal{F}_1$ is also Donsker, since $H(x;\nu_0) \in [0,1]$ uniformly.

\textbf{Step 3.} 
Recall that $H(x;\nu_0) = p^{(0)}(Y=1\mid x) \in [0,1]$. Then
\begin{align}
\label{squ-ineq-1}
&\sup_{t} \int H^2(x;\nu_0)\left\{ I(c(x;\hat{\nu}) \leq t) - I(c(x;{\nu_0}) \leq t)  \right\}^2 dP^{(0)}(x) \notag \\
&\leq \sup_{t} \int  |I(c(X;\hat{\nu}) \leq t) - I(c(X;{\nu_0}) \leq t) | dP^{(0)}(x) \notag\\
&= \sup_{t} \int I(\min\{c(X;\hat{\nu}),c(X;\nu_0)\} \leq t < \max\{c(X;\hat{\nu}),c(X;\nu_0)\}) dP^{(0)}(x) \notag\\
&\leq \sup_{t} \int I(t-L(X)\|\hat{\nu} - \nu_0\| \leq c(X;\nu_0)  < t+L(X)\|\hat{\nu} - \nu_0\| )dP^{(0)}(x),
\end{align}
where the last inequality follows from the fact that if $t$ lies between  $\min\{c(x;\hat{\nu}),c(x;\nu_0)\} $ and $\max\{c(x;\hat{\nu}),c(x;\nu_0)\}$, then
\begin{equation*}
 |t - c(x;\nu_0) | \leq    |c(x;\hat{\nu})-c(x;\nu_0)| \leq L(x) \|\hat\nu -\nu_0\|. 
\end{equation*}

Let $Z_{cL}$, $Z_{c|L}$, and $Z_L$ denote, respectively, the joint cumulative distribution function (cdf) of $\big(c(X;\nu_0), L(X)\big)$ under the target population, the conditional cdf of $c(X;\nu_0)$ given $L(X)$ under the target population, and the marginal cdf of $L(X)$ under the target population.
Under Part (b) of Condition~8, 
applying the mean value theorem, 
\begin{align}
\label{squ-ineq-2}
&\sup_{t} \int I(t-L(X)\|\hat{\nu} - \nu_0\| \leq c(X;\nu_0)  < t+L(X)\|\hat{\nu} - \nu_0\|)dP^{(0)}(x) \notag \\
&= \sup_{t} \int \left\{ Z_{c|L}(t+u  \|\hat{\nu} - \nu_0\|) - Z_{c|L}(t-u  \|\hat{\nu} - \nu_0\|)\right\} dZ_L(u) 
\notag\\
&= 2\|\hat{\nu} - \nu_0\| \cdot  \sup_{t} \int z_{cL}(t_{u}, u)\cdot u d u   \notag \\
&\leq  2\|\hat{\nu} - \nu_0\| \cdot    \int \sup_{t} z_{cL}(t_{u}, u)\cdot u du \notag \\
&= 2\|\hat{\nu} - \nu_0\| \cdot    \int  z^{*}(u) udu \convergepto 0,
\end{align}
where $t_{u}$ lies between $t-u\|\hat\nu-\nu_0\|$ and $t+u\|\hat\nu-\nu_0\|$.

Combining these three steps, Theorem~2.1 of \citet{van2007empirical}, together with Lemma~19.24 of \citet{van2000asymptotic}, implies
\begin{equation*}
   \sup_{t}|\mathbb{G}_{n_0}\left[ \left\{I(c(X;\hat{\nu}) \leq t) - I(c(X;{\nu_0}) \leq t)\right\} H(X;\nu_0) \right]|  \convergepto 0.
\end{equation*}
Hence, 
\begin{equation}
\label{Im1-final}
    \sup_{t}|I_{n_0 1}(t)| =  o_p(n_0^{-1/2}) = o_p((n_0+n_1)^{-1/2}).
\end{equation}

Next, we turn to the analysis of $I_{n_0 2}(t)$.  Recall that $H(x;\nu_0) = p^{(0)}(Y=1 \mid x)$.  
Then,
\begin{align*}
     I_{n_0 2}(t) &= \mu_0 \mathbb{P}_{0}\left\{ I(c(X;\hat{\nu}) \leq t) - I(c(X;{\nu_0}) \leq t) |Y=1\right\} \\
     &= \mu_0 \left\{F_{1}(t;\hat{\nu}) - F_{1}(t;\nu_0)\right\}.  
\end{align*}

Applying a first-order Taylor expansion around $\nu_0$, we obtain   
\begin{align}
\label{Im2-1}
     I_{n_0 2}(t) 
     &=  \mu_0 \nabla^{\top}_{\nu} F_{1}(t;\nu_0) (\hat{\nu} - \nu_0) + \mu_0  (\hat{\nu} - \nu_0) ^{\top}\nabla_{\nu\nu^{\top}}F_{1}(t;\tilde{\nu})  (\hat{\nu} - \nu_0),
\end{align}
where $\tilde{\nu}$ lies between $\hat{\nu}$ and $\nu_0$.

By Part (c) of Condition~8, 
and using the result in Part~(a) of Lemma~\ref{lem-toprovethm3}, it follows that
\begin{equation}
    \label{err-Im2}
    (\hat{\nu} - \nu_0) ^{\top}\nabla_{\nu\nu^{\top}}F_{1}(t;\tilde{\nu})  (\hat{\nu} - \nu_0) = o_p((n_0 + n_1)^{-1/2}). 
\end{equation}

Substituting \eqref{err-Im2} into \eqref{Im2-1}, we conclude that
\begin{align}
\label{Im2-final}
      I_{n_0 2}(t) 
      =  \mu_0 \nabla^{\top}_{\nu} F_{1}(t;\nu_0)(\hat{\nu} - \nu_0) 
      + o_p\bigl((n_0+n_1)^{-1/2}\bigr),
\end{align}
uniformly in $t$.  

Finally, combining \eqref{Im1-final} and \eqref{Im2-final}, we obtain
\begin{align*}
    &\frac{1}{n_0}\sum_{j=1}^{n_0} H(x_{n_1 + j};\nu_0)\Bigl\{I(c(x_{n_1 + j};\hat{\nu}) \leq t) - F_1(t)\Bigr\} \notag \\
    &= I_{n_0 1}(t) + I_{n_0 2}(t) \notag \\
    &= \mu_0 \nabla^{\top}_{\nu} F_{1}(t;\nu_0)(\hat{\nu} - \nu_0) 
    + o_p\bigl((n_0+n_1)^{-1/2}\bigr).
\end{align*}

Next, we analyze $\bm{I}_{n_0 3}(t).$ 
To bound $\sup_{t} |\bm{I}_{n_0 3}^{\,1}(t)|$, we proceed in two steps, following arguments analogous to those used for $\sup_{t} |I_{n_0 1}(t)|$.

\textbf{Step 1. }
Define
\bas
\mathcal{F}_{2} = \left\{x \mathrel{\mapsto}
\nabla_{\nu}H(x;\nu_0) I(c(x;\nu) \leq u):\nu \in \mathbb{R}_{\nu_0},t \in \mathbb{R}_{c}
\right\}. 
\eas 
Recall that the class $\mathcal{F}$ in \eqref{class-F} is Donsker. Since $\nabla_{\nu}H(x;\nu_0)$ is square-integrable, by Example~2.10.10 of \citet{van1996weak} (multiplication by a bounded/square-integrable factor preserves Donsker under an $L_2$-envelope), the class $\mathcal{F}_2$ is also Donsker.

\textbf{Step 2.}
By Hölder’s inequality,
\begin{align}
\label{squ-ineq-3}
&\sup_{t} \int \|\nabla_{\nu}H(x;\nu_0)\|^2\left\{ I(c(x;\hat{\nu}) \leq t) - I(c(x;{\nu_0}) \leq t)  \right\}^2 dP^{(0)}(x) \notag \\
&\leq \left[\Ex_0 \left\{  \|\nabla_{\nu}H(X;\nu_0)\|^{2+\delta} \right\}\right]^{2/(2+\delta)}
\cdot \sup_{t}  \left[\mathbb{P}_{0} \left\{  |I(c(X;\hat{\nu}) \leq t) - I(c(X;{\nu_0}) \leq t)| \right\}\right]^{\delta/(2+\delta)}  \notag \\
&\lesssim  \sup_{t}  \left[\mathbb{P}_{0} \left\{  |I(c(X;\hat{\nu}) \leq t) - I(c(X;{\nu_0}) \leq t)| \right\}\right]^{\delta/(2+\delta)},
\end{align}
where the last inequality holds since  $\Ex_0 \left\{  \|\nabla_{\nu}H(x;\nu_0) \|^{2+\delta}\right\} < \infty$ under Condition 6.

Furthermore, by \eqref{squ-ineq-1} and \eqref{squ-ineq-2}, it follows that 
\begin{equation}
\label{idx-conv}
    \sup_{t} \int |I(c(X;\hat{\nu}) \leq t) - I(c(X;{\nu_0}) \leq t) | dP^{(0)}(x)  \convergepto 0. 
\end{equation}

Combining \eqref{squ-ineq-3} with \eqref{idx-conv}, we conclude that 
\begin{equation}
\label{conv-secm}
     \sup_{t} \int  \|\nabla_{\nu}H(x;\nu_0)\|^2\left\{ I(c(x;\hat{\nu}) \leq t) - I(c(x;{\nu_0}) \leq t)  \right\}^2 dP^{(0)}(x) \convergepto 0. 
\end{equation}

Therefore, by an argument parallel to Step~1 for $I_{n_0 1}(t)$, we finally obtain
\begin{align}
\label{Im31-final}
    \sup_{t} \bm{I}^{1}_{n_0 3}(t) = o_p(n_0^{-1/2}) = o_p((n_0+n_1)^{-1/2}).
\end{align}

Next, we turn to the analysis of $\bm{I}^2_{n_0 3}(t)$.  Observe that
\begin{align*}
    \|\bm{I}^2_{n_0 3}(t)\|^2 &\leq \mathbb{P}_0\left[ \|\nabla_{\nu} H(X;\nu_0) \|^2 \left\{I(c(X;\hat{\nu}) \leq t) - I(c(X;{\nu_0}) \leq t)\right\}^2  \right]. 
\end{align*}
The right-hand side coincides with the convergence bound in~\eqref{conv-secm}, and hence
\begin{equation}
    \label{Im32-final}
\sup_{t}\bm{I}^{2}_{n_0 3}(t) 
= o_p(1). 
\end{equation}

Therefore, combining \eqref{Im31-final} and \eqref{Im32-final}, we obtain
\begin{equation}
     \bm{I}_{n_0 3}(t) = \bm{I}^{1}_{n_0 3}(t)  + \bm{I}^{2}_{n_0 3}(t)  = o_p(1),
\end{equation}
 uniformly in $t$.  Consequently, we deduce that
\begin{equation}
    \label{Im3-final}
    \bm{I}_{n_0 3}^{\top}(t) (\hat{\nu} -\nu_0) 
    = o_p((n_0+n_1)^{-1/2}).
\end{equation}

Collecting the results in \eqref{Im1-final}, \eqref{Im2-final}, and \eqref{Im3-final}, the expansion \eqref{keyexp-til} can be rewritten as
\begin{align*}
        \widetilde{D}_{n_0 1}(t) =& 
\frac{1}{n_0}\sum_{j=1}^{n_0} H(x_{n_1 + j};\nu_0) \left\{I(c(x_{n_1 + j};\nu_0) \leq t) - F_1(t)\right\} + {\bm{K}}^{\top}_{1}(t)(\hat{\nu} -\nu_0) \notag \\
&+ \mu_0 \nabla^{\top}_{\nu} F_{1}(t;\nu_0) (\hat{\nu} - \nu_0) + o_P((n_0+n_1)^{-1/2})  \notag\\
:=& \frac{1}{n_0}\sum_{j=1}^{n_0} H(x_{n_1 + j};\nu_0) \left\{I(c(x_{n_1 + j};\hat{\nu}) \leq t) - F_1(t)\right\} \notag \\
&  + \widetilde{\bm{K}}^{\top}_{1}(t)(\hat{\nu} -\nu_0) + o_P((n_0+n_1)^{-1/2}),
\end{align*}
where 
\begin{align*}
     \widetilde{\bm{K}}_{1}(t) &= {\bm{K}}_{1}(t)  + \mu_{0} \nabla_{\nu} F_{1}(t;\nu_0) \notag \\
        &= \Ex_{0}\left[\nabla_{\nu}H(X;\nu_0)\left\{I(c(X;\nu_0) \leq t) -F_1(t) \right\} \right] + \mu_{0} \nabla_{\nu} F_{1}(t;\nu_0).
\end{align*}

In Section~\ref{proof-thm3}, we have shown that 
$$
\sup_t \|{\bm{K}}_{k}(t)\| < \infty,k=0,1.
$$
Similarly, under the additional assumption in Condition 8 that
$$
\sup_{t}\|\nabla_{\nu}F_1(t;\nu_0) \|< \infty,
$$
for the modified $\widetilde{\bm{K}}_{1}(t)$, we also have
\begin{equation*}
 \sup_t \|\widetilde{\bm{K}}_{1}(t)\| < \infty,    
\end{equation*}
and consequently,
\begin{equation}
\label{sup1}
    \sup_t |\widetilde{\bm{K}}_{1}^{\top}(t) (\hat{\nu} - \nu_0)|  = O_p((n_0+n_1)^{-1/2})
\end{equation}
due to the uniform boundedness of $\nabla_{\nu}F_1(t;\nu_0)$ and the consistency of $\hat{\nu}$.

Moreover, Section~\ref{proof-thm3} also established that 
\begin{equation}
\label{sup2}
    \sup_{t}|\frac{1}{n_0}\sum_{j=1}^{n_0} H(x_{n_1 + j};\nu_0) \left[I(c(x_{n_1 + j};\nu_0) \leq t) - F_1(t)\right]| = O_p((n_0+n_1)^{-1/2}).
\end{equation}

Therefore, combining \eqref{sup1} and \eqref{sup2}, we conclude that 
\begin{equation*}
    \sup_{t}|\widetilde{D}_{n_0 1}(t)| = O_p((n_0+n_1)^{-1/2}).
\end{equation*}

Following a similar proof procedure as in Section~\ref{S4.section1}, and omitting intermediate technical details for brevity, we obtain the expansion
    \begin{align}
    \label{final-tilf1}
    \widetilde{F}_{1}(t) - F_1(t) =& \frac{1}{n_0+n_1}
\sum_{i=1}^{n_0+n_1}  \left[
\frac{1+\rho}{\mu_0 \rho}(1-s_i)
H(x_i;\nu_0)\left\{I(c(x_i;\nu_0) \leq t) - F_1(t)\right\} \right. \notag  \\
&\left.
- \frac{1+\rho}{\mu_0}\widetilde{\bm{K}}^{\top}_{1}(t)V T(x_i,y_i,s_i;\nu_0) \right] +  o_p((n_0+n_1)^{-1/2})\notag \\
:=&   \frac{1}{n_0+n_1} 
\sum_{i=1}^{n_0+n_1} \tilde{L}_{1}(x_i,y_i,s_i;t) + o_p((n_0+n_1)^{-1/2}).
    \end{align}
Similarly, 
   \begin{align}
          \label{final-tilf0}
    \widetilde{F}_{0}(u) - F_0(u) =& \frac{1}{n_0+n_1}
\sum\limits_{i=1}^{n_0+n_1}  \left[
\frac{1+\rho}{(1-\mu_0) \rho}(1-s_i)
(1-H(x_i;\nu_0))\left\{I(c(x_i;\nu_0) \leq u) - F_0(u)\right\} \right. \notag  \\
& \left.
+ \frac{1+\rho}{1-\mu_0}\widetilde{\bm{K}}^{\top}_{0}(u)V T(x_i,y_i,s_i;\nu_0) \right] +  o_p((n_0+n_1)^{-1/2}) \notag \\
:=& \frac{1}{n_0+n_1}
\sum\limits_{i=1}^{n_0+n_1} \tilde{L}_{0}(x_i,y_i,s_i;u) + o_p((n_0+n_1)^{-1/2}), 
    \end{align}
where
\begin{align*}
        \widetilde{\bm{K}}_{0}(u) &= {\bm{K}}_{0}(k) + (1-\mu_0) \nabla_{\nu} F_{0}(u;\nu_0)\notag \\
        &= \Ex_{0}\left[ \nabla_{\nu}H(X;\nu_0)\left\{I(c(X;\nu_0) \leq u) -F_0(u) \right\}  \right] + (1-\mu_0) \nabla_{\nu} F_{0}(u;\nu_0).
    \end{align*}

}

The asymptotic approximations in \eqref{final-tilf1} and \eqref{final-tilf0} imply that both $\widetilde{F}_{1}(t)$
and $\widetilde{F}_{0}(u)$ are asymptotically
linear with influence functions being $\tilde{L}_1(x, y,s;t)$ and $\tilde{L}_0(x,y,s;u)$, respectively.
With the above discussion, the function class of $ \tilde{L}_1(x,y,s;t) $ (indexed by $t$) and the class of $\tilde{L}_0(x,y,s;u)$ (indexed by $u$) are both Donsker classes.  Consequently, the empirical processes  $\sqrt{n_0+n_1} \{ \widetilde{F}_0(u) - F_0(u) \}$
and  $\sqrt{n_0+n_1} \{ \widetilde{F}_1(t) - F_1(t) \}$
converge jointly to a bivariate, mean-zero, Gaussian process with right-continuous sample paths.
 \end{proof}

\subsection{Proof of the first part of Theorem 4}
Similar to \eqref{expan-auc0}, we have 
\begin{align*}
     \widetilde{AUC} - {AUC} &= - \int \bigl\{\widetilde{F}_{1}(u) - F_1(u)\bigr\}dF_0(u) \notag  \\
        &\hspace{0.5cm} + \int \bigl\{\widetilde{F}_{0}(u) - F_0(u)\bigr\}dF_1(u) 
        + o_p((n_0+n_1)^{-1/2}).
\end{align*}
With \eqref{til-F0}--\eqref{til-F1} and straightforward algebra, we further get 
    \begin{align}
     \label{expan-tilauc}
     \widetilde{AUC} - {AUC} =& \bar{H}^{-1}\cdot \frac{1}{n_0}\sum_{j=1}^{n_0}H(x_{n_1 + j};\hat{\nu})\left\{F_0(c(x_{n_1 + j};\hat{\nu})) - AUC\right\} \notag    \\
     &+(1-\bar{H})^{-1}\cdot \frac{1}{n_0}\sum_{j=1}^{n_0}\left\{1-H(x_{n_1 + j};\hat{\nu}) \right\}\left\{1-F_1(c(x_{n_1 + j};\hat{\nu})) - AUC\right\} \notag \\
     &+ o_p((n_0+n_1)^{-1/2}).
    \end{align}
Under Condition~6, and recalling the definitions of $\widetilde{\bm{E}}_{1}$ and $\widetilde{\bm{E}}_{2}$ in \eqref{matrix-tile1} and \eqref{matrix-tile2}, respectively, a first-order Taylor expansion yields
    \begin{align}
        \label{expan-tilauc1}
 & \frac{1}{n_0}\sum_{j=1}^{n_0}H(x_{n_1 + j};\hat{\nu})\left\{F_0(c(x_{n_1 + j};\hat{\nu})) - AUC\right\} \notag  \\
  &= \frac{1+\rho}{\rho}\cdot \frac{1}{n_0+n_1} \sum_{i=1}^{n_0+n_1}(1-s_i)H(x_i;\nu_0)
  \left\{F_0(c(x_{i};{\nu_0})) - AUC\right\} \notag \\
  &\hspace{0.5cm}+\widetilde{\bm{E}}_{1}^{\top}(\hat{\nu} - \nu_0) +  o_p((n_0+n_1)^{-1/2})  \notag  \\
  &= \frac{1+\rho}{\rho}\cdot 
  \widetilde{Z}_{n_0+n_1,2}  +\widetilde{\bm{E}}_{1}^{\top}(\hat{\nu} - \nu_0) +  o_p((n_0+n_1)^{-1/2}).
    \end{align}
Similarly,
    \begin{align}
        \label{expan-tilauc2}
 & \frac{1}{n_0}\sum_{j=1}^{n_0}\left\{1-H(x_{n_1 + j};\hat{\nu})\right\}\left\{1 -F_1(c(x_{n_1 + j};\hat{\nu})) - AUC\right\} \notag \\
  &= \frac{1+\rho}{\rho}\cdot \frac{1}{n_0+n_1} \sum_{i=1}^{n_0+n_1}(1-s_i)\left\{ 1-H(x_i;\nu_0) \right\}
  \left\{1-F_1(c(x_{i};{\nu_0})) - AUC\right\}\notag \\
  &\hspace{0.5cm}-\widetilde{\bm{E}}_{2}^{\top}(\hat{\nu} - \nu_0) +  o_p((n_0+n_1)^{-1/2}) \notag  \\
  &= \frac{1+\rho}{\rho}\cdot 
  \widetilde{Z}_{n_0+n_1,3}  -\widetilde{\bm{E}}_{2}^{\top}(\hat{\nu} - \nu_0) +  o_p((n_0+n_1)^{-1/2}).
    \end{align}

Plugging \eqref{expan-tilauc1} and \eqref{expan-tilauc2} into \eqref{expan-tilauc} and using the result from Part (a) of Lemma~\ref{lem-toprovethm3}, we further obtain 
\begin{align*}
      \widetilde{AUC} - {AUC} =& \frac{1}{\bar{H}}\left\{ 
        \frac{1+\rho}{\rho}\cdot 
  \widetilde{Z}_{n_0+n_1,2}  -(1+\rho)\widetilde{\bm{E}}_{1}^{\top}V{Z}_{n_0+n_1,1}
        \right\} \notag\\
        & + \frac{1}{1-\bar{H}}
        \left\{ 
        \frac{1+\rho}{\rho}\cdot 
  \widetilde{Z}_{n_0+n_1,3}  +(1+\rho)\widetilde{\bm{E}}_{2}^{\top}V{Z}_{n_0+n_1,1}
        \right\}  + o_p((n_0+n_1)^{-1/2}) \notag\\
        =& \frac{1}{\mu_0}\left\{ 
        \frac{1+\rho}{\rho}\cdot 
  \widetilde{Z}_{n_0+n_1,2}  -(1+\rho)\widetilde{\bm{E}}_{1}^{\top}V{Z}_{n_0+n_1,1}
        \right\} \notag \\
        & + \frac{1}{1-\mu_0}
        \left\{ 
        \frac{1+\rho}{\rho}\cdot 
  \widetilde{Z}_{n_0+n_1,3} +(1+\rho)\widetilde{\bm{E}}_{2}^{\top}V{Z}_{n_0+n_1,1}
        \right\}  + o_p((n_0+n_1)^{-1/2}) \notag\\
        =& \widetilde{\bm{M}}_{1}^{\top}\widetilde{Z}_{n_0+n_1}  + o_p((n_0+n_1)^{-1/2}),
\end{align*}
where
\begin{equation*}
    \widetilde{\bm{M}}_{1}  =  \begin{pmatrix}
     V^{\top}\left\{\widetilde{\bm{E}}_2/(1-\mu_0) - \widetilde{\bm{E}}_1/\mu_0\right\} \cdot (1+\rho)\\
    1/\mu_0\cdot (1+\rho)/\rho\\
    1/(1-\mu_0)\cdot (1+\rho)/\rho\\
    0\\
    0
\end{pmatrix}.
\end{equation*}
Applying Slutsky's theorem and \eqref{asy-tilZ}, we have
\begin{equation*}
    \sqrt{n_0+n_1}(\widetilde{AUC} - {AUC} ) \convergeto \mathcal{N}(0,\tilde{\sigma}_{AUC}^2),
\end{equation*}
with 
\begin{equation}
    \tilde{\sigma}_{AUC}^2 = \widetilde{\bm{M}}_1^{\top}\widetilde{\Sigma}_Z \widetilde{\bm{M}}_1.
\end{equation}

\subsection{Proof of the second part of Theorem 4}
Similarly to~\eqref{reexpan1-roc}, we obtain 
    \begin{align}
    \label{reexpan1-tilroc}
        \widetilde{ROC}(u) - ROC(u) =&\frac{f_1(\tau_{1-u})}{f_0(\tau_{1-u})} \bigl\{\widetilde{F}_0(\tau_{1-u}) - (1-u)\bigr\} \notag \\
        & +  \bigl\{F_1(\tau_{1-u}) - \widetilde{F}_1({\tau}_{1-u})\bigr\} + o_p((n_0+n_1)^{-1/2}).
\end{align}

Recall the definitions of $\widetilde{\bm{E}}_{3}$ and $\widetilde{\bm{E}}_{4}$ in \eqref{matrix-tile3} and \eqref{matrix-tile4}, respectively. 
For fixed $\tau_{1-u}$, we have 
\begin{equation}
\label{expan-tilF0}
    \widetilde{F}_0(\tau_{1-u}) - F_0(\tau_{1-u}) = \frac{\widetilde{D}_{n_0 0}(\tau_{1-u})}{1- \bar{H}},
\end{equation}
where 
\begin{align}
\label{Dm0-til}
  \widetilde{D}_{n_0 0}(\tau_{1-u}) =& 
  \frac{1}{n_0}\sum_{j=1}^{n_0} (1-H(x_{n_1 + j};{\nu_0}) ) \left\{ I(c(x_{n_1 + j};\nu_0) \leq \tau_{1-u} ) - F_0(\tau_{1-u})\right\} \notag  \\
  & - \widetilde{\bm{K}}^{\top}_{0}(\tau_{1-u})(\hat{\nu} -\nu_0) + o_P((n_0+n_1)^{-1/2}) \notag \\
  =& \frac{1+\rho}{\rho}\cdot \widetilde{Z}_{n_0+n_1,4}  - \widetilde{\bm{E}}_3^{\top}(\hat{\nu} - \nu_0) + 
  o_P((n_0+n_1)^{-1/2}). 
\end{align}
Similarly, 
\begin{equation}
\label{expan-tilF1}
    \widetilde{F}_1(\tau_{1-u}) - F_1(\tau_{1-u}) = \frac{\widetilde{D}_{n_0 1}(\tau_{1-u})}{1- \bar{H}},
\end{equation}
where 
\begin{align}
\label{Dm1-til}
  \widetilde{D}_{n_0 1}(\tau_{1-u}) =& 
  \frac{1}{n_0}\sum_{j=1}^{n_0} H(x_{n_1 + j};\nu_0) \left\{I(c(x_{n_1 + j};\nu_0) \leq \tau_{1-u} ) - F_1(\tau_{1-u})\right\} \notag \\
  &\ + \widetilde{\bm{K}}^{\top}_{1}(\tau_{1-u})(\hat{\nu} -\nu_0) + o_P((n_0+n_1)^{-1/2}) \notag \\
  =& \frac{1+\rho}{\rho}\cdot \widetilde{Z}_{n_0+n_1,5} + \widetilde{\bm{E}}_4^{\top}(\hat{\nu} - \nu_0) + 
  o_P((n_0+n_1)^{-1/2}).
\end{align}

Considering that $1-u = F_0(\tau_{1-u})$, plugging~\eqref{expan-tilF0} and~\eqref{expan-tilF1} into \eqref{reexpan1-tilroc}, we have 
\begin{equation}
    \label{final0-tilroc}
    \widetilde{ROC}(u) - ROC(u)  = \frac{f_1(\tau(1-u))}{f_0(\tau(1-u))}\cdot \frac{\widetilde{D}_{n_0 0}(\tau_{1-u})}{1-\bar{H}}  - \frac{\widetilde{D}_{n_0 1}(\tau_{1-u})}{\bar{H}} + o_P((n_0+n_1)^{-1/2}). 
\end{equation}
Replacing $\bar{H}$ with $\mu_0$ in~\eqref{final0-tilroc}, we obtain
\begin{equation}
\label{final1-tilroc}
\widetilde{ROC}(u) - ROC(u) 
    = \frac{f_1(\tau(1-u))}{f_0(\tau(1-u))}\cdot \frac{\widetilde{D}_{n_0 0}(\tau_{1-u})}{1-\mu_0}  - \frac{\widetilde{D}_{n_0 1}(\tau_{1-u})}{\mu_0} + o_P((n_0+n_1)^{-1/2}).
\end{equation}

Substituting $\hat{\nu} - \nu_0 = -(1+\rho)VZ_{n_0+n_1,1}$ into \eqref{Dm0-til} and \eqref{Dm1-til}, and then plugging the results into \eqref{final1-tilroc}, we get
\begin{align*}
\widetilde{ROC}(u) - ROC(u) 
    =&\frac{f_1(\tau(1-u))}{f_0(\tau(1-u))}\cdot \frac{1}{1-\mu_0} \left\{
    \frac{1+\rho}{\rho}\cdot \widetilde{Z}_{n_0+n_1,4}  +(1+\rho) \widetilde{\bm{E}}_3^{\top}VZ_{n_0+n_1,1}
    \right\}  \notag \\
    & - \frac{1}{\mu_0}\left\{
    \frac{1+\rho}{\rho}\cdot \widetilde{Z}_{n_0+n_1,5} -(1+\rho) \widetilde{\bm{E}}_4^{\top}VZ_{n_0+n_1,1}
    \right\} + o_P((n_0+n_1)^{-1/2}) \notag \\
    =&\widetilde{\bm{M}}_{2}^{\top}\widetilde{Z}_{n_0+n_1} + + o_P((n_0+n_1)^{-1/2}),
\end{align*} 
where 
\bas
\widetilde{\bm{M}}_2 = 
\begin{pmatrix}
     V^{\top}\left[\frac{f_1(\tau_{1-u})}{f_0(\tau_{1-u})} \widetilde{\bm{E}}_3 /(1-\mu_0) + \widetilde{\bm{E}}_4/\mu_0\right] \cdot (1+\rho)\\
     0\\
     0\\
     1/(1-\mu_0)\cdot \frac{f_1(\tau_{1-u})}{f_0(\tau_{1-u})} \cdot (1+\rho)/\rho\\
   - 1/\mu_0\cdot (1+\rho)/\rho
\end{pmatrix}.
\eas 
Applying Slutsky’s theorem and \eqref{asy-tilZ}, we get 
\bas
\sqrt{n_0+n_1}( \widetilde{ROC}(u) - ROC(u)) =
\widetilde{\bm{M}}_2^{\top}\cdot \sqrt{n_0+n_1} \widetilde{Z}_{n_0+n_1} + o_p(1) \convergeto N(0,  \tilde{\sigma}_{ROC}^2(u)),
\eas
where 
\begin{equation}
    \tilde{\sigma}_{ROC}^2(u) = \widetilde{\bm{M}}_2^{\top}\widetilde{\Sigma}_{Z}\widetilde{\bm{M}}_2.
\end{equation}
This completes the proof of Theorem 4 in the main paper.

\section{Additional simulation results}
\label{add-simu-res}

In this section, we present additional simulation results, including parameter estimation and the case from the main paper under unbalanced sample sizes with $n_1 \ne n_0$.

\subsection{Simulation results on  parameter estimation}

Table~\ref{tab:esttheta-simu} reports the parameter estimates for all sample size combinations. As described in Section~2.3 of the main paper, our procedure first obtains $\hat{\xi}$ by fitting a correctly specified logistic model for $P^{(1)}(Y=1 \mid x)$, and then estimates $\theta$ conditional on $\hat{\xi}$; we refer to this as the \emph{Proposed} method. For comparison, we also consider an \emph{Ideal} method that estimates $\theta$ using full information on $(X,Y)$
from both source and target populations. This oracle estimator provides a benchmark in the simulation study but is infeasible in practice, since outcome labels are unavailable in the target population.

The results in Table~\ref{tab:esttheta-simu} demonstrate that the proposed estimator performs well in finite samples.
Across all parameters and sample size settings, the \emph{Proposed} method achieves low {RB} and MSE, which are competitive compared to the ideal benchmark. The coverage probabilities (CPs) of the bootstrap confidence intervals are close to the nominal level, and the average lengths (ALs) decrease substantially as the sample size increases.

However, some simulation results still show slight differences depending on the sample size.
For cases with sample size imbalance, particularly when $n_1 = 500$ and $n_0 = 2000$, the proposed method yields smaller MSEs compared to the balanced case $n_1 = n_0 = 500$, yet the corresponding CPs are slightly further from the nominal level. For instance, the CP for $\alpha_1$ is only 92.6\%.
This phenomenon aligns with intuition: while a larger sample size generally provides more information for estimation, it is challenging to effectively leverage a relatively small source sample to improve estimation accuracy in the target population, especially when the source data may be less representative or coarser.
In contrast, when $n_1 = 2000$ and $n_0 = 500$, the \emph{Proposed} method exhibits both lower MSEs and shorter ALs relative to the balanced setting, along with CPs that are closer to the nominal level. These findings confirm the validity and efficiency of our two-step estimation procedure, even in the presence of distributional shift and partial label information for unbalanced sample sizes.

\begin{table}[!htt]
  \centering
  \setlength{\tabcolsep}{4pt} 
  \caption{Estimation performance of the Proposed and Ideal methods. RB: relative bias ($\%$); MSE: mean square error ($\times 1000$).}
  \begin{tabular}{clcccccccc}
  	\toprule
  	&       & {RB} & {MSE} & {CP} & {AL} & {RB} & {MSE} & {CP} & {AL} \\
  	\midrule
  	&       & \multicolumn{4}{c}{$n_1 = 500,n_0 = 500$} & \multicolumn{4}{c}{$n_1 = 2000,n_0 = 2000$} \\
  	\cmidrule(lr){3-6} \cmidrule(lr){7-10}
  	    \multirow{2}[0]{*}{$\alpha_0$} & Proposed & -0.375  & 0.024  & 95.60\% & 0.632  & 0.066  & 0.006  & 96.20\% & 0.307  \\
          & Ideal & 0.394  & 0.020  & 95.60\% & 0.560  & 0.283  & 0.005  & 94.40\% & 0.275  \\
    \multirow{2}[0]{*}{$\beta_0$} & Proposed & -0.650  & 0.109  & 96.20\% & 1.564  & -0.384  & 0.028  & 95.40\% & 0.678  \\
          & Ideal & 0.405  & 0.054  & 94.20\% & 0.877  & 0.099  & 0.011  & 94.40\% & 0.432  \\
    \multirow{2}[0]{*}{$\alpha_1$} & Proposed & -0.701  & 0.123  & 95.20\% & 1.873  & -0.548  & 0.027  & 96.00\% & 0.683  \\
          & Ideal & 0.521  & 0.022  & 94.20\% & 0.578  & 0.199  & 0.005  & 94.60\% & 0.283  \\
    \multirow{2}[0]{*}{$\beta_1$} & Proposed & -1.069  & 0.165  & 95.20\% & 2.128  & -0.574  & 0.041  & 95.00\% & 0.808  \\
          & Ideal & 0.300  & 0.058  & 93.00\% & 0.901  & 0.125  & 0.014  & 92.00\% & 0.445  \\
     \midrule 
          &       & \multicolumn{4}{c}{$n_1 = 500,n_0 = 2000$} & \multicolumn{4}{c}{$n_1 = 2000,n_0 = 500$} \\
          \cmidrule(lr){3-6} \cmidrule(lr){7-10}
   \multirow{2}[0]{*}{$\alpha_0$} & Proposed & -0.024  & 0.021  & 95.60\% & 0.598  & -0.085  & 0.008  & 94.80\% & 0.370  \\
          & Ideal & 0.505  & 0.019  & 95.00\% & 0.541  & 0.144  & 0.006  & 93.60\% & 0.315  \\
    \multirow{2}[0]{*}{$\beta_0$} & Proposed & -1.087  & 0.071  & 93.80\% & 1.062  & 0.002  & 0.064  & 95.60\% & 1.137  \\
          & Ideal & 0.179  & 0.034  & 93.40\% & 0.696  & 0.273  & 0.028  & 94.80\% & 0.686  \\
    \multirow{2}[0]{*}{$\alpha_1$} & Proposed & -2.802  & 0.056  & 92.60\% & 1.033  & 1.919  & 0.087  & 95.80\% & 1.313  \\
          & Ideal & -0.048  & 0.008  & 93.00\% & 0.338  & 0.710  & 0.020  & 93.80\% & 0.545  \\
    \multirow{2}[0]{*}{$\beta_1$} & Proposed & -2.192  & 0.092  & 94.40\% & 1.286  & 0.782  & 0.111  & 95.80\% & 1.457  \\
          & Ideal & -0.037  & 0.033  & 94.00\% & 0.702  & 0.463  & 0.036  & 93.20\% & 0.717  \\
          \bottomrule
  \end{tabular}%
  \label{tab:esttheta-simu}%
\end{table}%

\subsection{Additional results for  prediction performance}
{The results in Table~\ref{tab:summary-simu} are consistent with those reported in Table 1 of the main paper. While different sample sizes lead to some numerical variation, the relative performance of the methods remains largely unchanged. Accordingly, we omit a detailed discussion of Table~\ref{tab:summary-simu} here. }

\begin{table}[!htt]
  \centering
  \setlength{\tabcolsep}{2pt} 
  \caption{Simulation results for classification performance of four methods under unbalanced sample sizes ($n_1 \neq n_0$). RB: relative bias ($\%$); MSE: mean square error ($\times 1000$).
 }
   		\begin{tabular}{lccccccccccccccc}
   		\toprule 
   		&  {mean} &  {RB} &  {MSE} &  {mean} &  {RB} &  {MSE} &  {mean} &  {bias} &  {MSE}  \\
   		\midrule 
   		 & \multicolumn{9}{c}{$n_1 = 500, n_0 = 2000$} \\
         \cmidrule(lr){2-10}  
          & \multicolumn{3}{c}{Recall} & \multicolumn{3}{c}{Accuracy} & \multicolumn{3}{c}{Precision} \\
          \cmidrule(lr){2-4} \cmidrule(lr){5-7} \cmidrule(lr){8-10}
   Proposed & 0.786 & -0.708 & 1.302 & 0.844 & -0.727 & 0.150  & 0.818 & -1.018 & 0.786 \\
    Reweight & 0.786 & -0.731 & 1.496 & 0.843 & -0.882 & 0.175 & 0.816 & -1.307 & 0.934 \\
    Naive & 0.561 & -29.062 & 55.011 & 0.536 & -36.953 & 99.946 & 0.438 & -47.057 & 152.517 \\
    Oracle & 0.792 & 0.114 & 0.198 & 0.851 & 0.112 & 0.070  & 0.828 & 0.104 & 0.114 \\
    \midrule
          & \multicolumn{9}{c}{$n_1 =2000, n_0 = 500$} \\
          \cmidrule(lr){2-10}
          & \multicolumn{3}{c}{Recall} & \multicolumn{3}{c}{Accuracy} & \multicolumn{3}{c}{Precision} \\
          \cmidrule(lr){2-4} \cmidrule(lr){5-7} \cmidrule(lr){8-10}
     Proposed & 0.786 & -0.712 & 2.193 & 0.846 & -0.453 & 0.270  & 0.824 & -0.372 & 1.024 \\
    Reweight & 0.786 & -0.674 & 2.266 & 0.846 & -0.510 & 0.278 & 0.822 & -0.510 & 1.063 \\
    Naive & 0.560  & -29.186 & 55.28 & 0.535 & -37.021 & 99.839 & 0.437 & -47.095 & 152.727 \\
    Oracle & 0.795 & 0.481 & 0.700   & 0.852 & 0.264 & 0.231 & 0.829 & 0.237 & 0.386 \\
    \bottomrule
   	\end{tabular}%
  \label{tab:summary-simu-supp}%
\end{table}%

\subsection{Additional results for the target mean estimation}
{Table~\ref{tab:esty} presents the estimation results for $\mu = \Ex_0(Y)$ under unbalanced sample sizes. The overall patterns are consistent with those in Table~2 of the main paper. Among implementable methods, the \emph{Proposed} method generally achieves the lowest MSEs and shortest ALs while maintaining nominal coverage across both scenarios. An exception occurs when $n_1 = 2000$ and $n_0 = 500$, where the \emph{Reweight} method attains the smallest RB, as expected since it leverages the large source sample. Nevertheless, the \emph{Proposed} method consistently outperforms \emph{Reweight} in terms of MSE, CP, and AL, reaffirming its overall superiority. In contrast, the \emph{Naive} method performs poorly across all metrics due to ignoring distributional shift.}

\begin{table}[!htt]
  \centering
  \caption{Simulation results for estimating $\mu = \Ex_0(Y)$ (true value 0.4) under unbalanced sample sizes ($n_1 \neq n_0$). RB: relative bias ($\%$); MSE: mean square error ($\times 1000$).}
    \begin{tabular}{lcccccccc}
    \toprule
          &  {RB} &  {MSE} &  {CP} &  {AL} &  {RB} &  {MSE} &  {CP} &  {AL} \\
          \midrule
          & \multicolumn{4}{c}{$n_1=500,n_0=2000$} & \multicolumn{4}{c}{$n_1=2000,n_0=500$} \\
          \cmidrule(lr){2-5} \cmidrule(lr){6-9}
      IW    & 0.458  & 0.856  & 96.60\% & 0.121  & -0.123  & 1.271  & 95.40\% & 0.143  \\
    Proposed & 0.288  & 0.783  & 95.40\% & 0.112  & -0.177  & 1.263  & 95.40\% & 0.142  \\
    Reweight & 0.449  & 0.913  & 94.60\% & 0.119  & 0.018  & 1.302  & 95.40\% & 0.144  \\
    Naive & 27.228  & 12.318  & 0.20\% & 0.091  & 27.148  & 12.082  & 0.00\% & 0.068  \\
    Oracle & -0.132  & 0.123  & 93.80\% & 0.043  & 0.076  & 0.450  & 95.00\% & 0.085  \\
    \bottomrule
    \end{tabular}%
  \label{tab:esty}%
\end{table}%

\subsection{Additional results for ROC and AUC}


Table~\ref{tab:estedROC-simu} summarizes the ROC estimation results under unbalanced sample sizes for both the fixed classifier $c(x) = P^{(1)}(Y = 1 \mid x)$ and the estimated classifier $\widehat{c}(x) = \widehat{P}^{(0)}(Y = 1 \mid x)$ at thresholds $u = 0.1$ and $0.2$.

\begin{table}[!htt]
  \centering
  \setlength{\tabcolsep}{2pt} 
  \caption{Estimation of ROC for the fixed classifier $c(x) = p^{(1)}(Y=1|x)$ and the 
  	estimated classifier  $\widehat{c}(x) = \widehat{p}^{(0)}(Y=1|x)$ at thresholds $u = 0.1$ and $0.2$ under unbalanced sample sizes ($n_1 \neq n_0$). RB: relative bias ($\%$); MSE: mean square error ($\times 1000$). }

   \begin{tabular}{cclcccccccc}
   	\toprule
   	& Threshold	&       &  {RB} &  {MSE} &  {CP} &  {AL} &  {RB} &  {MSE} &  {CP} &  {AL} \\
   	\midrule
   	&  & & \multicolumn{4}{c}{$n_1=500,n_0=2000$} & \multicolumn{4}{c}{$n_1=2000,n_0=500$} \\
     \cmidrule(lr){4-7}  \cmidrule(lr){8-11}
    \multirow{8}[0]{*}{$c(x)$}   &\multirow{4}[0]{*}{0.1} 
         &Proposed & 6.634  & 2.355  & 92.00\% & 0.196  & 2.077  & 1.730  & 96.60\% & 0.172  \\
    &&Reweight & 8.922  & 2.684  & 90.60\% & 0.201  & 2.884  & 1.825  & 95.80\% & 0.175  \\
    &&Naive & 177.769  & 133.318  & 0.00\% & 0.139  & 176.373  & 130.556  & 0.00\% & 0.096  \\
    &&Oracle & 0.250  & 0.197  & 95.40\% & 0.057  & 1.450  & 0.819  & 95.00\% & 0.116  \\
          \cmidrule{2-11}
    &\multirow{4}[0]{*}{0.2}
     &Proposed & 5.456  & 3.437  & 92.40\% & 0.236  & 2.373  & 2.561  & 95.60\% & 0.210  \\
    &&Reweight & 6.964  & 3.762  & 91.00\% & 0.238  & 2.999  & 2.641  & 95.20\% & 0.212  \\
    &&Naive & 138.003  & 178.115  & 0.00\% & 0.117  & 137.419  & 176.096  & 0.00\% & 0.081  \\
    &&Oracle & 0.124  & 0.296  & 95.60\% & 0.069  & 1.207  & 1.174  & 96.00\% & 0.140  \\
          \midrule    
          &       &       & \multicolumn{4}{c}{$n_1=500,n_0=2000$} & \multicolumn{4}{c}{$n_1=2000,n_0=500$} \\
           \cmidrule(lr){4-7}  \cmidrule(lr){8-11}
    \multirow{8}[0]{*}{$\widehat{c}(x)$} &\multirow{4}[0]{*}{0.1} 
    &Proposed & -0.506  & 1.240  & 96.80\% & 0.147  & -0.042  & 2.041  & 96.00\% & 0.181  \\
    &&Reweight & 0.131  & 1.438  & 97.20\% & 0.157  & 0.113  & 2.046  & 97.20\% & 0.182  \\
    &&Naive & -25.207  & 39.391  & 0.00\% & 0.142  & -26.706  & 43.304  & 0.00\% & 0.095  \\
    &&Oracle & 0.171  & 0.318  & 94.40\% & 0.067  & 0.547  & 1.156  & 93.00\% & 0.133  \\
           \cmidrule{2-11}
   & \multirow{4}[0]{*}{0.2} 
   &Proposed & -0.330  & 0.481  & 96.60\% & 0.092  & -0.118  & 0.789  & 96.00\% & 0.113  \\
    &&Reweight & 0.059  & 0.586  & 97.20\% & 0.101  & -0.036  & 0.789  & 96.80\% & 0.114  \\
    &&Naive & -16.386  & 21.706  & 0.00\% & 0.116  & -17.455  & 23.979  & 0.00\% & 0.080  \\
    &&Oracle & 0.104  & 0.135  & 95.40\% & 0.044  & 0.345  & 0.482  & 93.60\% & 0.086  \\
   	\bottomrule
   \end{tabular}%
  \label{tab:estedROC-simu}%
\end{table}%

{For the fixed classifier $c(x)$ under the setting $n_1 = 500$ and $n_0 = 2000$, the \emph{Proposed} method consistently outperforms \emph{Reweight} across both thresholds (0.1 and 0.2), achieving smaller RB and MSE, as well as shorter AL. Its coverage probabilities (92.0\% and 92.4\%) are slightly below nominal but still closer to the target level than those of \emph{Reweight} (90.6\% and 91.0\%). This reflects the limitation of a relatively small source sample, which hinders effective transfer and affects uncertainty quantification in the larger target population. While both methods leverage source information, the heavier reliance of \emph{Reweight} exacerbates undercoverage, whereas the \emph{Proposed} method more effectively balances information from the source and target populations. In contrast, the \emph{Naive} method performs poorly, with very large RBs (177.769 and 138.003), high MSEs (133.318 and 178.115), and zero coverage. When $n_1 > n_0$, both the \emph{Proposed} and \emph{Reweight} methods show marked improvements in RB and CP.}


{For the estimated classifier $\widehat{c}(x)$ at thresholds 0.1 and 0.2, \emph{Reweight} attains the smallest RB among feasible methods when $n_1<n_0$, and also at the 0.2 threshold when $n_1>n_0$. 
Despite these RB advantages, \emph{Proposed} consistently delivers lower MSE and shorter AL across both sample-size configurations and both thresholds than \emph{Reweight}. In terms of coverage, \emph{Proposed} achieves 96.0\%–96.8\%, whereas \emph{Reweight} is slightly higher at 96.8\%–97.2\%; however, the tighter intervals and smaller MSE of \emph{Proposed} suggest that it achieves higher efficiency without sacrificing nominal coverage.
By contrast, the \emph{Naive} method exhibits large RB, high MSE, and near-zero coverage. 
The \emph{Oracle} benchmark, as expected, yields the smallest MSE and shortest intervals, though its coverage falls below nominal in the $n_1=2000, n_0=500$ case (e.g., 93.0\% and 93.6\%), highlighting that sample size imbalance can still affect inference even when target outcomes are fully observed. 
Unlike the fixed classifier $c(x)$, there is no clear pattern associated with the direction of sample size imbalance.
This is because the estimation of $\widehat{c}(x)$ introduces additional variability: the ROC now depends not only on $\hat{\theta}$ but also on the uncertainty in $\widehat{c}(x)$. In contrast, when $c(x)$ is fixed, the ROC depends only on $\hat{\theta}$, so the coverage remains nominal. Hence, the small undercoverage observed with $\widehat{c}(x)$ is expected and reflects the extra uncertainty from estimating the classifier.
}

{In addition to ROC estimation, we also report the corresponding AUC estimates in Table~\ref{tab:estedAUC-simu}. Given that the AUC represents the area under the ROC curve, the findings regarding AUC estimation are largely consistent with those for ROC. Hence, for brevity, we omit a detailed discussion of the ROC results. }

\begin{table}[!htt]
   \centering
  \setlength{\tabcolsep}{4pt} 
  \caption{Estimation of AUC for the fixed classifier $c(x) = p^{(1)}(Y=1|x)$ and 
  	estimated classifiers $\widehat{c}(x) = \widehat{p}^{(0)}(Y=1|x)$ in simulation setting when $n_1 \neq n_0$. RB: relative bias ($\%$); MSE: mean square error ($\times 1000$).}
   	\begin{tabular}{clcccccccc}
   	\toprule
   	&       &  {RB} &  {MSE} &  { CP} &  {AL} &  {RB} &  {MSE} &  { CP} &  {AL} \\
   	\midrule
   	&       & \multicolumn{4}{c}{$n_1=500,n_0=2000$} & \multicolumn{4}{c}{$n_1=2000,n_0=500$} \\
    \cmidrule(lr){3-6} \cmidrule(lr){7-10} 
    \multirow{4}[0]{*}{$c(x)$} 
    &Proposed & 2.167  & 2.263  & 90.80\% & 0.187  & 0.594  & 1.540  & 96.20\% & 0.160  \\
    &Reweight & 2.809  & 2.417  & 90.20\% & 0.186  & 0.789  & 1.578  & 95.20\% & 0.161  \\
    &Naive & 46.801  & 73.439  & 0.00\% & 0.059  & 46.705  & 73.001  & 0.00\% & 0.038  \\
    &Oracle & 0.011  & 0.169  & 94.20\% & 0.050  & 0.130  & 0.626  & 95.20\% & 0.101  \\
          \midrule
   	&       & \multicolumn{4}{c}{$n_1=500,n_0=2000$} & \multicolumn{4}{c}{$n_1=2000,n_0=500$} \\
    \cmidrule(lr){3-6} \cmidrule(lr){7-10} 
    \multirow{4}[0]{*}{$\widehat{c}(x)$} 
    &Proposed & -0.108  & 0.118  & 96.60\% & 0.046  & 0.008  & 0.192  & 96.20\% & 0.055  \\
    &Reweight & 0.108  & 0.156  & 96.80\% & 0.051  & 0.058  & 0.196  & 96.40\% & 0.056  \\
    &Naive & -7.358  & 4.821  & 0.60\% & 0.058  & -7.817  & 5.272  & 0.00\% & 0.038  \\
    &Oracle & 0.070  & 0.038  & 95.60\% & 0.023  & 0.234  & 0.140  & 90.40\% & 0.045  \\
   	\bottomrule
   \end{tabular}%
  \label{tab:estedAUC-simu}%
\end{table}%

Overall, the results under unbalanced sample sizes demonstrate that the \emph{Proposed} method remains robust, exhibiting only mild sensitivity to sample size imbalance. Notably, even the \emph{Oracle} procedure experiences noticeable coverage deterioration when the source and target sample sizes differ substantially. These findings, together with the main results, further support the reliability of the proposed approach for transfer learning under distributional shift.

\section{Additional results for waterbirds dataset}
\label{add-real-res}
To conserve space in the main text, we report the results for estimating the model parameter $\theta$ under the waterbirds dataset here. For this semi-synthetic dataset, true labels are available for all observations, including those in the target sample. Thus, analogous to the simulation study, we present both \emph{Proposed} and \emph{Ideal} estimates (Ests) of the parameter $\theta$. Additionally, we provide bootstrap-based 95\% confidence intervals (CIs) and their interval lengths using 500 bootstrap samples. The detailed results are summarized in Table~\ref{tab:esttheta-water}.

\begin{table}[!htt]
	\centering
	\caption{Point estimates and confidence intervals for $\theta$ using the waterbirds dataset.}
	\begin{tabular}{clccc}
		\toprule 
		&       &  {Est} &  {CI} &  {CI length} \\
		\cmidrule{3-5}
		\multirow{2}[0]{*}{ $\alpha_0$} & Proposed & -0.721  &\([-0.750,-0.573]\) & 0.177  \\
		& Ideal & -0.628  &\([-0.663,-0.592]\) & 0.071  \\
		\multirow{2}[0]{*}{ $\beta_0$} & Proposed & 2.951  &\([\phantom{-}2.788,\phantom{-}3.117]\) & 0.328  \\
		& Ideal & 2.945  &\([\phantom{-}2.803,\phantom{-}3.101]\) & 0.298  \\
		\midrule
		\multirow{2}[0]{*}{ $\alpha_1$} & Proposed & 2.290  &\([\phantom{-}1.999,\phantom{-}2.563]\) & 0.564  \\
		& Ideal & 2.250  &\([\phantom{-}2.019,\phantom{-}2.559]\) & 0.540  \\
		\multirow{2}[0]{*}{$\beta_1$} & Proposed & -3.118  &\([-3.341,-2.504]\) & 0.837  \\
		& Ideal & -2.938  &\([-3.253,-2.682]\) & 0.571  \\
		\bottomrule
	\end{tabular}%
	\label{tab:esttheta-water}%
\end{table}%

Notably, the proposed estimates closely align with the ideal estimates, both in terms of point estimation and confidence interval coverage. This result highlights the effectiveness and reliability of the proposed estimation approach, thereby ensuring the validity and robustness of subsequent estimation and inference for functionals of interest based on $\hat{\theta}$.

\end{document}